\documentclass[sigconf]{acmart}

\makeatletter
\def\@ACM@checkaffil{
	\if@ACM@instpresent\else
	\ClassWarningNoLine{\@classname}{No institution present for an affiliation}%
	\fi
	\if@ACM@citypresent\else
	\ClassWarningNoLine{\@classname}{No city present for an affiliation}%
	\fi
	\if@ACM@countrypresent\else
	\ClassWarningNoLine{\@classname}{No country present for an affiliation}%
	\fi
}
\makeatother

\renewcommand\footnotetextcopyrightpermission[1]{} 

\usepackage{stmaryrd}
\usepackage{balance}
\usepackage{algorithm,tabularx}
\usepackage{algpseudocode}
\usepackage{amsmath}
\usepackage{amsfonts}
\usepackage[english]{babel}
\usepackage{amsthm}
\usepackage{tablefootnote}

\usepackage{amssymb}
\usepackage{multirow}
\usepackage{hyperref}
\usepackage{graphicx}
\usepackage{subfigure}
\usepackage{balance}
\usepackage{pifont}
\usepackage{xcolor}

\usepackage{mathabx}
\usepackage{makecell}

\makeatletter
\newcommand{\multiline}[1]{%
	\begin{tabularx}{\dimexpr\linewidth-\ALG@thistlm}[t]{@{}X@{}}
		#1
	\end{tabularx}
}
\makeatother

\algrenewcommand\textproc{}

\newcommand{\main}{\textsf{Camel}}
\newcommand{\baseline}{{\main}-\textsf{vec}}
\UseRawInputEncoding

\newcommand{\revise}{\textcolor{black}}

\theoremstyle{plain}
\newtheorem{thm}{Theorem}
\newtheorem{lem}[thm]{Lemma}

\theoremstyle{definition}
\newtheorem{defn}{Definition}


\begin{document}
        \pagenumbering{gobble}
        \copyrightyear{2024}
        \acmYear{2024}
        \setcopyright{acmlicensed}\acmConference[CCS '24]{Proceedings of the 2024 ACM SIGSAC
        Conference on Computer and Communications Security}{October 14--18, 2024}{Salt Lake City, UT,
        USA}
        \acmBooktitle{Proceedings of the 2024 ACM SIGSAC Conference on Computer and Communications
        Security (CCS '24), October 14--18, 2024, Salt Lake City, UT, USA}
        \acmDOI{10.1145/3658644.3690200}
        \acmISBN{979-8-4007-0636-3/24/10} 
	\title{{\main}: Communication-Efficient and Maliciously Secure Federated Learning in the Shuffle Model of Differential Privacy}
	
        \titlenote{This is a full version of the paper originally published in ACM CCS 2024 \cite{camel-submission}.}

	\author{Shuangqing Xu}
	\affiliation{
		\institution{Harbin Institute of Technology, Shenzhen}
            \city{Shenzhen}
            \country{China}}

	\author{Yifeng Zheng}
	\authornote{Corresponding authors.}
	\affiliation{
		\institution{Harbin Institute of Technology, Shenzhen}
            \city{Shenzhen}
            \country{China}}

	\author{Zhongyun Hua}
        \authornotemark[2]
	\affiliation{
		\institution{Harbin Institute of Technology, Shenzhen}
            \city{Shenzhen}
            \country{China}}

	\renewcommand{\shortauthors}{Shuangqing Xu, Yifeng Zheng, \& Zhongyun Hua}
	
	\begin{abstract}
		
		Federated learning (FL) has rapidly become a compelling paradigm that enables multiple clients to jointly train a model by sharing only gradient updates for aggregation, without revealing their local private data. In order to protect the gradient updates which could also be privacy-sensitive, there has been a line of work studying local differential privacy (LDP) mechanisms to provide a formal privacy guarantee. With LDP mechanisms, clients locally perturb their gradient updates before sharing them out for aggregation. However, such approaches are known for greatly degrading the model utility, due to heavy noise addition. To enable a better privacy-utility trade-off, a recently emerging trend is to apply the \textit{shuffle model} of DP in FL, which relies on an intermediate shuffling operation on the perturbed gradient updates to achieve privacy amplification. Following this trend, in this paper, we present {\main}, a new communication-efficient and maliciously secure FL framework in the shuffle model of DP. {\main} first departs from existing works by ambitiously supporting integrity check for the shuffle computation, achieving security against malicious adversary. Specifically, {\main} builds on the trending cryptographic primitive of secret-shared shuffle, with custom techniques we develop for optimizing system-wide communication efficiency, and for lightweight integrity checks to harden the security of server-side computation. In addition, we also derive a significantly tighter bound on the privacy loss through analyzing the R\'enyi differential privacy (RDP) of the overall FL process. Extensive experiments demonstrate that {\main} achieves better privacy-utility trade-offs than the state-of-the-art work, with promising performance. 
	\end{abstract}
	
	
	\settopmatter{printfolios=true}
	
	\maketitle

	\section{Introduction}
	Federated learning (FL) has recently emerged as an appealing paradigm \cite{fedavg} that allows clients to jointly train a model by sharing gradient updates instead of their local data. 
	However, recent works have shown that the shared gradient updates can still leak private information about clients' training datasets \cite{ZhuLH19}. 
	To mitigate this issue, differential privacy (DP) \cite{DworkR14} has been widely incorporated within FL to provide a formal privacy guarantee. 
	
	Initially, DP is studied in the centralized context where a \textit{trusted} server centrally collects raw training data from clients \cite{iyengar2019towards, DPSGD, wang2019differentially}. 
	In contrast, the notion of local differential privacy (LDP) \cite{KasiviswanathanLNRS11} is more appropriate for distributed learning \cite{Truex0CGW20,ChamikaraLCNGBK22}. 
	In the LDP setting, each client locally adds sufficient noise to its gradient update, and then sends the noisy gradient update to the \textit{untrusted} server. 
	Although LDP mechanisms offer more appealing privacy properties, they usually come with a significant sacrifice in model utility compared to centralized DP mechanisms \cite{kairouz21}. 
	
	Recently, there has emerged a new trend of applying the \emph{shuffle model} \cite{BalleBGN19,ErlingssonFMRTT19} of DP in distributed learning \cite{liu2021flame,Erlingsson20,AISTATS21,NeurIPS21} to enable a significantly better privacy-utility trade-off. 
	In the shuffle model of DP, each client locally perturbs its message, and then sends it to a shuffler, which sits between the clients and the server. 
	The shuffler, which is assumed not to collude with the server, randomly shuffles the noisy messages from clients and then forwards them to the server.
	%
	%
	In this way, the server in the shuffle model of DP cannot associate messages with clients.  
	Such anonymity amplifies privacy in that less local noise is required for achieving the same privacy guarantee as LDP, thereby enabling a better privacy-utility trade-off. 
	While previous works \cite{liu2021flame,Erlingsson20,AISTATS21,NeurIPS21} offer potential solutions for applying the shuffle model of DP in FL, they suffer from three key limitations as described below.

	Firstly, all these works assume a shuffler that honestly shuffles the perturbed gradients from clients and forwards them to the server for model update. 
	They do not support verifiability for the shuffle computation and thus would fail to provide integrity (as well as privacy) guarantees in the presence of a malicious shuffler that tampers with the shuffle computation.
	%
	%
	%
	Secondly, for accounting for the overall privacy loss of repeated interactions in the training process, most existing works only characterize the approximate DP of each training iteration and utilize the advanced composition theorem \cite{DworkR14} to quantify the privacy leakage of multiple iterations, which is known to be loose compared to the analytical results using R\'enyi differential privacy (RDP) \cite{RDP}. 
	Thirdly, most existing works only rely on the privacy amplification effect from shuffling, overlooking the integration with another strategy---\textit{subsampling}---that can also benefit privacy amplification \cite{BalleBG18}.

	In light of the above, we present {\main}, a new communication-efficient and maliciously secure FL framework in the shuffle model of DP. 
	To achieve strong privacy amplification, {\main} leverages the synergy of the strategies of both shuffling and subsampling.
	Regarding achieving privacy amplification by shuffling, {\main} is designed to securely and efficiently realize the shuffle of (noisy) gradients from the clients, with robustness against malicious adversary. 

	Our starting point is to leverage a state-of-the-art secret-shared shuffle protocol \cite{NDSS22} in the three-server honest-majority setting, which outperforms conventional mixnet-based verifiable shuffle approaches \cite{mixnet,HooffLZZ15}. 
	However, simply utilizing this protocol to shuffle (noisy) gradients in {\main} still does not promise high efficiency, due to the communication bottleneck resulting from directly secret-sharing and (securely) shuffling the large-sized gradient vectors in FL. 
	%
	To tackle this efficiency challenge, our key insight is to delicately build on the widely used LDP mechanism in \cite{duchi2018minimax} for local gradient perturbation and losslessly compress a noisy high-dimensional gradient vector into a random seed and a sign bit. 
	%
	%
	%
	With our proposed mechanism for producing compressed noisy gradients at clients, {\main} can have the secret-shared shuffle efficiently applied for small-sized random seeds and sign bits, rather than for high-dimensional noisy gradient vectors.
	This notably reduces both client-server and inter-server communication costs.

	We then consider how to achieve security against a malicious server in {\main}.
	It is noted that the protocol in \cite{NDSS22} involves a verification mechanism to detect the misbehavior of a malicious server during the shuffle computation. 
	However, we observe that it is vulnerable to the online selective failure attack as shown recently in \cite{NDSS24}.
	This prevents us from directly using the verification mechanism in \cite{NDSS22} to achieve malicious security for the shuffle computation.
	Additionally, it is worth noting that in {\main} the servers not only need to perform the (secure) shuffle, but also to perform several post-shuffle operations, including sampling, decompression, and aggregation of (noisy) gradients. 
	The integrity of these operations should also be ensured.
	Therefore, to achieve malicious security for the overall server-side computation in {\main}, we devise a series of lightweight integrity checks for the secret-shared shuffle computation as well as the post-shuffle operations.
	
	With the strategies of shuffling and subsampling adequately instantiated in {\main}, we compose privacy amplification by shuffling theorem \cite{BalleBGN19} and privacy amplification by subsampling theorem \cite{BalleBG18} to analyze the coupled privacy amplification and then obtain a tighter bound on the privacy loss on each FL iteration. 
	By tightly composing all FL iterations using RDP and analyzing the RDP of the overall FL process, we derive a tighter bound on the total privacy loss compared to existing approximate DP bounds \cite{AISTATS21,Erlingsson20}. 
	
	We implement the protocols of {\main} and empirically evaluate {\main}'s utility and efficiency on two widely-used real-world datasets (MNIST \cite{LeCunBBH98} and FMNIST \cite{FMNIST}). 
	The results demonstrate the significant performance advantage of our proposed noisy gradient compression mechanism over the baseline without compression.
	For example, for a single FL iteration on the FMNIST dataset, {\main} reduces the system-wide communication cost (including client-server and inter-server communication) by 20,029$\times$ and achieves an improvement of 1,607$\times$ in server-side overall runtime.
	Compared to the state-of-the-art \cite{AISTATS21}, {\main} achieves better privacy-utility trade-offs. 
	For example, as tested over the MNIST dataset, {\main} achieves an accuracy of 84.83\%, while the work \cite{AISTATS21} attains 78.69\% accuracy under the same budget of $\varepsilon=10$.
	
	We highlight our main contributions below: 
	\vspace{-2pt}
	\begin{itemize}
		
		\item We present {\main}, a new communication-efficient and maliciously secure FL framework in the shuffle model of DP, which delicately bridges the advancements in secure multi-party computation (MPC) and shuffle model of DP.
		
		\item We leverage the trending cryptographic primitive of secret-shared shuffle and introduce techniques for compressing gradients perturbed under LDP to optimize system-wide communication efficiency, and for lightweight integrity checks to harden the security of server-side computation.


		%
		
		\item We derive a significantly tighter bound on the privacy loss by analyzing the RDP of the overall FL process compared to existing approximate DP bounds. 
		
		\item We formally analyze the privacy, communication, convergence, and security of {\main}. 
		We implement and empirically evaluate {\main}'s utility and efficiency on two widely-used real-world datasets. 
		The results demonstrate that {\main} achieves better privacy-utility trade-offs than the state-of-the-art work, with promising performance. 
		%
		
	\end{itemize}

	The rest of this paper is organized as follows. 
	Section \ref{sec:related-work} discusses the related work. 
	Section \ref{sec:preliminaries} introduces some preliminaries. 
	Section \ref{sec:prob-statement} gives the problem statement. 
	Section \ref{sec:protocol} presents the detailed design of {\main}. 
	Section \ref{sec:theoretical-results} provides the privacy (analytical bound), communication, and convergence analysis, followed by the experimental evaluation in Section \ref{sec:exp}. 
	Section \ref{sec:discussion} discusses other new concepts and possible extensions. 
	Section \ref{sec:conclusion} concludes the whole paper.
	
	
	\section{Related Work}
	\label{sec:related-work}
	As a rigorous measure of information disclosure, differential privacy (DP) has been studied extensively for private learning in the centralized setting \cite{iyengar2019towards, DPSGD, wang2019differentially}. 
	For example, Abadi et al. propose DP-SGD \cite{DPSGD} to train models via differentially private stochastic gradient descent (SGD) with provably limited information leakage. 
	Abadi et al. \cite{DPSGD} also propose and use a stronger accounting method called \emph{moments accountant} to obtain much tighter estimates of the privacy loss. 
	%
	%
	However, this line of work assumes that the raw training data of clients is collected by a \textit{trusted} server, a condition that is challenging to satisfy in real-world applications given the growing awareness of data privacy and increasingly strict data regulations.
	
	On the other hand, there has been growing interests in applying LDP mechanisms in FL \cite{Truex0CGW20, ChamikaraLCNGBK22, miao2022compressed}. 
	In the LDP framework, the gradients are locally perturbed by the clients with sufficient noise before they are collected by an \textit{untrusted} server, but it is known for yielding low model utility \cite{kairouz21}.   
	For example, the work of Truex et al. \cite{Truex0CGW20} requires an overly large LDP-level privacy budget $\varepsilon_0=\alpha\cdot2c\cdot10^\rho$ with parameters  $\alpha=1,c=1,\rho=10$ to achieve satisfactory utility performance \cite{SunQC21}. 

	Different from these works that solely rely on LDP mechanisms to protect individual gradients, {\main} is designed to work under the recently emerging shuffle model of DP \cite{BalleBGN19,ErlingssonFMRTT19}.
	In contrast to the LDP setting, the shuffle model of DP introduces a shuffler sitting between clients and the aggregation server to shuffle the perturbed messages from clients, achieving the privacy amplification effect.  
	It has recently garnered substantial attention and been increasingly adopted in distributed learning \cite{liu2021flame,Erlingsson20,AISTATS21,NeurIPS21}, due to its significantly better privacy-utility trade-off over LDP.
	%
	In \cite{Erlingsson20}, Erlingsson et al. use the privacy amplification by shuffling theorem from \cite{BalleBGN19} to amplify privacy per training iteration. 
	They also apply the advanced composition theorem \cite{DworkR14} to analyze the approximate DP of the overall training process.
	However, the advanced composition theorem is known to be loose for composition \cite{RDP,DPSGD}. 
	Besides, although Erlingsson et al. \cite{Erlingsson20} propose to reduce communication cost by compressing the perturbed gradients at clients, they consider a setting where each client only has one data point, which is rarely seen in practice \cite{AISTATS21}. 
	Compared to \cite{Erlingsson20}, the works \cite{liu2021flame,AISTATS21} consider a more practical setting where each client holds multiple data points. 
	Among them, the work of Liu et al. \cite{liu2021flame} lets clients locally perturb the gradients and then directly forward the full-precision gradients to the shuffler, which is not communication-efficient. 
	Besides, the work of \cite{liu2021flame} also naively utilizes the advanced composition theorem to characterize the approximate DP of the training process. 
	
	\revise{
		The state-of-the-art work that is most related to ours is \cite{AISTATS21}, which proposes a communication-efficient noisy gradient compression mechanism by firstly perturbing a gradient using Duchi et al.'s LDP mechanism \cite{duchi2018minimax} and then compressing the perturbed gradient using the non-private compression mechanism from \cite{mayekar2020limits}. 
		However, the compression method used in \cite{AISTATS21} is not lossless and introduces errors after compression. 
		Besides, although the work of \cite{AISTATS21} composes privacy amplification by shuffling with privacy amplification by subsampling to amplify the privacy at each FL iteration, it simply characterizes the approximate DP of the proposed FL process. 
		The follow-up work in \cite{NeurIPS21} analyzes the RDP of the whole training process and derives a significantly tighter bound compared to approximate DP (using advanced composition theorem). 
		However, the work in \cite{NeurIPS21} is limited to and analyzed for a specialized scenario where each client only has one data point, which allows for straightforward uniform sampling of gradients and the direct use of existing subsampling amplification results. 
		In contrast, our work considers practically each client holding multiple data points, requiring a more sophisticated sampling process and a different privacy analysis, i.e., we cannot directly use existing amplification results as in the work \cite{NeurIPS21}. 
		We also notice that the recent work of \cite{CCS21} analyzes the RDP of the shuffle model. However, it only considers privacy amplification by shuffling in deriving the RDP bound, without incorporating privacy amplification by subsampling. Besides, similar to \cite{NeurIPS21}, the work of \cite{CCS21} considers each client only holding one data point. 
	}
	
	In addition, we note that all existing works on applying the shuffle model of DP in distributed learning assume a shuffler that honestly executes the shuffle. 
	So they would fail to provide integrity and privacy guarantees in case that the shuffling computation is not correctly conducted. 
	%
	{\main} largely departs from existing works \cite{liu2021flame,Erlingsson20,AISTATS21,NeurIPS21} in that it 
	(1) eliminates the reliance on an honest shuffler and provides malicious security leveraging advancements in MPC, where the integrity of server-side computation can be efficiently checked, 
	(2) is communication-efficient and supports the more practical FL setting that each client holds multiple data points, and 
	(3) composes privacy amplification by shuffling with privacy amplification by subsampling to achieve a stronger amplification effect at each FL iteration, and analyzes the RDP of the overall FL process to derive a tighter bound through RDP composition. 
	
	\revise{
		We note that there is an orthogonal line of work \cite{kairouz21,AgarwalKL21} that combines the secure aggregation technique \cite{BonawitzIKMMPRS17} and DP to train differentially private models in FL.
		%
		This line of work relies on the use of secure aggregation to support summation of individual perturbed gradient updates of clients, enabling the server to only learn the aggregated noisy gradient updates.
		In particular, this kind of approach allows small local noise to be added at a volume insufficient for a meaningful LDP guarantee. 
		However, when aggregated, the noise is sufficient to ensure a meaningful DP guarantee.
		As the secure aggregation technique does not allow extra operations on the masked noisy individual gradient updates (i.e., they can only be simply added up to produce an aggregated result), it hinders system-wide communication efficiency optimization for FL (e.g., through gradient compression as in {\main}, which requires decompression before aggregation).
		This poses a barrier to simultaneously balancing privacy, utility, and efficiency for FL. 
		Different from this orthogonal line of work, {\main} follows the emerging shuffle model of DP, and shows how communication efficiency can be substantially optimized via customized gradient compression techniques for FL in this new setting. 
		We demonstrate in Section \ref{sec:exp:comparison-secagg} the prominent advantage of {\main} in communication efficiency compared to the approach combining secure aggregation and DP.
		In addition, we note that the trending shuffle model of DP under which {\main} operates can also provide support for flexibly enforcing custom aggregation rules to cater for different needs, as compared to the approach combining secure aggregation and DP (which does not allow computation before aggregation).
		For example, Byzantine-robust aggregation rules \cite{LiuCLWW023} may require custom computation on the individual gradient updates so as to combat adversarial attacks on the training process. 
		And this is hard to be implemented with FL adopting the approach combining secure aggregation and DP.
		%
	}
	
	\section{Preliminaries}
	\label{sec:preliminaries}
	\subsection{Notations}
	We denote by $[n]$ the set $\{1,\cdots,n\}$ for $n\in\mathbb{N}$. 
	${s} \stackrel{\$}{\leftarrow} \mathbb{Z}_p$ denotes that $s$ is uniformly randomly sampled from a finite field $\mathbb{Z}_p$. 
	$||$ denotes string concatenation, and for a string $s$, we use $s[i,j]$ to represent the substring of $s$ spanning from the $i$-th bit to the $j$-th bit. 
	We use boldface letters such as $\boldsymbol{v}$ to represent vectors. 
	$\|\boldsymbol{v}\|_2$ represents the $\ell_2$-norm of $\boldsymbol{v}$. 
	For a vector $\boldsymbol{v} = (v_1,\cdots,v_N)$ and a permutation $\pi:\mathbb{Z}_N\rightarrow\mathbb{Z}_N$, we denote $\pi(\boldsymbol{v})$ as the permuted vector $(v_{\pi(1)},\cdots,v_{\pi(N)})$.

	\subsection{Differential Privacy}
	This section gives essential definitions and properties related to differential privacy (DP). 
	In central differential privacy (CDP), a trusted server collects users' raw data and applies a private mechanism. 
	Define two datasets $\mathcal{D}=\left\{d_1, \ldots, d_n\right\}$ and $\mathcal{D}'=\left\{d_1^{\prime}, \ldots, d_n^{\prime}\right\}$ (each comprises $n$ data points from $\mathcal{X}$) as neighboring datasets if they differ in one data point, i.e., there exists an $i\in[n]$ such that $d_i \neq d_i^{\prime}$ and for every $j\in[n],j\neq i$, we have $d_j = d_j^{\prime}$. Then, CDP is defined as follows. 
	
	\begin{defn}
		\textit{\textbf{(Central Differential Privacy - ($\varepsilon,\delta$)-DP \cite{DworkR14}).}}
		\emph{
			A randomized mechanism $\mathcal{M}$: $\mathcal{X}^n\rightarrow \mathcal{Y}$ satisfies $(\varepsilon,\delta)$-DP if for any two neighboring datasets $\mathcal{D},\mathcal{D}^{\prime}\in\mathcal{X}^n$ and for any subset of outputs $\mathcal{S}\subseteq\mathcal{Y}$ it holds that
			$$
			\operatorname{Pr}[\mathcal{M}(\mathcal{D}) \in \mathcal{S}] \leq e^{\varepsilon} \operatorname{Pr}\left[\mathcal{M}\left(\mathcal{D}^{\prime}\right) \in \mathcal{S}\right]+\delta.
			$$
		}
	\end{defn}
	
	In comparison, local differential privacy (LDP) does not rely on a trusted server, as raw data is locally perturbed before collection. 
	The formal definition of LDP with privacy level $\varepsilon_0$ follows. 
	
	\begin{defn}
		\textit{\textbf{(Local Differential Privacy - $\varepsilon_0$-LDP \cite{KasiviswanathanLNRS11}).}}
		\emph{
			A mechanism $\mathcal{R}: \mathcal{X} \rightarrow \mathcal{Y}$ satisfies $\varepsilon_0$-LDP if for any two inputs $d, d^{\prime} \in \mathcal{X}$ and any subset of outputs $\mathcal{S}\subseteq\mathcal{Y}$, we have
			$$
			\operatorname{Pr}[\mathcal{R}(d) \in \mathcal{S}] \leq e^{\varepsilon_0} \operatorname{Pr}\left[\mathcal{R}\left(d^{\prime}\right) \in \mathcal{S}\right].
			$$
		}
	\end{defn}

	To tightly track the privacy loss when composing multiple private mechanisms, we also introduce the notion of R\'enyi differential privacy (RDP) \cite{RDP} as a generalization of differential privacy. 
	
	\begin{defn}
		\textit{\textbf{(R\'enyi Differential Privacy - ($\lambda,\varepsilon(\lambda)$)-RDP {\cite{RDP,CCS21}}).}}
		\emph{
			A randomized mechanism $\mathcal{M}$: $\mathcal{X}^n\rightarrow \mathcal{Y}$ satisfies $(\lambda,\varepsilon(\lambda))$-RDP if for any two neighboring datasets $\mathcal{D},\mathcal{D}^{\prime}\in\mathcal{X}^n$, we have 
			$$
			D_\lambda\left(\mathcal{M}(\mathcal{D})\|\mathcal{M}\left(\mathcal{D}^{\prime}\right)\right) \leq \varepsilon(\lambda),
			$$ 
			where $D_\lambda(P\| Q)$ is the $\lambda$-R\'enyi divergence between two probability distributions P and Q and is given by
			$$
			D_\lambda(P\| Q) \triangleq \frac{1}{\lambda-1} \log \left(\mathbb{E}_{x \sim Q}\left[\left(\frac{P(x)}{Q(x)}\right)^\lambda\right]\right).
			$$
		}
		
	\end{defn}
	
	The main advantage of RDP compared to other DP notions lies in its composition property, which is highlighted as follows. 
	
	\begin{lem}\textit{\textbf{(Adaptive Composition of RDP \cite{RDP}).}}
		\label{lem:sequantial}
		Let $\mathcal{M}_1: \mathcal{D}\rightarrow\mathcal{R}_1$ be a mechanism satisfying ($\lambda,\varepsilon_1(\lambda)$)-RDP and $\mathcal{M}_2: \mathcal{D}\times\mathcal{R}_1\rightarrow\mathcal{R}_2$ be a mechanism satisfying ($\lambda,\varepsilon_2(\lambda)$)-RDP. Define their combination $\mathcal{M}_{1,2}:\mathcal{D}\rightarrow\mathcal{R}_2$ by $\mathcal{M}_{1,2}(\mathcal{D})=\mathcal{M}_2(\mathcal{D},\mathcal{M}_1(\mathcal{D}))$. Then $\mathcal{M}_{1,2}$ satisfies ($\lambda,\varepsilon_1(\lambda)+\varepsilon_2(\lambda)$)-RDP.
	\end{lem}
	
	Although our primary goal is to analyze the RDP of FL in the shuffle model, we also care about the more meaningful notion of ($\varepsilon,\delta$)-DP. 
	To convert ($\lambda,\varepsilon(\lambda)$)-RDP to ($\varepsilon,\delta$)-DP, we can use the state-of-the-art conversion lemma as follows. 
	
	\begin{lem}\textit{\textbf{(From RDP to DP \cite{BalleBGHS20,Canonne0S20}).}}
		\label{lem:rdp_to_dp}
		If a randomized mechanism $\mathcal{M}$ is ($\lambda,\varepsilon(\lambda)$)-RDP, then the mechanism is also ($\varepsilon,\delta$)-DP, where $\varepsilon$ is defined as below for a given $\delta\in(0,1)$:
		$$
		\varepsilon=\min_\lambda \left( \varepsilon(\lambda)+\frac{\log (1 / \delta)+(\lambda-1) \log (1-1 / \lambda)-\log (\lambda)}{\lambda-1}\right).
		$$
	\end{lem}

	\subsection{Additive Secret Sharing}
	Given a private value $x\in\mathbb{Z}_{p}$, the 2-out-of-2 additive secret sharing (ASS) \cite{mohassel2017secureml} splits it into two secret shares $\langle x \rangle_1$ and $\langle x \rangle_2\in\mathbb{Z}_{p}$ such that $x = \langle x \rangle_1+\langle x \rangle_2\: \bmod p$. 
	The secret shares are held by two parties $\mathcal{P}_1$ and $\mathcal{P}_2$, respectively. 
	Such a sharing of $x$ is denoted as $\llbracket x \rrbracket$.
	Note that secure computation with such secret shares works in $\mathbb{Z}_{p}$. For ease of presentation, we will omit the modulo operation in the subsequent description of secret sharing-based operations.  
	
	The basic operations related to additive secret sharing are as follows. 
	%
	%
	(1) \textit{Reconstruction}. To reconstruct ($\mathsf{Rec(\cdot)}$) a sharing $\llbracket x \rrbracket$, $\mathcal{P}_1$ sends $\langle x \rangle_1$ to $\mathcal{P}_2$, and $\mathcal{P}_2$ sends $\langle x \rangle_2$ to $\mathcal{P}_1$. 
	Both $\mathcal{P}_1$ and $\mathcal{P}_2$ compute and obtain $x = \langle x \rangle_1 + \langle x \rangle_2$.
	(2) \textit{Addition/subtraction}. Addition/subtraction of secret-shared values can be completed by party $\mathcal{P}_i$ non-interactively for $i\in\{1,2\}$: To securely compute $\llbracket z \rrbracket = \llbracket x\pm y \rrbracket$, each party $\mathcal{P}_i$ locally computes $\langle z \rangle_i = \langle x \rangle_i\pm\langle y \rangle_i$. 
	(3) \textit{Multiplication}. Multiplication of two secret-shared values is computed using \textit{Beaver triples} \cite{Beaver}. 
	A Beaver triple is a multiplication triple $(a, b, c)$ secret-shared among $\mathcal{P}_1$ and $\mathcal{P}_2$, where $a,b$ are uniformly random values in $\mathbb{Z}_{p}$ and $c=ab$. 
	In practice, we can let Beaver triples be generated in advance by a third party \cite{Chameleon} and distributed to $\mathcal{P}_1$ and $\mathcal{P}_2$ . 
	To multiply two secret-shared values $\llbracket x \rrbracket$ and $\llbracket y \rrbracket$, each party $\mathcal{P}_i$ first locally computes $\langle e \rangle_i=\langle x \rangle_i-\langle a \rangle_i$ and $\langle f \rangle_i=\langle y \rangle_i-\langle b \rangle_i$ for $i\in\{1,2\}$.
	Then both parties reconstruct $e$ and $f$. 
	Finally, each party $\mathcal{P}_i$ proceeds to compute $\langle x\cdot y \rangle_i=(i-1)\cdot e\cdot f + f\cdot\langle a\rangle_i +e\cdot\langle b \rangle_i+\langle c \rangle_i$ for $i\in\{1,2\}$.
	%

	\section{Problem Statement}
	\label{sec:prob-statement}
	\begin{figure}[!t]
		\centering
		\includegraphics[scale=0.4]{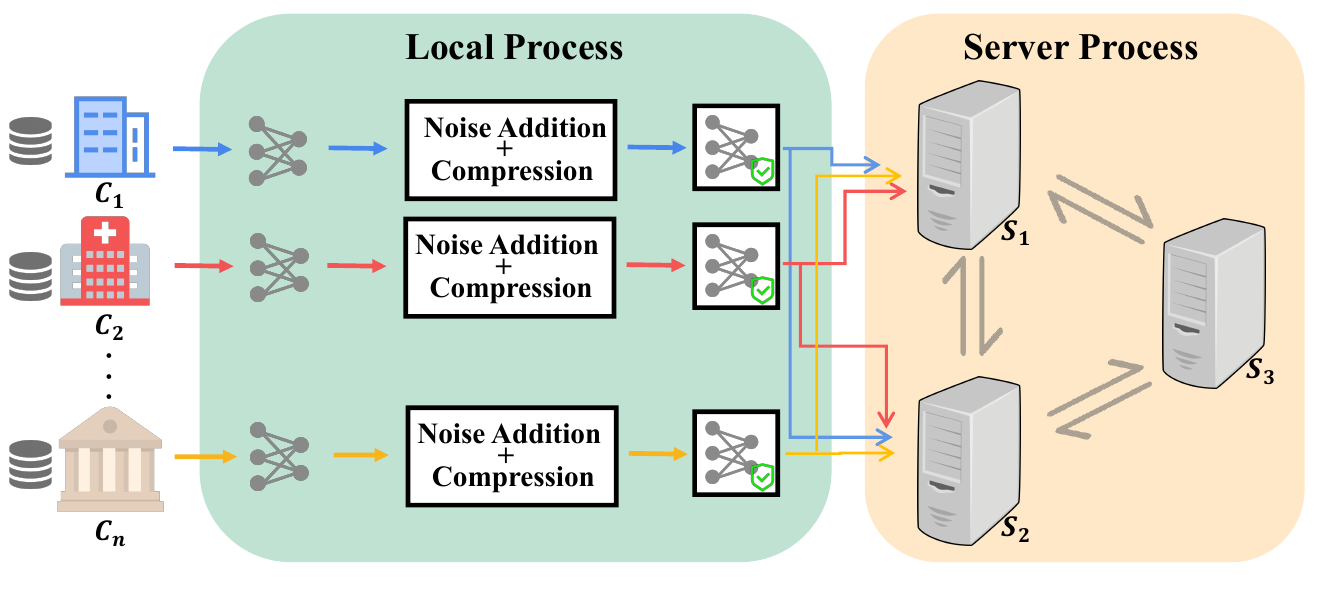}
		\caption{The system model of {\main}.}
		\label{fig:sys_model}
		\vspace{-15pt}
	\end{figure}

	\subsection{System Model}
	\label{sec:problem}
	
	Fig. \ref{fig:sys_model} illustrates {\main}'s system model. 
	{\main}'s design leverages an emerging distributed trust setting where three servers deployed in separate trust domains, each denoted by $\mathcal{S}_i$ for $i\in\{1,2,3\}$, collaboratively provide the FL service for the clients holding local datasets. 
	The adoption of such setting has also appeared recently in works on FL \cite{ELSA,GehlharM0SWY23,tang2024flexible} as well as in other secure systems and applications \cite{DautermanRPS22,NDSS22,NDSS24}. 
	For simplicity of presentation, we will denote the three servers $\mathcal{S}_1$, $\mathcal{S}_2$, and $\mathcal{S}_3$ collectively as $\mathcal{S}_{\{1,2,3\}}$.

	From a high-level point of view, each training iteration of {\main} starts with each client locally computing the gradient for each local data point. 
	Next, each client needs to adequately perturb its gradients under LDP before forwarding them to the servers. 
	Considering that the large sizes of gradients can lead to a communication bottleneck for FL, we introduce a noisy gradient compression mechanism run by the clients locally to compress the noisy gradients. 
	The compressed noisy gradients are then sent to the servers, which collaboratively perform a tailored secure shuffle on the received gradients to achieve privacy amplification by shuffling. 
	\revise{
		It is important to note that it is the compressed gradients that get securely shuffled in {\main}. }
	Finally, these gradients are decompressed and integrated into the global model.

	\subsection{Threat Model and Security Guarantees}
	\label{sec:threat_model}

	{\main} is focused on providing privacy protection for the clients, and we consider the threats primarily from the servers. 
	Like prior works under the three-server setting \cite{NDSS22,PLASMA}, we assume a non-colluding and honest-majority threat model. 
	That is, we assume that each of $\mathcal{S}_{\{1,2,3\}}$ may individually try to deduce private information during the protocol execution, and at most one of the three servers will maliciously deviate from the protocol specification. 
	%
	%
	%
	%
	%

	In {\main}, the individual gradients are locally perturbed by the clients satisfying $\varepsilon_0$-LDP.
	The servers can only view the shuffled noisy gradients and cannot learn which client sends which gradient. 
	That is, no server could learn the permutation used for shuffling the compressed noisy gradients. 
	Moreover, the integrity of the computation regarding the shuffle, sampling, decompression, and aggregation can be checked. 
	Specifically, if there exists a malicious server attempting to tamper with the integrity of the related computation, {\main} will detect the malicious behavior and output abort.

	\section{The Design of {\main}}
	\label{sec:protocol}
	\subsection{Overview}
	\label{sec:protocol:overview}
	
	{\main} is aimed at enabling communication-efficient and private FL services in the shuffle model of DP, with security against a maliciously acting server.
	We note that when the shuffle operation is adequately instantiated, {\main} can enjoy the benefit of privacy amplification provided by shuffling. In the meantime, we observe that the sampling of gradients can also be beneficial to privacy amplification \cite{NeurIPS21,AISTATS21,BalleBG18}. 
	Hence, {\main} also integrates the strategy of gradient sampling so as to achieve stronger privacy amplification.
	Inspired by \cite{AISTATS21}, {\main} conducts gradient sampling at both the client-side and the server-side.

	Next, we consider how to have an efficient and secure realization of the core shuffle computation with robustness against a malicious server in {\main}.
	At first glance, it seems that traditional mixnet-based methods \cite{mixnet,HooffLZZ15} could be used, where a set of servers take turns to perform a verifiable shuffle of data. 
	However, such approaches are expensive due to the computation of a verifiable shuffle at each server \cite{NDSS22}, making it not a promising choice for building {\main} efficiently.
	We observe that a recent trend for secure shuffle with better efficiency is to build on lightweight secret sharing techniques and have a set of servers collaboratively and securely shuffle secret-shared data \cite{ChaseGP20, NDSS22, NDSS24}.

	Our starting point is to leverage a state-of-the-art secret-shared shuffle protocol in the three-server setting from \cite{NDSS22}.
	This protocol allows three servers to jointly perform a shuffle of secret-shared data so that no server learns anything about the permutation used to shuffle data.
	%
	%
	%
	However, simply adopting this protocol in {\main} to shuffle the gradients (perturbed under a LDP mechanism) from the clients would still suffer from inefficiency. 
	This is because the servers would need to routinely exchange information whose size depends on the gradients' sizes to perform the secret-shared shuffle.
	Since gradients are typically high-dimensional vectors in FL, directly secret-sharing (noisy) gradients for the shuffle would incur high performance overhead.

	To overcome this efficiency challenge, our key insight is to have the servers perform a secret-shared shuffle of \emph{compressed} noisy gradients, rather than raw noisy gradients.
	This requires the development of solution that can simultaneously perturb (as per LDP) and compress the clients' gradients, so that the compressed noisy gradients can be suitably used in the secure shuffle computation. 
	To meet this requirement, our key idea is to delicately build on the widely-popular LDP mechanism of \cite{duchi2018minimax}---referred to as DJW18 in this paper---for local gradient perturbation and losslessly compress a noisy high-dimensional gradient vector from this LDP mechanism using simply a random seed and a sign bit. 
	As per our custom design, each client can just secret-share compressed noisy gradients and the servers work over them to perform the secure shuffle.
	By reducing the secret-shared shuffle of high-dimensional noisy gradient vectors to compressed noisy gradients simply consisting of small-sized random seeds and sign bits, we manage to substantially diminish both client-server and inter-server communication costs.	
	
	On the security side, we note that the protocol in \cite{NDSS22} provides a verification mechanism to harden the security in the presence of a malicious server.
	Specifically, it appends MACs to the data to be shuffled, and adds a series of integrity checks performed by the servers during the shuffle computation.
	A malicious server trying to tamper with the shuffling process will be detected.
	However, as shown by the recent work in \cite{NDSS24}, the verification mechanism in \cite{NDSS22} is vulnerable to online selective failure attacks, which result in a \textit{non-random} shuffle in the view of a malicious server. 
	Simply following the verification mechanism of \cite{NDSS22} thus would make the privacy amplification by shuffling theorem problematic for use in {\main}.
	%
	Therefore, we follow the general idea of MAC-based verification in \cite{NDSS22}, but instead develop a new series of lightweight integrity checks so as to ensure the security of the shuffle computation in the presence of a malicious server, with the online selective failure attacks taken into account. 
	Additionally, it is noted that in {\main} the servers need to perform not only the secret-shared shuffle, but also a set of post-shuffle operations including sampling, decompression, and aggregation of (noisy) gradients.
	We also show to perform integrity checks for these operations, and thus deliver a complete solution for maliciously secure computation at the server-side.

	In what follows, we first design a noisy gradient compression mechanism (Section \ref{sec:protocol:ldp-compression}) in which clients locally perturb and losslessly compress gradients under LDP.
	Based on this mechanism, we introduce a basic construction (Section \ref{sec:protocol:semi-honest}) for communication-efficient and private FL in the shuffle model of DP.
	Our basic construction involves the secure shuffling and sampling of gradients to achieve \textit{composed privacy amplification} effects, with semi-honest servers assumed.
	%
	%
	We then show how to extend the basic construction to provide malicious security when at least two servers are honest and arrive at {\main}.

	\subsection{Noisy Gradient Compression under LDP} 
	\label{sec:protocol:ldp-compression}

	\begin{algorithm}[!t]
		\caption{Perturbing and Compressing A Vector under LDP (\textsf{NoisyGradCmpr})} 
		\label{algo:NoisyGradCmpr}
		\begin{algorithmic}[1]
			\Require LDP level $\varepsilon_0$, a vector $\boldsymbol{x}$ clipped to $\ell_2$-norm bound $L$.  
			\Ensure A compressed vector $\boldsymbol{r}$ that satisfies $\varepsilon_0$-LDP. 
			
			\State Compute $\overline{\boldsymbol{x}} \leftarrow \begin{cases}L \cdot \frac{\boldsymbol{x}}{\|\boldsymbol{x}\|_2} & \text { with probability } \frac{1}{2}+\frac{\|\boldsymbol{x}\|_2}{2 L}\text {,} \\ -L \cdot \frac{\boldsymbol{x}}{\|\boldsymbol{x}\|_2} & \text { otherwise. }\end{cases}$ 
			\State Sample $U \sim $ Bernoulli$(\frac{e^{\varepsilon_0}}{e^{\varepsilon_0}+1})$. \Comment \emph{$U=1$ with probability $\frac{e^{\varepsilon_0}}{e^{\varepsilon_0}+1}$, otherwise $U=0$.}
			\State ${s} \stackrel{\$}{\leftarrow} \mathbb{Z}_{p}$.  
			\State $\hat{{s}} \leftarrow \textrm{PRG}({s})$. \Comment \emph{Use the seed ${s}$ and a \textrm{PRG} to generated a sequence of $de$ bits.} 
			\State Expand $\hat{{s}}$ into a $d$-dimensional vector $\boldsymbol{v}$, where each dimension is sequentially represented using $e$ bits from $\hat{{s}}$:
			\begin{align*}
				\boldsymbol{v} &\leftarrow \begin{bmatrix}
					\textsf{Encode}(\hat{{s}}[1,e]) \\
					\textsf{Encode}(\hat{{s}}[e+1,2e]) \\
					\vdots \\
					\textsf{Encode}(\hat{{s}}[(d-1)\cdot e+1,de])
				\end{bmatrix}
			\end{align*} \label{code:1:5}
			\If {$\langle \boldsymbol{v}, \boldsymbol{x}\rangle < 0$}
			\State $\boldsymbol{v} \leftarrow -\boldsymbol{v}$.
			\EndIf
			\State $sgn \leftarrow 2U -1$. \Comment \emph{Calculate the sign of the perturbed vector.}
			\State $\boldsymbol{r} \leftarrow sgn||{s}$. \Comment \emph{Concatenate the sign and seed.}
			
		\end{algorithmic}
	\end{algorithm}

	We first introduce how to perturb and losslessly compress a gradient via our proposed \textsf{NoisyGradCmpr} mechanism, which inputs a gradient (treated as a $d$-dimensional vector) and LDP level $\varepsilon_0$, and outputs a compressed vector that satisfies $\varepsilon_0$-LDP. 
	As will be shown in our complete protocol of FL (Section \ref{sec:protocol:semi-honest} and \ref{sec:protocol:malicious}), \textsf{NoisyGradCmpr} can be integrated into the training process to perturb and compress each gradient at the client, thus facilitating communication-efficient federated model training.
	
	At a high level, \textsf{NoisyGradCmpr} builds upon the widely adopted LDP mechanism DJW18 \cite{duchi2018minimax}, yet introduces a novel approach by employing a PRG to losslessly compress the output perturbed vector. 
	In DJW18, simply put, a perturbed vector is generated by creating a random vector $\boldsymbol{v}$ and calculating a sign bit $sgn$ based on $\boldsymbol{v}$ and the input vector $\boldsymbol{x}$. 
	The perturbed vector is output as $sgn\cdot\boldsymbol{v}$; for more details, refer to \cite{Erlingsson20,duchi2018minimax,AISTATS21}.
	Our key observation is that generating $\boldsymbol{v}$ from a random seed and only transmitting the compressed noisy vector (comprising a random seed and a sign bit) to servers could significantly reduce the client-server communication cost \cite{Erlingsson20}. 
	Moreover, servers could perform a secret-shared shuffle on the compressed noisy gradients, further reducing the inter-server communication cost associated with secret sharing. 
	Algorithm \ref{algo:NoisyGradCmpr} outlines how to locally perturb a vector and utilizes a PRG to compress the perturbed vector. 
	Algorithm \ref{algo:NoisyGradDcmp} describes the decompression of a noisy vector, coupled with Algorithm \ref{algo:NoisyGradCmpr}. 
	%
	
	Apart from the substantial reduction in communication cost achieved by our proposed method, our proposed \textsf{NoisyGradCmpr} also possesses the following properties.
	
	\begin{lem}
		\label{lem:NoisyGradCmpr}
		Our \textsf{NoisyGradCmpr} presented in Algorithm \ref{algo:NoisyGradCmpr}, when used in couple with \textsf{NoisyGradDcmp} presented in Algorithm \ref{algo:NoisyGradDcmp}, achieves lossless compression, is unbiased, guarantees $\varepsilon_0$-LDP, and ensures that the decompressed vector has bounded variance.
		Specifically, for every $\boldsymbol{x}\in\mathcal{B}^d_2(L)$, where $\mathcal{B}^d_2(L) = \left\{\boldsymbol{x} \in \mathbb{R}^d:\|\boldsymbol{x}\|_2 \leq L\right\}$ denotes the $\ell_2$-norm ball of radius $L$, we have $\mathbb{E}[\textsf{R}(\boldsymbol{x})] = \boldsymbol{x}$ and 
		$$
		\mathbb{E}\|\textsf{R}(\boldsymbol{x}) - \boldsymbol{x}\|^2_2 \leq L^2 d\left(\frac{3 \sqrt{\pi}}{4} \frac{e^{\varepsilon_0}+1}{e^{\varepsilon_0}-1}\right)^2, 
		$$
		\noindent where $\textsf{R}(\cdot)=\textsf{NoisyGradDcmp}(\textsf{NoisyGradCmpr}(\cdot))$. 
	\end{lem}
	
	\begin{algorithm}[!t]
		\caption{Decompressing A Compressed Noisy Vector (\textsf{NoisyGradDcmp})} 
		\label{algo:NoisyGradDcmp}
		\begin{algorithmic}[1]
			\Require A vector $\boldsymbol{r}$ compressed by \textsf{NoisyGradCmpr} with LDP level $\varepsilon_0$, $\ell_2$-norm bound $L$. 
			\Ensure  A decompressed vector $\boldsymbol{x}\in\mathbb{R}^d$ that satisfies $\varepsilon_0$-LDP. 
			
			\State $(sgn,{s}) \leftarrow \boldsymbol{r}$. \Comment{\emph{Deconcatenation.}}
			\State $\hat{{s}} \leftarrow \textrm{PRG}({s})$. 
			\State $\boldsymbol{v} \xleftarrow{\text{expand}} \hat{{s}}$. \Comment \emph{Same operation as line \ref{code:1:5}, Algorithm \ref{algo:NoisyGradCmpr}.}
			\State $\overline{\boldsymbol{v}}\leftarrow \frac{\boldsymbol{v}}{\|\boldsymbol{v}\|_2}$. 
			\State $M \leftarrow L \frac{\sqrt{\pi}}{2} \frac{d\Gamma\left(\frac{d-1}{2}+1\right)}{\Gamma\left(\frac{d}{2}+1\right)} \frac{e^{\varepsilon_0+1}}{e^{\varepsilon_0-1}}$. 
			\State $\boldsymbol{x} \leftarrow sgn\cdot M\cdot\overline{\boldsymbol{v}}$. 
		\end{algorithmic}
		\vspace{-1mm}
	\end{algorithm}

	We provide the proof of Lemma \ref{lem:NoisyGradCmpr} in Appendix \ref{appendix:lem_NoisyGradCmpr_proof}. 
	Note that Lemma \ref{lem:NoisyGradCmpr} indicates that our mechanism \textsf{NoisyGradCmpr} provides a gradient with $\varepsilon_0$-LDP guarantee. 
	Later in this paper, we will show that the (compressed) noisy gradients will be shuffled and sampled, and the $\varepsilon_0$-LDP guarantee can be amplified through composing privacy amplification by shuffling and privacy amplification by subsampling. 
	This results in a tighter bound on privacy loss for each iteration (see Section \ref{sec:theoretical-results} for a more detailed analysis).

	\subsection{Communication-Efficient FL in the Shuffle Model of DP with Semi-Honest Security}
	\label{sec:protocol:semi-honest}

	%
	In this section, we introduce our basic construction for federated model training in the shuffle model of DP, which considers a semi-honest adversary setting. 
	Algorithm \ref{algo:SGD-workflow} presents our basic construction, which bridges our proposed noisy gradient compression scheme \textsf{NoisyGradCmpr} and the secret-shared shuffle protocol. 
	%
	%
	We will show how to extend our basic construction for achieving malicious security later in Section \ref{sec:protocol:malicious}.

	At the beginning, a global model $\theta_0$ is required to be initialized on the server side. 
	Since our basic construction considers semi-honest servers, the initialization can be done by any server. 
	For simplicity, we assign the initialization of $\theta_0$ to $\mathcal{S}_1$. 
	Thus $\mathcal{S}_1$ initializes and broadcasts $\theta_0$ to each client $\mathcal{C}_{i}$ for $i\in[n]$. 
	Recall from Section \ref{sec:protocol:overview}, the client-side sampling of gradients, together with server-side gradient sampling, enables an \textit{additional privacy amplification} effect (via subsampling \cite{BalleBG18}) apart from privacy amplification by shuffling. 
	To achieve such additional privacy amplification, inspired by the sampling strategy from \cite{AISTATS21}, we let $\mathcal{C}_i$ uniformly sample a subset $\mathcal{U}_{i}$ of $s$ data points and compute gradients for the sampled data points at the beginning of each training iteration. 
	The server-side gradient sampling will be shown later in this section.
	
	Next, each client $\mathcal{C}_i$ clips the gradient $\nabla_{\theta_t}\ell(\theta_t, \boldsymbol{x}_{ij})$ for each data point $j\in \mathcal{U}_{i}$ using clipping parameter $L$ to bound the $\ell_2$-norm of the gradient. 
	Here, $\ell(\theta_t, \cdot)$ is a loss function. 
	Then $\mathcal{C}_i$ applies our proposed noisy gradient compression mechanism \textsf{NoisyGradCmpr} to perturb and compress each (sampled) gradient. 
	At the end of the local process, $\mathcal{C}_i$ obtains a set of $s$ compressed noisy gradients, each comprising a random seed and a sign bit, satisfying $\varepsilon_0$-LDP. 
	
	After that, $\mathcal{C}_i$ splits the compressed noisy gradients into two shares and distributes them among $\mathcal{S}_{\{1,2\}}$. 
	Thus, $\mathcal{S}_{\{1,2\}}$ hold a length-$N$ vector $\llbracket \boldsymbol{x}\rrbracket= \llbracket\left( \{\boldsymbol{r}^t_{1j}\}_{j\in\mathcal{U}_{1}}, \{\boldsymbol{r}^t_{2j}\}_{j\in\mathcal{U}_{2}}, \cdots, \{\boldsymbol{r}^t_{nj}\}_{j\in\mathcal{U}_{n}} \right)\rrbracket$, comprising $ns$ secret-shared compressed noisy gradients ($N=ns$). 
	Then $\mathcal{S}_{\{1,2,3\}}$ collaboratively shuffle $\llbracket \boldsymbol{x}\rrbracket$, which enables a privacy amplification effect by shuffling \cite{BalleBGN19}. 
	
	We adapt the secret-shared shuffle protocol \cite{NDSS22} in {\main} to instantiate the secure shuffle computation (denoted as $\textsf{SecShuffle}$) as follows. 
	To securely shuffle $\llbracket \boldsymbol{x}\rrbracket$, $\mathcal{S}_3$ needs to interact with $\mathcal{S}_{\{1,2\}}$ to generate the correlations required for performing a secret-shared shuffle in advance. 
	Specifically, $\mathcal{S}_{1}$ chooses a random seed and expands it to obtain $\pi_1,\boldsymbol{a}'_2,\boldsymbol{b}_2$, where $\pi_1:\mathbb{Z}_N \rightarrow \mathbb{Z}_N$ is a random permutation and $\boldsymbol{a}'_2,\boldsymbol{b}_2$ are length-$N$ random vectors. 
	Next, $\mathcal{S}_{1}$ sends the seed to $\mathcal{S}_3$, which then expands it to recover the same values. 
	Similarly, $\mathcal{S}_{2}$ chooses a random seed, expands it to obtain $\pi_2,\boldsymbol{a}_1$, and sends this seed to $\mathcal{S}_3$. 
	We can significantly reduce communication costs by transmitting only the relevant seeds to $\mathcal{S}_3$. 
	Note that since $\mathcal{S}_3$ learns the permutations $\pi_1,\pi_2$ used by $\mathcal{S}_{\{1,2\}}$, we follow \cite{NDSS22} to additionally have $\mathcal{S}_{\{1,2\}}$ pre-share a permutation $\pi_{12}$ and apply it to $\llbracket \boldsymbol{x}\rrbracket$ to get a locally permuted vector $\llbracket\hat{\boldsymbol{x}}\rrbracket$. 
	Upon expanding the received seeds to retrieve the values, $\mathcal{S}_3$ calculates a vector $\boldsymbol{\Delta} = \pi_2(\pi_1(\boldsymbol{a}_1) + \boldsymbol{a}'_2) - \boldsymbol{b}_2$ and sends it to $\mathcal{S}_{2}$. 
	This completes the \textit{offline} phase of the secret-shared shuffle. 
	
	The \textit{online} phase of the secret-shared shuffle proceeds as follows: 
	Firstly, $\mathcal{S}_2$ masks its input share $\langle \hat{\boldsymbol{x}}\rangle_2$ using $\boldsymbol{a}_1$ and sends $\boldsymbol{z}_2 \leftarrow \langle \hat{\boldsymbol{x}}\rangle_2 - \boldsymbol{a}_1$ to $\mathcal{S}_1$. 
	Secondly, $\mathcal{S}_1$ sets its output to be $\langle\boldsymbol{y}\rangle_1\leftarrow\boldsymbol{b}_2$ and sends $\boldsymbol{z}_1 \leftarrow \pi_1(\boldsymbol{z}_2+\langle\hat{\boldsymbol{x}}\rangle_1) - \boldsymbol{a}'_2$ to $\mathcal{S}_2$. 
	Finally, $\mathcal{S}_2$ sets its output to be $\langle\boldsymbol{y}\rangle_2\leftarrow\pi_2(\boldsymbol{z}_1) + \boldsymbol{\Delta}$. 
	
	The correctness of the secret-shared shuffle is as follows: 
	\begin{equation}
		\begin{aligned} 
			\langle \boldsymbol{y} \rangle_1 &+ \langle \boldsymbol{y} \rangle_2 =\pi_2\left(\boldsymbol{z}_1\right)+\boldsymbol{\Delta}_2+\boldsymbol{b}_2  \\ 
			& = \pi_2\left(\pi_1\left(\boldsymbol{z}_2 + \langle\hat{\boldsymbol{x}}\rangle_1\right)-\boldsymbol{a}_2^{\prime}\right) + \pi_2\left(\pi_1\left(\boldsymbol{a}_1\right)+\boldsymbol{a}_2^{\prime}\right) \\
			& = \pi_2(\pi_1(\langle \hat{\boldsymbol{x}}\rangle_2 - \boldsymbol{a}_1 + \langle\hat{\boldsymbol{x}}\rangle_1) + \pi_2(\pi_1(\boldsymbol{a}_1)) \\
			& = \pi_2(\pi_1(\hat{\boldsymbol{x}})) \\
			& = \pi_2(\pi_1(\pi_{12}(\boldsymbol{x}))).
		\end{aligned}
		\nonumber
	\end{equation}
	
	\noindent Here, no single server could view all three permutations $\pi_1,\pi_2,\pi_{12}$. 
	Specifically, $\mathcal{S}_1$ has $\pi_1,\pi_{12}$, $\mathcal{S}_2$ has $\pi_2,\pi_{12}$, and $\mathcal{S}_3$ has $\pi_1,\pi_2$. 
	
	\begin{algorithm}[!t]
		\caption{Our Basic Construction for Communication-Efficient Secure FL in the Shuffle Model of DP} 
		\label{algo:SGD-workflow}
		\begin{algorithmic}[1]
			\Require Each client $\mathcal{C}_i$ holds a local dataset $\mathcal{D}_i$ for $i\in[n]$. 
			\Ensure$ \mathcal{S}_1$ and each client $\mathcal{C}_{i}$ obtain a global model $\theta$ that satisfies $(\varepsilon,\delta)$-DP  for $i\in[n]$. 
			
			\State $\mathcal{S}_{1}$ initializes: ${\theta}_0 \in \mathcal{G}$. 
			\For {$t\in [T]$}\
			\State // \emph{\underline{Client local process:}}
			\For{each client $\mathcal{C}_i$}
			\State \multiline{$\mathcal{C}_i$ chooses uniformly at random a set $\mathcal{U}_{i}$ of $s$ data points.}
			\For{data point $j\in\mathcal{U}_{i}$}
			\State $\mathbf{g}^t_{ij}\leftarrow\nabla_{\theta_t}\ell(\theta_t, \boldsymbol{x}_{ij})$.
			\State $\tilde{\mathbf{g}}^t_{ij} \leftarrow \mathbf{g}^t_{ij} / \max \left\{1, \frac{\left\|\mathbf{g}^t_{ij}\right\|_2}{L}\right\}$.
			\State \multiline{$\boldsymbol{r}^t_{ij} \leftarrow$ \textsf{NoisyGradCmpr}($\tilde{\mathbf{g}}^t_{ij}$). \Comment \emph{Noisy gradient compression under LDP.}}
			\EndFor
			\State \multiline{$\mathcal{C}_i$ splits compressed gradients $\{\boldsymbol{r}^t_{ij}\}_{j\in\mathcal{U}_{i}}$ into two shares and distributes them among $\mathcal{S}_1$ and $\mathcal{S}_2$. }
			\EndFor
			
			\State // \emph{\underline{Server-side computation:}}
			
			\State $\llbracket\boldsymbol{x}\rrbracket \leftarrow \llbracket \left( \{\boldsymbol{r}^t_{1j}\}_{j\in\mathcal{U}_{1}}, \{\boldsymbol{r}^t_{2j}\}_{j\in\mathcal{U}_{2}}, \cdots, \{\boldsymbol{r}^t_{nj}\}_{j\in\mathcal{U}_{n}} \right) \rrbracket$.
			\State \multiline{$\llbracket\pi(\boldsymbol{x})\rrbracket$ $\leftarrow$ \textsf{SecShuffle}($\llbracket \boldsymbol{x}\rrbracket$). \Comment \emph{Secret-shared shuffle of $N=ns$ compressed noisy gradients.}}
			
			\State \multiline{$\{\boldsymbol{r}_i\}_{i\in [B]}$ $\xleftarrow{\text{sample}}$ \textsf{Rec}($\llbracket\pi(\boldsymbol{x})\rrbracket$). \Comment \emph{Reconstruct and sample the first $B=ks$ shuffled compressed gradients.}} \label{code:fl:reconstruct} 
			
			\For{$i\in [B]$}
			\State \multiline{$\mathbf{g}_i\leftarrow$ \textsf{NoisyGradDcmp}($\boldsymbol{r}_i$). \Comment \emph{Decompression.}\label{code:fl:decompress}}
			\EndFor
			
			\State $\overline{\mathbf{g}}^t\leftarrow \frac{\sum_{i=1}^B\mathbf{g}_i}{B}$. \label{code:fl:sample}
			
			\State $\theta_{t+1} \leftarrow \prod_{\mathcal{G}}\left(\theta_t-\eta_t \overline{\mathbf{g}}^t\right)$. \label{code:fl:aggregation}
			
			\EndFor
		\end{algorithmic}
	
	\end{algorithm}

	To achieve additional privacy amplification via subsampling, sampling of shuffled (compressed noisy) gradients after the secure shuffle is conducted once again at the server side.
	Inspired by the sampling strategy in \cite{AISTATS21}, we set $B=ks$ ($k\in[n]$) for analyzing the privacy amplification effect. 
	Such sampling can be securely achieved by $\mathcal{S}_{\{1,2\}}$, which reconstruct the first $B$ elements out of the $N$ (shuffled) compressed noisy gradients by disclosing each other the corresponding secret shares. 
	These $B$ compressed noisy gradients can be regarded as \textit{uniformly} sampled from the initial set of $N$ secret-shared compressed noisy gradients, given that the permutations used in the secret-shared shuffle are uniformly random. 
	
	After secure sampling, these compressed noisy gradients are required to be decompressed, and then integrated to update the global model $\theta_t$ at the $t$-th iteration. 
	Given that we are operating under the assumption of a semi-honest adversary setting in our basic construction, the computations involving gradient decompression, aggregation, and global model update (line \ref{code:fl:decompress} - \ref{code:fl:aggregation}) can be assigned to either $\mathcal{S}_{1}$ or $\mathcal{S}_{2}$. 
	Since we have already let $\mathcal{S}_{1}$ initialize the global model at the beginning, we can assign these computations to $\mathcal{S}_{1}$.

	\subsection{Achieving Malicious Security}
	\label{sec:protocol:malicious}
	
	In our basic construction described in Section \ref{sec:protocol:semi-honest}, we have assumed a semi-honest adversary setting.   
	To offer an integrity guarantee against the malicious adversary defined in Section \ref{sec:threat_model}, we need to have integrity checks for the following operations: (1) shuffle, (2) sampling, (3) decompression, and (4) aggregation of the noisy gradients. 
	Specifically, if a malicious server attempts to deviate from the protocol during these operations, the honest servers will detect this misbehavior and output abort.  
	We present how {\main} guarantees integrity in the presence of a malicious server as follows. 
	
	At a high level, regarding verifying the correctness of a secret-shared shuffle, we follow \cite{NDSS22} to adopt the (post-shuffle) blind MAC verification scheme. 
	In particular, we have each (compressed noisy) gradient locally MACed with a key before the secret-shared shuffle. 
	The MAC and the key are then (securely) shuffled together with this gradient. 
	After shuffling, the shuffled MACs could be \textit{blindly} verified in a batch in the secret sharing domain. 
	
	However, we also notice that the recent work \cite{NDSS24} highlights the vulnerability of solely relying on a post-shuffle blind MAC verification: the attacks proposed in \cite{NDSS24}, if successful, leak information about the underlying permutations in \cite{NDSS22}. 
	This results in the shuffle not random in a malicious server's view, making the privacy amplification by shuffling theorem not applicable in our considered shuffle model of DP, as it necessitates random shuffling of the data \cite{BalleBGN19} (see Section \ref{sec:protocol:verify_shuffle} for a more detailed discussion of the attack). 
	In contrast to \cite{NDSS24}, which attempts to reduce leakage by naively repeating the secret-shared shuffle, we propose a novel defense mechanism to defend against the attacks identified in \cite{NDSS24}. 
	Our approach involves integrating integrity checks into the \textit{online} phase of the secret-shared shuffle, as will be shown in Section \ref{sec:protocol:verify_shuffle}.
	
	In addition to checking the integrity of the secret-shared shuffle, we propose additional checks (Section \ref{sec:protocol:verifiable_aggregation}) to ensure the integrity of (post-shuffle) sampling, decompression, and aggregation results using lightweight cryptographic techniques such as hashing. 
	
	\begin{figure}[!t]
		\centering
		\scalebox{0.9}{
			\fbox
			{
				\shortstack[l]{
					$\mathcal{S}_1$ holds $\pi_1,\boldsymbol{a}'_2,\boldsymbol{b}_2$; $\mathcal{S}_2$ holds $\pi_2,\boldsymbol{a}_1$;  $\mathcal{S}_3$ holds $\pi_1,\pi_2,\boldsymbol{a}'_2,\boldsymbol{b}_2,\boldsymbol{a}_1,\boldsymbol{\Delta}$. \\
					After $\mathcal{S}_2$ sends $\boldsymbol{z}_2\leftarrow \langle \hat{\boldsymbol{x}}\rangle_2 - \boldsymbol{a}_1$ to $\mathcal{S}_1$, $\mathcal{S}_{\{1,3\}}$ conduct:\\
					(1) $\mathcal{S}_1$ locally computes $\hat{\boldsymbol{x}} - \boldsymbol{a}_1$, where $\hat{\boldsymbol{x}} - \boldsymbol{a}_1 \leftarrow \boldsymbol{z}_2+\langle\hat{\boldsymbol{x}}\rangle_1$. \\
					(2) $\mathcal{S}_1$ splits $\hat{\boldsymbol{x}} - \boldsymbol{a}_1$ into two shares and discloses one share to $\mathcal{S}_3$. \\
					(3) $\mathcal{S}_3$ splits $\boldsymbol{a}_1$ into two shares and discloses one share to $\mathcal{S}_{1}$. \\
					(4) $\mathcal{S}_{\{1,3\}}$ calculate $\llbracket \hat{\boldsymbol{x}}\rrbracket$ by summing $\llbracket \hat{\boldsymbol{x}}- \boldsymbol{a}_1\rrbracket$ and $\llbracket\boldsymbol{a}_1\rrbracket$. \\
					(5) $\mathcal{S}_{\{1,3\}}$ calculate $f$ following Eq. \ref{eq:verification} and outputs abort if $f\neq0$.  
				}
		}}
		\caption{The integrity check for $\boldsymbol{z}_2$ sent from $\mathcal{S}_2$.}
		\label{fig:mac_check_z2}
		\vspace{-15pt}
	\end{figure}
	
	\subsubsection{Maliciously Secure Secret-Shared Shuffle}
	\label{sec:protocol:verify_shuffle}
	In this section, we first present how to use blind MAC verification to check the integrity of a secret-shared shuffle. 
	Then we review the attacks (from \cite{NDSS24}) on this verification mechanism, and finally give our defense mechanism against the attacks.
	
	\noindent\textbf{Achieving Malicious Security Using Blind MAC Verification.}
	Our starting point is to follow the general strategy of \cite{NDSS22}, leveraging the blind MAC verification scheme to \textit{blindly} check whether the secret-shared shuffle is performed correctly. 
	The blind MAC verification scheme in \cite{NDSS22} employs Carter-Wegman MAC \cite{WegmanC81}, which is defined as follows: 
	for a compressed noisy gradient $\boldsymbol{r}$ (padded to fixed length $l$), the MAC of this gradient is computed as 
	\begin{equation}
		\label{eq:MAC}
		t \leftarrow \sum_{j=1}^{l} \boldsymbol{k}[j] \cdot \boldsymbol{r}[j], 
	\end{equation}
	\noindent where $\boldsymbol{k}\in\mathbb{Z}^l_p$ is a random key and $\boldsymbol{k}[j]$ denotes the $j$-th element of $\boldsymbol{k}$. 
	Integrating Carter-Wegman MAC, our shuffle protocol can be divided into three steps: (1) client local processing, (2) secret-shared shuffle, and (3) post-shuffle blind MAC check.

	At the local process, each client $\mathcal{C}_i$ follows Eq. \ref{eq:MAC} to compute MAC $t_{ij}$ for each (compressed noisy) gradient $\boldsymbol{r}_{ij}$ using the MAC key $\boldsymbol{k}_{ij}$ for $i\in[n],j\in\mathcal{U}_{i}$. 
	Recall that $\mathcal{U}_{i}$ is the subset containing $s$ sampled data points. 
	Here, $\boldsymbol{k}_{ij}$ is formed by $\mathcal{C}_i$ sampling two MAC key seeds $\mathsf{m}_{ij,1}, \mathsf{m}_{ij,2} \in \mathbb{Z}_p$ and computing 
	\begin{equation}
		\notag
		\begin{aligned}
			\boldsymbol{k}_{ij}  = (\boldsymbol{k}_{ij}[1], \cdots, \boldsymbol{k}_{ij}[l]) 
			\leftarrow \textrm{G}(\mathsf{m}_{ij,1}) + \textrm{G}(\mathsf{m}_{ij,2}),  
		\end{aligned}
	\end{equation}
	where G: $\mathbb{Z}_p \rightarrow \mathbb{Z}^{l}_p$ is a PRG. 
	In this way, the MAC key $\boldsymbol{k}_i$ is split into two secret shares $(\langle t_{ij} \rangle_1, \langle \boldsymbol{r}_{ij} \rangle_1, \mathsf{m}_{ij,1})$ and $(\langle t_{ij} \rangle_2, \langle \boldsymbol{r}_{ij} \rangle_2, \mathsf{m}_{ij,2})$, which are then sent to $\mathcal{S}_1$ and $\mathcal{S}_2$, respectively. 
	Upon receiving the shares, $\mathcal{S}_1$ locally runs $\langle \boldsymbol{k}_{ij} \rangle_1 \leftarrow \textrm{G}(\mathsf{m}_{ij,1})$ and $\mathcal{S}_2$ locally runs $\langle \boldsymbol{k}_{ij} \rangle_2 \leftarrow \textrm{G}(\mathsf{m}_{ij,2})$. 
	Finally, $\mathcal{S}_{\{1,2\}}$ hold secret-shared elements $\llbracket \boldsymbol{x}_{ij} \rrbracket = (\llbracket t_{ij} \rrbracket, \llbracket \boldsymbol{r}_{ij} \rrbracket, \llbracket \boldsymbol{k}_{ij} \rrbracket)$ of size-($2l+1$) for $i\in[n], j\in\mathcal{U}_{i}$ and a length-$N$ vector $\llbracket \boldsymbol{x} \rrbracket = (\llbracket \boldsymbol{x}_{11} \rrbracket, \cdots,\llbracket \boldsymbol{x}_{1s} \rrbracket,\cdots,$ $ \llbracket \boldsymbol{x}_{n1} \rrbracket, \cdots,\llbracket \boldsymbol{x}_{ns} \rrbracket)$.

	At the second step, $\mathcal{S}_{\{1,2,3\}}$ only need to follow the same secret-shared shuffle process to shuffle $\llbracket \boldsymbol{x} \rrbracket$ as described in Section \ref{sec:protocol:semi-honest}. 
	If a malicious server tampers with any information in this process, misbehavior will be detected by blind MAC verification in the third step after shuffling. 
	The blind MAC verification is performed by $\mathcal{S}_{\{1,2,3\}}$ firstly calculating the secret-shared MAC 
	$$
	\sum_{i=1}^{n}\sum_{j=1}^{s} \left( \sum_{k=1}^{l} \llbracket \boldsymbol{k}_{ij}[k] \rrbracket \llbracket \boldsymbol{r}_{ij}[k] \rrbracket\right). 
	$$
	The secret-shared multiplication is computed using Beaver triples provided by $\mathcal{S}_{3}$. 
	Then the verification can be finished by computing and reconstructing 
	\begin{small}
	\begin{equation}
		\label{eq:verification}
		\begin{split}
			f = & \textsf{Rec}\left[\llbracket w \rrbracket \left(\sum_{i=1}^{n}\sum_{j=1}^{s} \llbracket t_{ij} \rrbracket -   \sum_{i=1}^{n}\sum_{j=1}^{s} \left( \sum_{k=1}^{l} \llbracket \boldsymbol{k}_{ij}[k] \rrbracket \llbracket \boldsymbol{r}_{ij}[k] \rrbracket \right) \right)\right], 
		\end{split}
	\end{equation}	
	\end{small}
	where $\llbracket w \rrbracket$ is a random secret-share of a random $w\in\mathbb{Z}_p$, obtained by servers sampling random shares. 
	If there does not exist a malicious server, the secret-shared MAC calculated using $\llbracket \boldsymbol{r}_{ij} \rrbracket$ and $\llbracket \boldsymbol{k}_{ij} \rrbracket$ should match the initial MAC $t_{ij}$ generated at the client's local process. 
	Therefore, if the verification outputs $f=0$, the correctness of the shuffle can be ensured. 
	Otherwise, the honest servers output abort and stop. 
	In this way, all the MACs are blindly verified by the servers together as one batch in the secret sharing domain. 
	
	Note that $\mathcal{S}_{1}$ or $\mathcal{S}_{2}$, if malicious, could lie about its share of $f$ to force $f=0$. 
	For example, $\mathcal{S}_{1}$ could sets $\langle f \rangle_1\leftarrow-\langle f \rangle_2$ if he first receives the share $\langle f \rangle_2$ sent by $\mathcal{S}_{2}$ for reconstructing $\llbracket f \rrbracket$. 
	To detect this misbehavior, we initiate a process where servers $\mathcal{S}_{\{1,2\}}$ exchange hashes of their shares of $f$ before disclosing the actual shares. 
	The exchanged hashes prevent any attempt by a malicious server to tamper with its shares in order to forge MACs.

	\begin{figure}[!t]
		\centering
		\scalebox{0.9}{
		\fbox
		{
			\shortstack[l]{
				$\mathcal{S}_1$ holds $\pi_1,\boldsymbol{a}'_2,\boldsymbol{b}_2$; $\mathcal{S}_2$ holds $\pi_2,\boldsymbol{a}_1$;  $\mathcal{S}_3$ holds $\pi_1,\pi_2,\boldsymbol{a}'_2, \boldsymbol{b}_2,\boldsymbol{a}_1,\boldsymbol{\Delta}$. \\
				After $\mathcal{S}_1$ sends $\boldsymbol{z}_1 \leftarrow \pi_1(\hat{\boldsymbol{x}} - \boldsymbol{a}_1)-\boldsymbol{a}'_2$ to $\mathcal{S}_2$, $\mathcal{S}_{\{2,3\}}$ conduct:\\
				(1) $\mathcal{S}_2$ locally computes $\pi_2(\boldsymbol{z}_1)$. \\
				(2) $\mathcal{S}_2$ splits $\pi_2(\boldsymbol{z}_1)$ into two shares and discloses one share to $\mathcal{S}_3$. \\
				(3) $\mathcal{S}_3$ splits $\pi_2(\pi_1(\boldsymbol{a}_1)+\boldsymbol{a}'_2)$ into two shares and discloses one\\  share to $\mathcal{S}_{2}$. \\
				(4) $\mathcal{S}_{\{2,3\}}$ calculate $\llbracket \pi_2(\pi_1(\hat{\boldsymbol{x}}))\rrbracket$ by summing  $\llbracket\pi_2(\boldsymbol{z}_1)\rrbracket$ and \\ $\llbracket\pi_2(\pi_1(\boldsymbol{a}_1)+\boldsymbol{a}'_2)\rrbracket$. \\
				(5) $\mathcal{S}_{\{2,3\}}$ calculate $f$ following Eq. \ref{eq:verification} and outputs abort if $f\neq0$.  
			}
		}}
		\caption{The integrity check for $\boldsymbol{z}_1$ sent from $\mathcal{S}_1$.}
		\label{fig:mac_check_z1}
		\vspace{-15pt}
	\end{figure}

	\noindent\textbf{Online Selective Failure Attacks \cite{NDSS24}.}
	Although the (post-shuffle) blind MAC verification provides an effective solution to detect misbehavior by a malicious server, recent work \cite{NDSS24} has shown that checking all the MACs in a batch still allows potential privacy leakage. 
	Specifically, in the secure shuffling process, $\mathcal{S}_{1}$ and $\mathcal{S}_{2}$ could launch the online selective failure attacks. 
	We first present how $\mathcal{S}_2$ launches an attack: before sending $\boldsymbol{z}_2 \leftarrow \langle \hat{\boldsymbol{x}}\rangle_2 - \boldsymbol{a}_1$ to $\mathcal{S}_{1}$, $\mathcal{S}_{2}$ locally samples a vector $\boldsymbol{u}$ the same structure as $\boldsymbol{x}$, with only one non-zero entry at position $q$ and other entries being 0. 
	$\mathcal{S}_{2}$ then sends $\boldsymbol{z}_2 \leftarrow \langle \hat{\boldsymbol{x}}\rangle_2 - \boldsymbol{a}_1 + \boldsymbol{u}$ to $\mathcal{S}_{1}$, instead of sending $\boldsymbol{z}_2 \leftarrow \langle \hat{\boldsymbol{x}}\rangle_2 - \boldsymbol{a}_1$. 
	Before the post-shuffle blind MAC verification process, $\mathcal{S}_{2}$ randomly guesses a position $p$ as the permuted position of $q$ after the secret-shared shuffle. 
	Then $\mathcal{S}_{2}$ shifts the non-zero entry of $\boldsymbol{u}$ from position $q$ to $p$ to create a new vector $\boldsymbol{v}$, and then sets the output of secret-shared shuffle to be $\langle\boldsymbol{y}\rangle_1\leftarrow\boldsymbol{b}_2 - \boldsymbol{v}$, instead of $\langle\boldsymbol{y}\rangle_1\leftarrow\boldsymbol{b}_2$. 
	If $\mathcal{S}_{2}$ makes a correct guess (with a probability of $1/N$), the added error $\boldsymbol{u}$ will be canceled by $\boldsymbol{v}$ before the post-shuffle verification. 
	Therefore, the integrity check will still output $f=0$, and the misbehavior will be left undetected. 
	Similarly, $\mathcal{S}_{1}$ could launch the same attack when sending $\boldsymbol{z}_1 \leftarrow \pi_1(\boldsymbol{z}_2+\langle\hat{\boldsymbol{x}}\rangle_1) - \boldsymbol{a}'_2$ to $\mathcal{S}_2$. 
	
	\noindent\textbf{Adding Integrity Checks During Secure Shuffling Process.} In practice, it is hard to directly verify the legitimacy of $\boldsymbol{z}_1,\boldsymbol{z}_2$. 
	The work \cite{NDSS24} that first proposes the online selective failure attack repeatedly performs secret-shared shuffle for $K$ times to reduce the success probability of a selective failure attack from $1/N$ to $1/N^K$. 
	However, this straightforward leakage reduction method significantly increases both communication and computation costs of a secure shuffle for $K$ times in our setting, while still being susceptible to selective failure attacks with a probability of $1/N^K$.
	
	In contrast, we propose a novel verification mechanism to check the legitimacy of $\boldsymbol{z}_1,\boldsymbol{z}_2$ during the secure shuffling process. 
	Our proposed verification mechanism effectively protects the integrity against selective failure attacks, at the expense of only two additional blind MAC verifications and four rounds of data transmission. 
	
	Our verification mechanism is proposed through an in-depth examination of the secure shuffling process. 
	In case of $\mathcal{S}_2$ being malicious in the aforementioned attack, we need to verify the legitimacy of $\boldsymbol{z}_2$ sent by $\mathcal{S}_2$. 
	After $\mathcal{S}_2$ sends $\boldsymbol{z}_2\leftarrow \langle \hat{\boldsymbol{x}}\rangle_2 - \boldsymbol{a}_1$ to $\mathcal{S}_1$, we let $\mathcal{S}_{\{1,3\}}$ follow the procedure given in Fig. \ref{fig:mac_check_z2} to check whether $\boldsymbol{z}_2$ is tampered with. 
	Suppose that $\mathcal{S}_2$ has added an error at position $q$ before sending $\boldsymbol{z}_2$. 
	As illustrated in Fig. \ref{fig:mac_check_z2}, we can exclude the malicious $\mathcal{S}_2$ in the verification process and only let the honest servers $\mathcal{S}_{\{1,3\}}$ conduct blind MAC verification based on $\llbracket \hat{\boldsymbol{x}}\rrbracket$. 
	This prevents $\mathcal{S}_2$ from canceling out the added error, leading to a failed integrity check. 
	Such integrity check for $\boldsymbol{z}_2$ only requires a blind MAC verification and two rounds of communication. 
	Note that in the verification process, $\mathcal{S}_1$ has no access to $\boldsymbol{a}_1$ and $\mathcal{S}_{3}$ has no access to $\hat{\boldsymbol{x}} - \boldsymbol{a}_1$. 
	This guarantees the confidentiality of $\hat{\boldsymbol{x}}$ as neither $\mathcal{S}_{\{1,3\}}$ can deduce $\hat{\boldsymbol{x}}$ based on their individual knowledge. 
    In case of $\mathcal{S}_1$ being malicious, we can apply a similar strategy to verify the legitimacy of $\boldsymbol{z}_1$ from $\mathcal{S}_1$, as illustrated in Fig. \ref{fig:mac_check_z1}. 
	
	We present the complete maliciously secure secret-shared shuffle protocol, integrated with our proposed defense mechanism, in Algorithm \ref{algo:malicious-shuffle} in Appendix \ref{appdix:complete_protocols}.

	\subsubsection{Maliciously Secure Gradient Sampling, Decompression and Aggregation}
	\label{sec:protocol:verifiable_aggregation}
	Recall that in our basic construction, the post-shuffle computations involving gradient sampling, decompression, and aggregation (line \ref{code:fl:reconstruct}-\ref{code:fl:aggregation}) can be assigned to either $\mathcal{S}_1$ or $\mathcal{S}_2$.
	However, ensuring the correctness of these operations becomes more challenging for achieving malicious security: (1) a malicious server could tamper with its share before reconstructing secret-shared compressed noisy gradients, and (2) a malicious server could tamper with the reconstructed compressed noisy gradients (line \ref{code:fl:reconstruct}), incorrectly decompress gradients (line \ref{code:fl:decompress}), or output an incorrect aggregation result (line \ref{code:fl:aggregation}). 
	
	To detect the aforementioned misbehavior, {\main} verifies the integrity of sampling, decompression, and aggregation results by $\mathcal{S}_{\{1,2\}}$ collaboratively. 
	Given that $\mathcal{S}_{\{1,2\}}$ hold shares of shuffled compressed noisy gradients, a malicious server cannot drop or add fake gradients without detection. 
	Thus, checking the integrity of gradient sampling results is facilitated: We only need to ensure that $\mathcal{S}_{\{1,2\}}$ transmit the correct shares required for reconstructing the first $B$ shuffled compressed noisy gradients. 
	The verification can be achieved by checking the MAC of each gradient, as follows. 
	Firstly, let $\mathcal{S}_{\{1,2\}}$ reconstruct the first $B$ shuffled compressed noisy gradients (along with the MACs and MAC keys): $\mathcal{S}_{\{1,2\}}$ get $\{ \boldsymbol{x}_{i}  = ( t_{i} ,  \boldsymbol{r}_{i} ,  \boldsymbol{k}_{i} )\}_{i\in[B]}$. 
	Similar to reconstructing $\llbracket f \rrbracket$ in Section \ref{sec:protocol:verify_shuffle}, $\mathcal{S}_{\{1,2\}}$ first locally compute the hashes of the first $B$ shuffled compressed noisy gradients' shares before disclosing the actual shares. 
	Then, $\mathcal{S}_{\{1,2\}}$ exchange hashes and check whether the hashes match. 
	If no inconsistencies are detected, $\mathcal{S}_{\{1,2\}}$ check each MAC following Eq. \ref{eq:MAC} and output abort if any verification fails. 
	
	Then we can check the integrity of decompression results and aggregation results together: Let $\mathcal{S}_{\{1,2\}}$ locally decompress the sampled compressed noisy gradients $\{\boldsymbol{r}_{i}\}_{i\in[B]}$ to get $\{\mathbf{g}_{i}\}_{i\in[B]}$. 
	Subsequently, $\mathcal{S}_{\{1,2\}}$ locally compute $\overline{\mathbf{g}}^t\leftarrow \frac{\sum_{i=1}^B\mathbf{g}_i}{B}$ and update the global model following $\theta_{t+1} \leftarrow \prod_{\mathcal{G}}\left(\theta_t-\eta_t \overline{\mathbf{g}}^t\right)$. 
	After that, $\mathcal{S}_{\{1,2\}}$ hash their locally computed $\theta_{t+1}$ and send each other the hash. 
	If the hashes do not match, the protocol outputs abort.

	The complete protocol of the maliciously secure construction of {\main} is given in Algorithm \ref{algo:malicious-SGD-workflow} in Appendix \ref{appdix:complete_protocols}. 
	We formally prove its security in Appendix \ref{appdix:security_proof}.

	\section{Theoretical Analysis of Privacy, Communication, and Convergence}
	\label{sec:theoretical-results}
	
	In this section, we analyze the  privacy guarantee of {\main} by first composing privacy amplification by subsampling with privacy amplification by shuffling at each iteration and then analyzing the RDP of the overall FL process. 
	We also provide analysis in terms of the communication cost and model convergence of {\main}. 
	Theorem \ref{thm:analysis_overview} gives the main analytical results. 

	\begin{thm}
		\label{thm:analysis_overview}
		For sampling rate $\gamma = \frac{B}{M}$, where $B=ks$ and $M=nr$, if we run {\main} over $T$ iterations, then we have: 
		\begin{itemize}
			\item \textbf{Privacy}: {\main} satisfies ($\varepsilon,\delta$)-DP, where $\varepsilon$ is defined as below for $\varepsilon_0 = \mathcal{O}(1)$ and $\delta\in(0,1)$: 
			\begin{equation}
				\notag
				\begin{aligned}
					\quad\quad \varepsilon=\min_\lambda \left( T\varepsilon(\lambda) + \log (1-{1}/{\lambda}) 
					+ \frac{\log ({1}/{\delta})-\log (\lambda)}{\lambda-1} \right).  
				\end{aligned}
			\end{equation}
			Here $\varepsilon(\lambda)$ is the RDP guarantee of each iteration that composes shuffling and subsampling: $ \varepsilon(\lambda) = \frac{\lambda \log^2(1+\gamma(e^{\tilde{\varepsilon}} - 1)) }{2},$
			\noindent where $
			\tilde{\varepsilon}=\mathcal{O}\left(\min \left\{\varepsilon_0, 1\right\} e^{\varepsilon_0} \sqrt{\frac{\log (1 / \tilde{\delta})}{\gamma M}}\right)$ and $\tilde{\delta}\in(0,1)$. 
			
			\item \textbf{Communication}: The complexity of client-server communication and inter-server communication in {\main} is $\mathcal{O}(N)$ for processing $N$ compressed gradients of $s+1$ bits each, where $s$ denotes the length of a random seed.

			\item \textbf{Convergence}: 
			Define the (distributed) empirical risk minimization (ERM) problem: 
			$\arg \min _{\theta \in \mathcal{G}}\left(F(\theta):=\frac{1}{n} \sum_{i=1}^n F_i(\theta)\right)$,
			where $\mathcal{G}\subset \mathbb{R}^d$ is a convex set and $F_i(\theta)=\frac{1}{r}\sum_{j=1}^{r} \ell(\theta,\boldsymbol{x}_{ij})$ is the local loss function at $\mathcal{C}_i$ for $i\in[n]$.
			Assume a convex and $L$-Lipschitz continuous function $\ell: \mathcal{G}\times\mathcal{D} \rightarrow \mathbb{R}$ defined over $\mathcal{G}$ with diameter $D$. 
			By letting $\theta^*=\mathop{\arg\min}_{\theta \in \mathcal{G}} F(\theta)$ denote the optimal solution of the distributed ERM problem and $\eta_t = \frac{D}{G\sqrt{t}}$, where $G=L\sqrt{1+\frac{14d}{\gamma M}(\frac{e^{\varepsilon_0}+1}{e^{\varepsilon_0}-1})^2}$, for any $T>1$, the following holds: 
			\begin{small}
			\begin{equation}
				\notag
				\begin{aligned}
					&\mathbb{E}[F(\theta_T)] - F(\theta^*) \leq \mathcal{O}\left(\frac{L D \log (T)}{\sqrt{T}} \sqrt{\frac{14 d}{\gamma M}}\left(\frac{e^{\varepsilon_0}+1}{e^{\varepsilon_0}-1}\right)\right). 
				\end{aligned}
			\end{equation}
			\end{small}
		\end{itemize} 
	\end{thm}
	
	We provide the proof for Theorem \ref{thm:analysis_overview} in Appendix \ref{appendix:analysis_overview}.

	\begin{table}[!t]
		
		\centering
		\caption{Model Architecture for MNIST.}
		\label{tab:mnist_model}
		\setlength\tabcolsep{3pt}
		
		\scalebox{0.9}{
			\begin{tabular}{|c|c|}
				\hline \text { Layer } & \text { Parameters } \\
				\hline \hline \text { Convolution } & 16 \text { filters of } 8 $\times$ 8, \text { Stride 2 } \\
				\text { Max-Pooling } & 2 $\times$ 2 \\
				\text { Convolution } & 32 \text { filters of } 4 $\times$ 4, \text { Stride 2 } \\
				\text { Max-Pooling } & 2 $\times$ 2 \\
				\text { Fully connected } & 32 { units } \\
				\text { Softmax } & 10 \text { units } \\
				\hline
			\end{tabular}
		}
		\vspace{-10pt}
	\end{table}
	
	\begin{figure*}[!t]
		
		\centering
		
		\begin{minipage}[t]{0.233\linewidth}
			\centering
			\includegraphics[width=\linewidth]{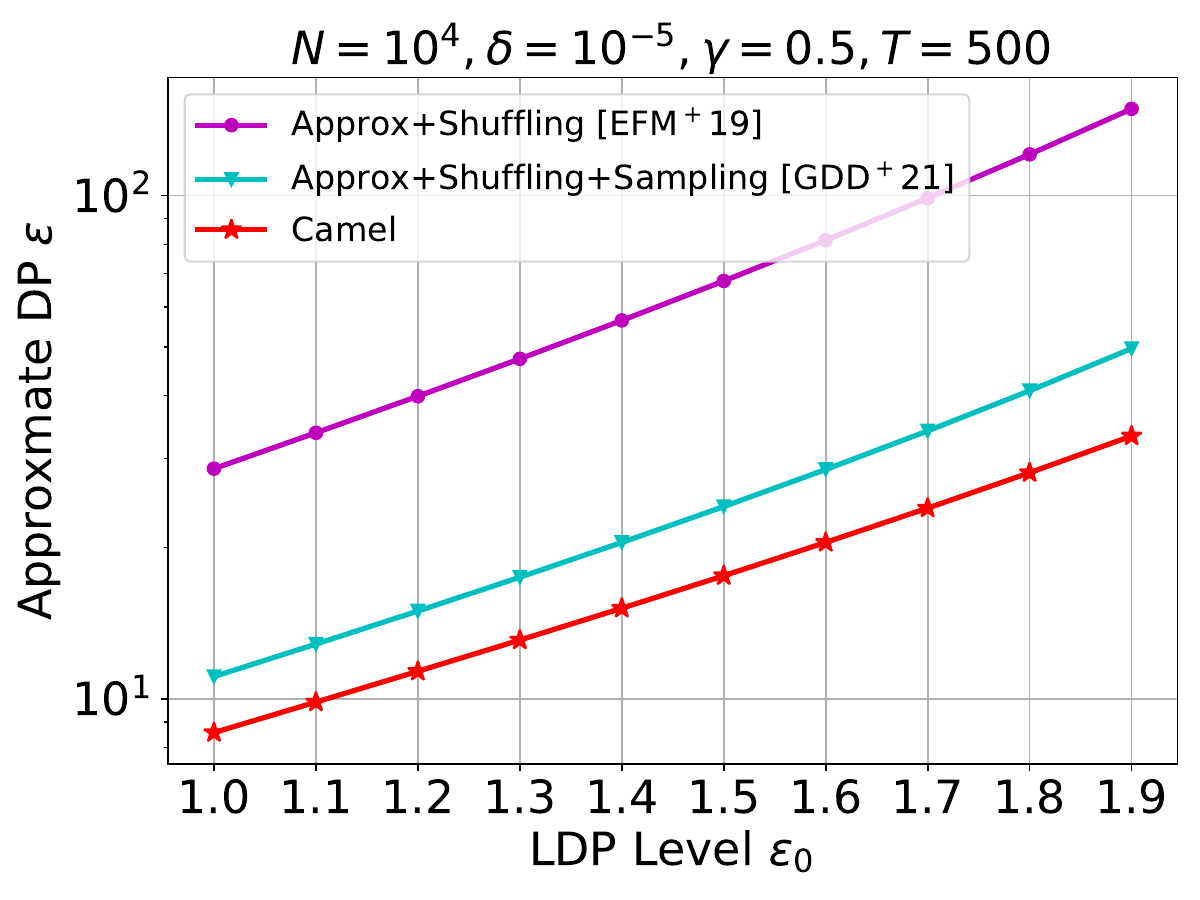}\\(a)
		\end{minipage}
		\begin{minipage}[t]{0.233\linewidth}
			\centering
			\includegraphics[width=\linewidth]{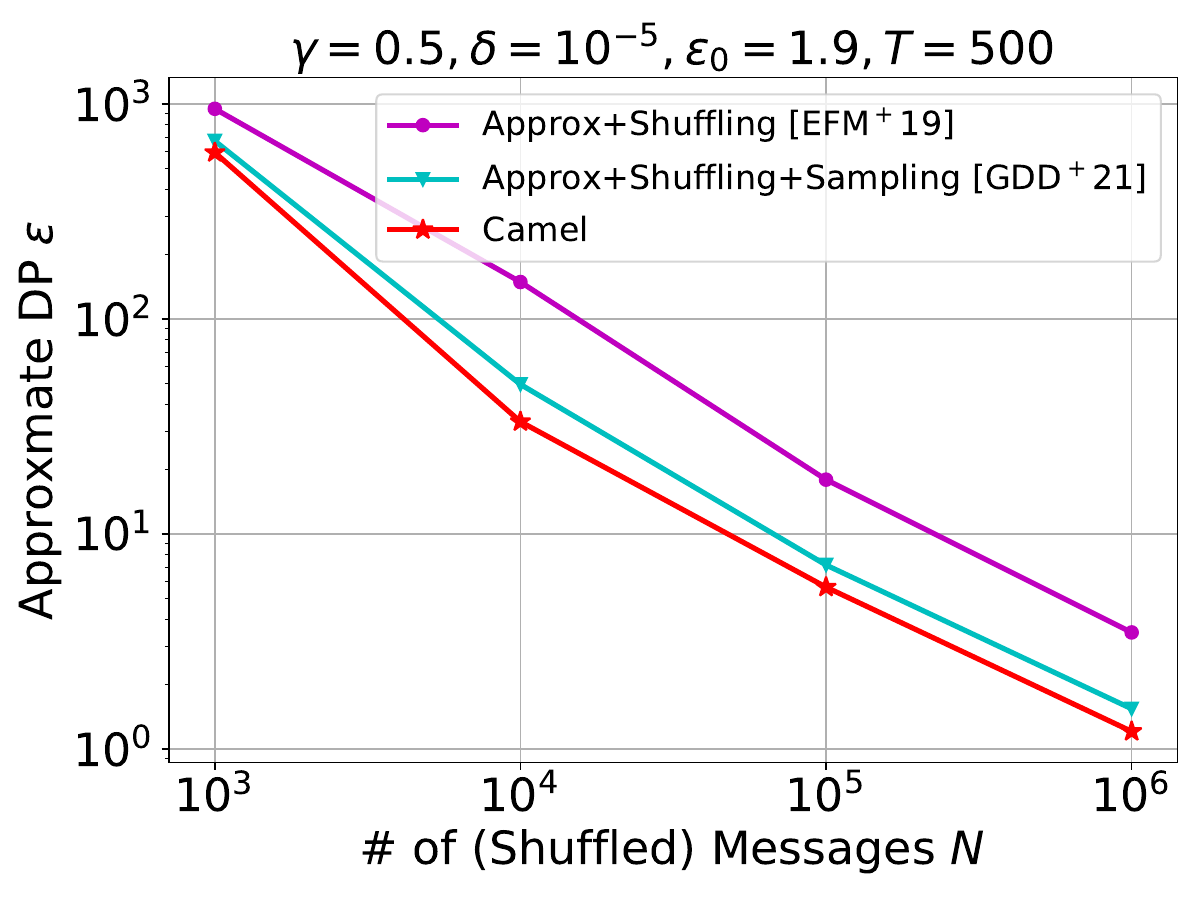}\\(b)
		\end{minipage}
		\begin{minipage}[t]{0.233\linewidth}
			\centering
			\includegraphics[width=\linewidth]{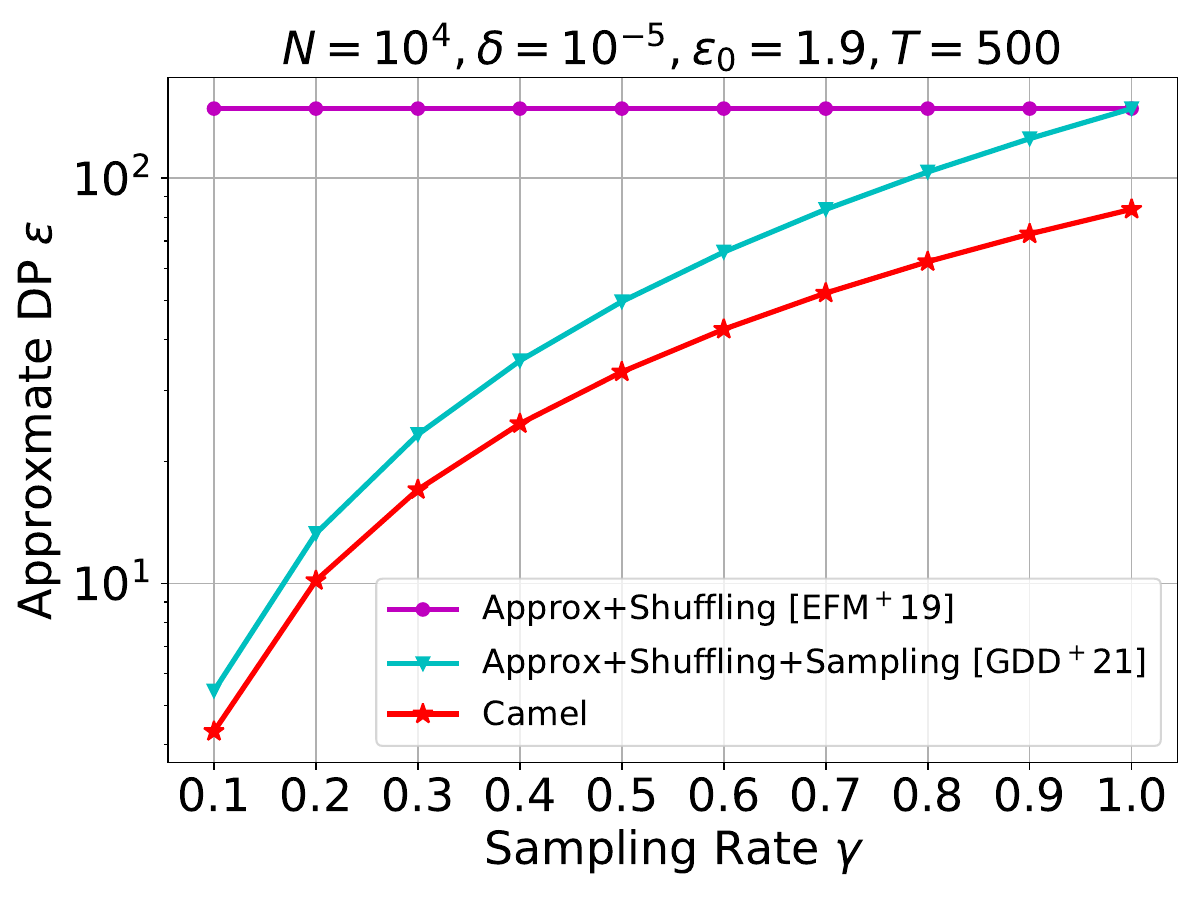}\\(c)
		\end{minipage}
		\begin{minipage}[t]{0.233\linewidth}
			\centering
			\includegraphics[width=\linewidth]{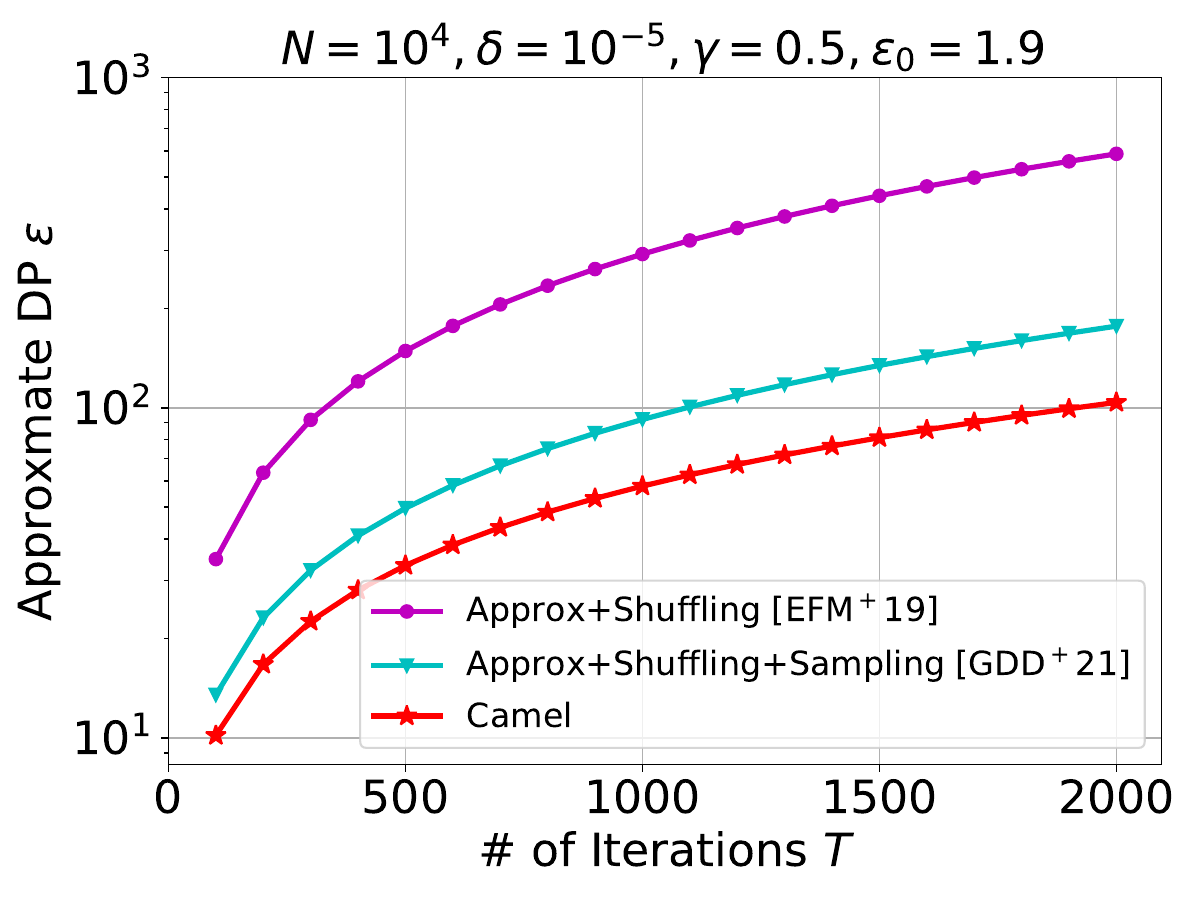}\\(d)
		\end{minipage}
		\caption{Comparison of different bounds on the Approximate ($\varepsilon,\delta$)-DP with fixed $\delta=10^{-5}$: (1) Approximate DP obtained from our derived bound in Theorem \ref{thm:analysis_overview} (converted from RDP to ($\varepsilon,\delta$)-DP) (red); (2) Approximate DP obtained from the bound given in \cite{AISTATS21} (cyan); and (3) Approximate DP obtained from the bound given in \cite{Erlingsson20} (magenta).} 
		\label{fig:numerical_bounds}
		\vspace{-10pt}
		
	\end{figure*}

	\section{Experimental Results}
	\label{sec:exp}
	
	\begin{figure}[!t]
		
		\centering
		
		\begin{minipage}[t]{0.45\linewidth}
			\centering
			\includegraphics[width=\linewidth]{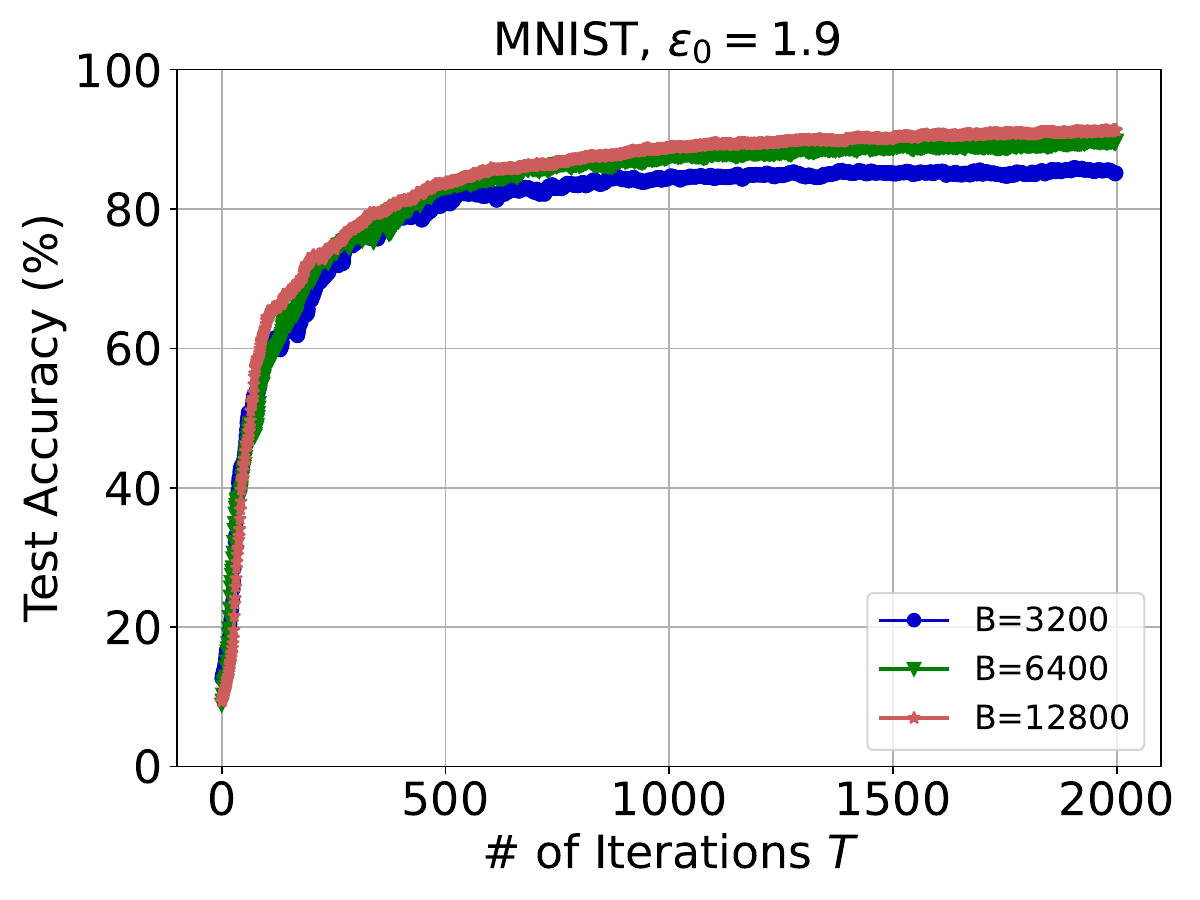}\\(a)
		\end{minipage}
		\begin{minipage}[t]{0.45\linewidth} 
			\centering
			\includegraphics[width=\linewidth]{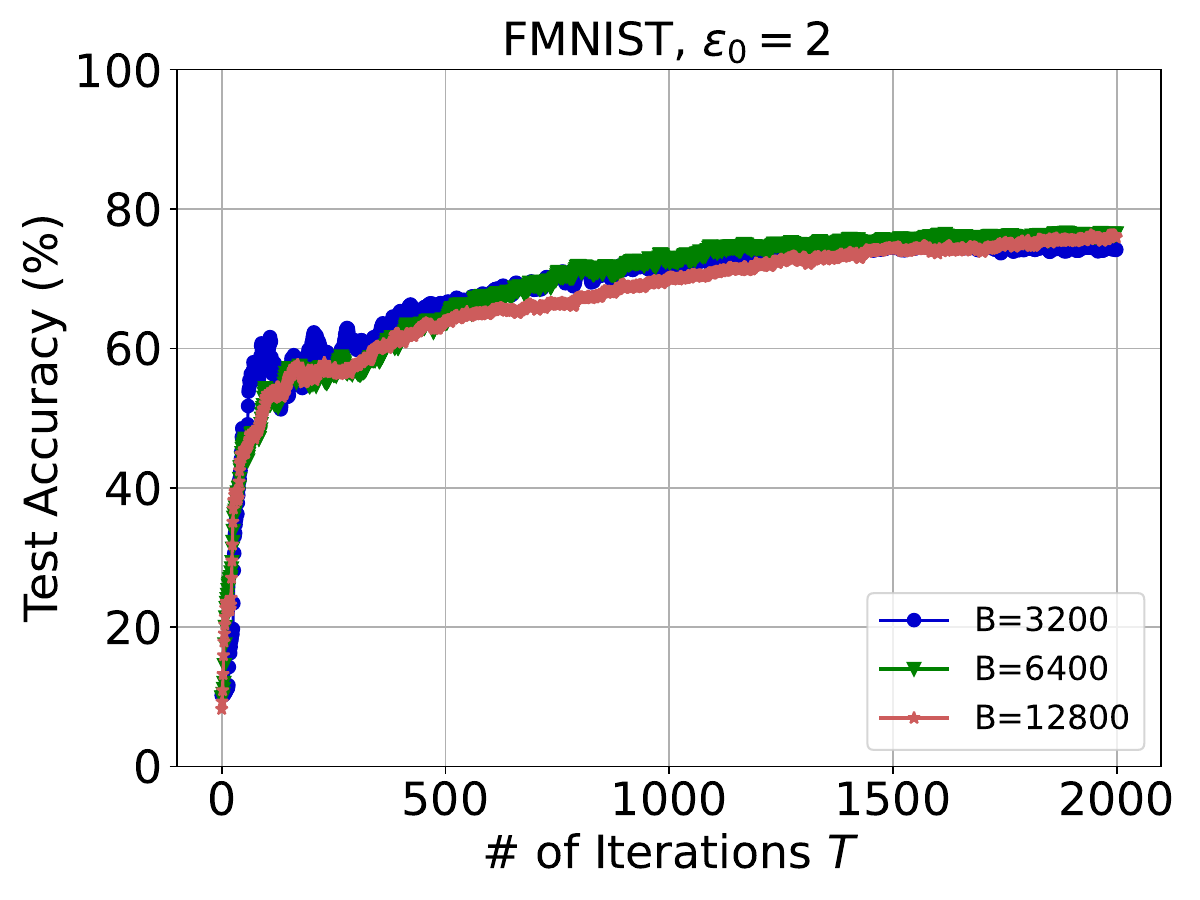}\\(b) 
		\end{minipage}
		
		\caption{Utility performance of {\main}, by fixing $\varepsilon_0=1.9$ for the MNIST dataset, $\varepsilon_0=2.0$ for the FMNIST dataset, and varying $T\in[2000]$ and $B\in\{3200,6400,12800\}$. } 
		\label{fig:utility_vary_T}
		\vspace{-10pt}
		
	\end{figure}
	
	\begin{figure}[!t]
		\centering
		\begin{minipage}[t]{0.233\textwidth}
			\centering
			\includegraphics[width=\textwidth]{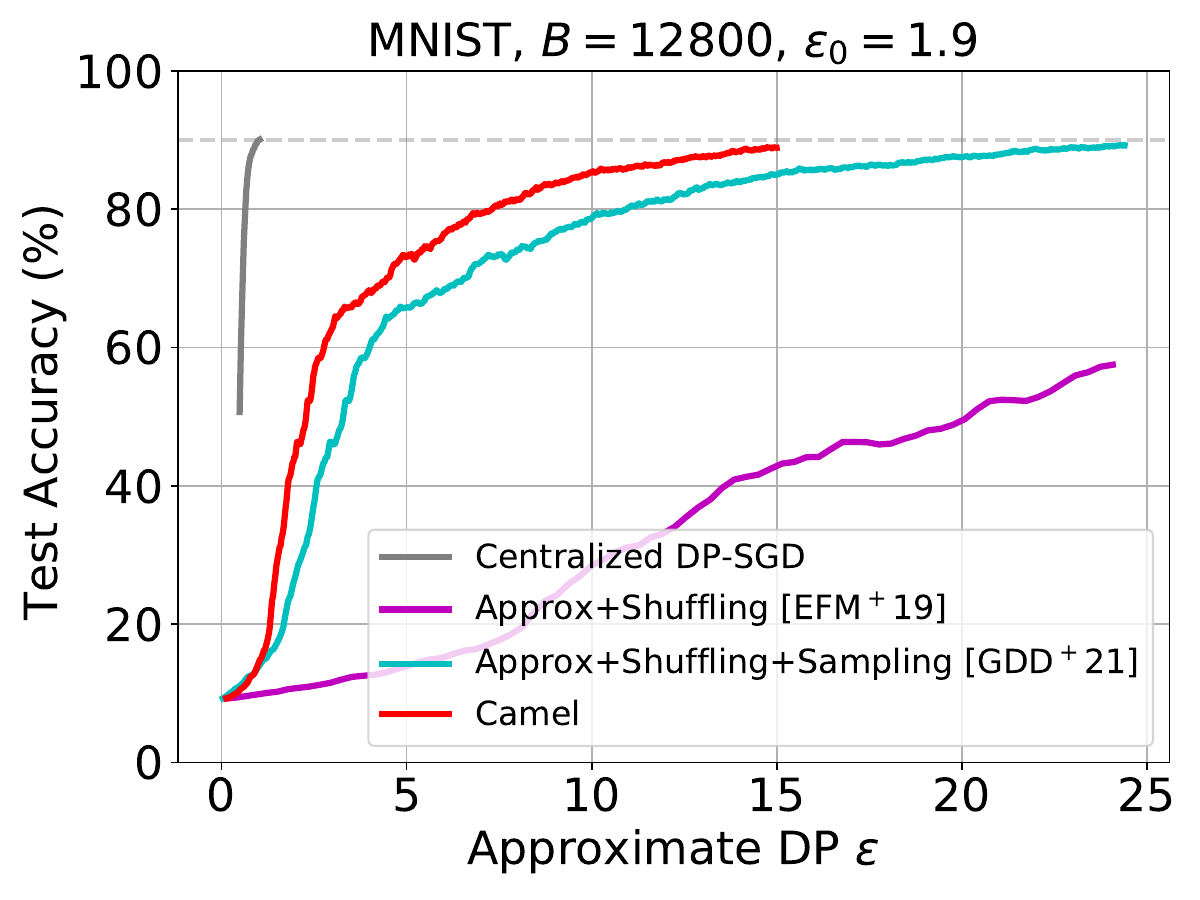}\\{(a)}
		\end{minipage}
		\hfill
		\begin{minipage}[t]{0.233\textwidth}
			\centering
			\includegraphics[width=\textwidth]{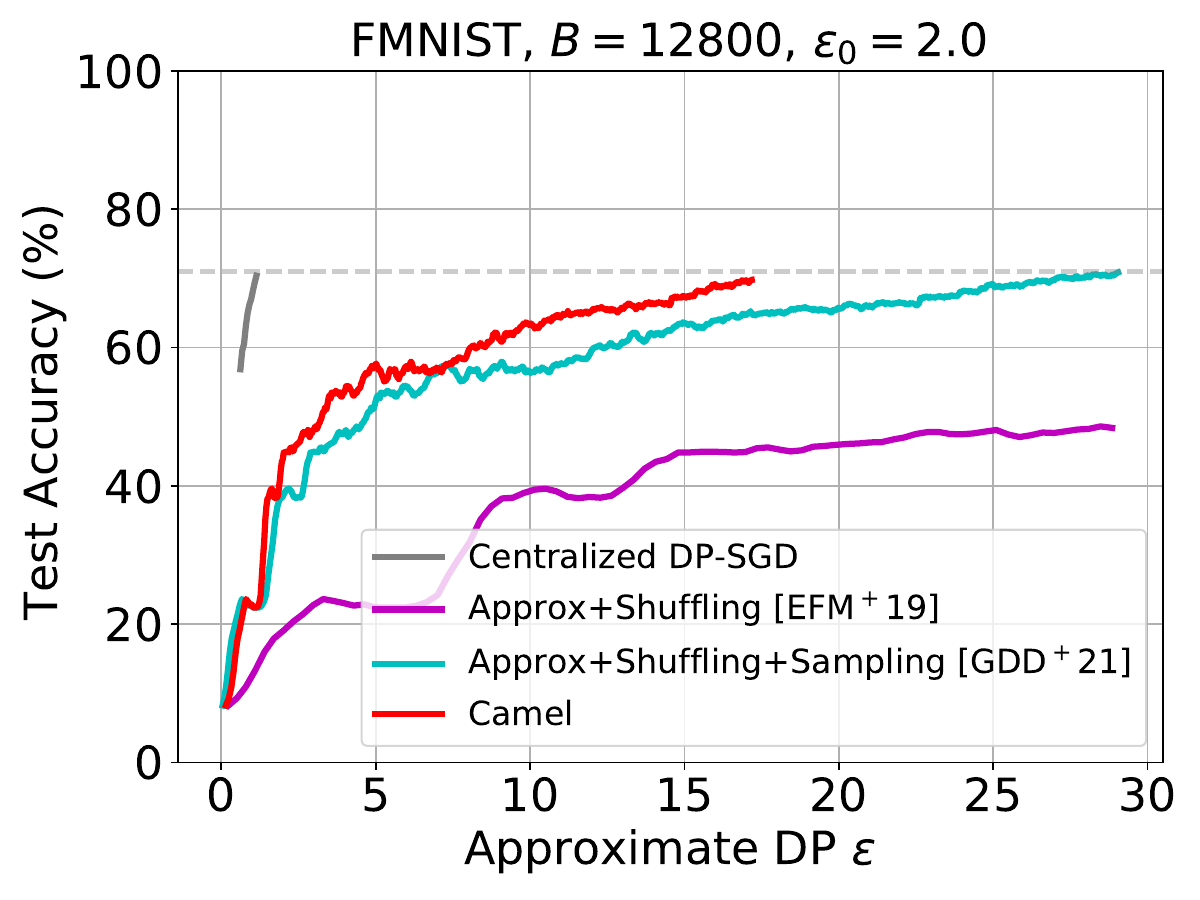}\\{(b)}
		\end{minipage}
		\\
		\begin{minipage}[t]{0.233\textwidth}
			\centering
			\includegraphics[width=\textwidth]{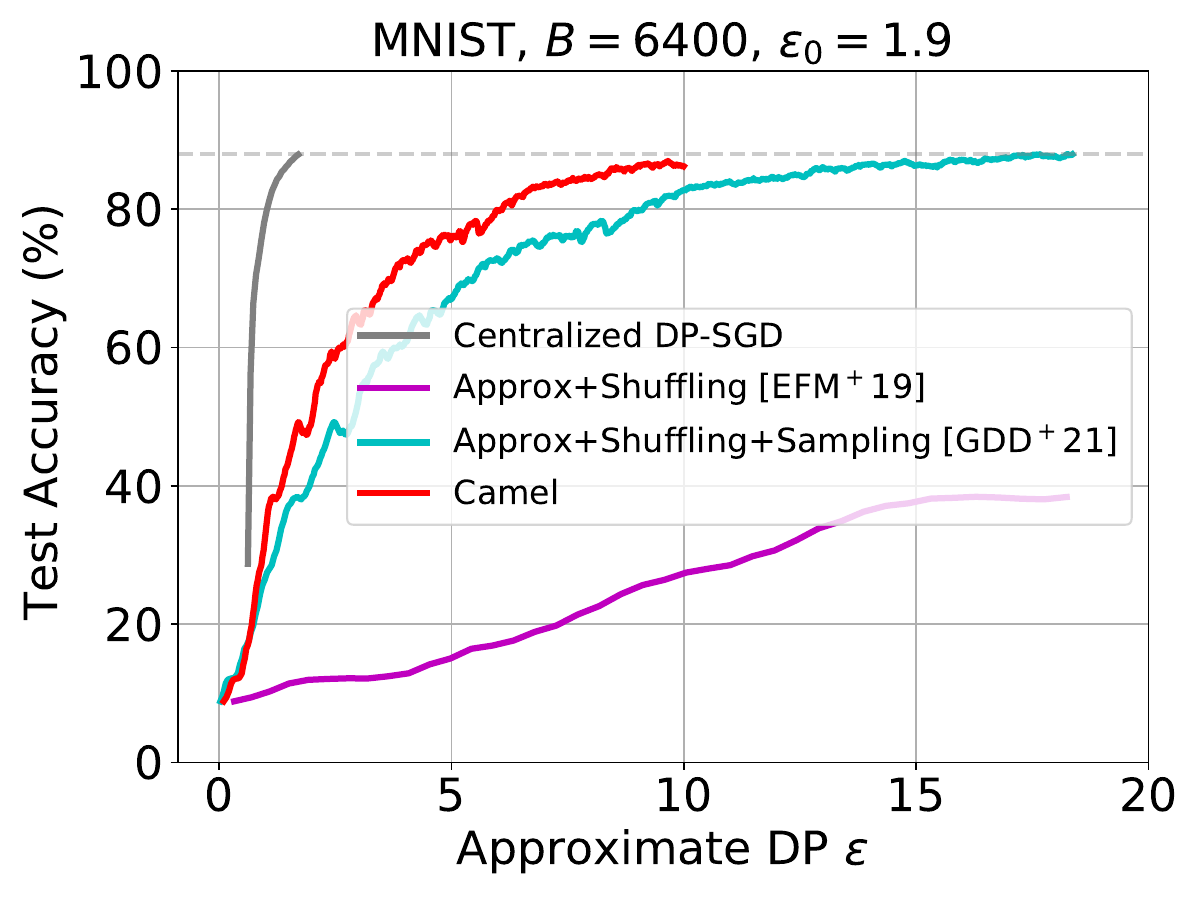}\\{(c)}
		\end{minipage}
		\hfill
		\begin{minipage}[t]{0.233\textwidth}
			\centering
			\includegraphics[width=\textwidth]{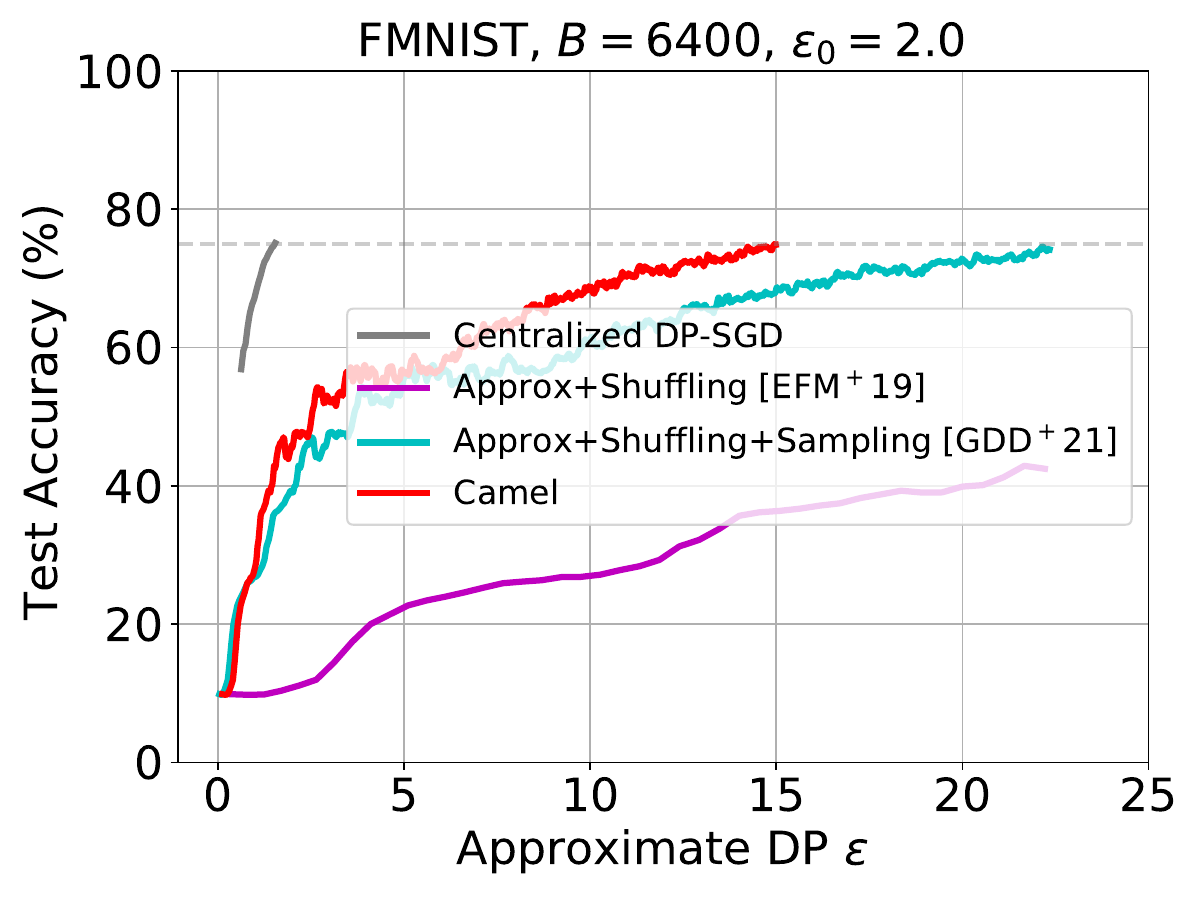}\\{(d)}
		\end{minipage}
		\\
		\begin{minipage}[t]{0.233\textwidth}
			\centering
			\includegraphics[width=\textwidth]{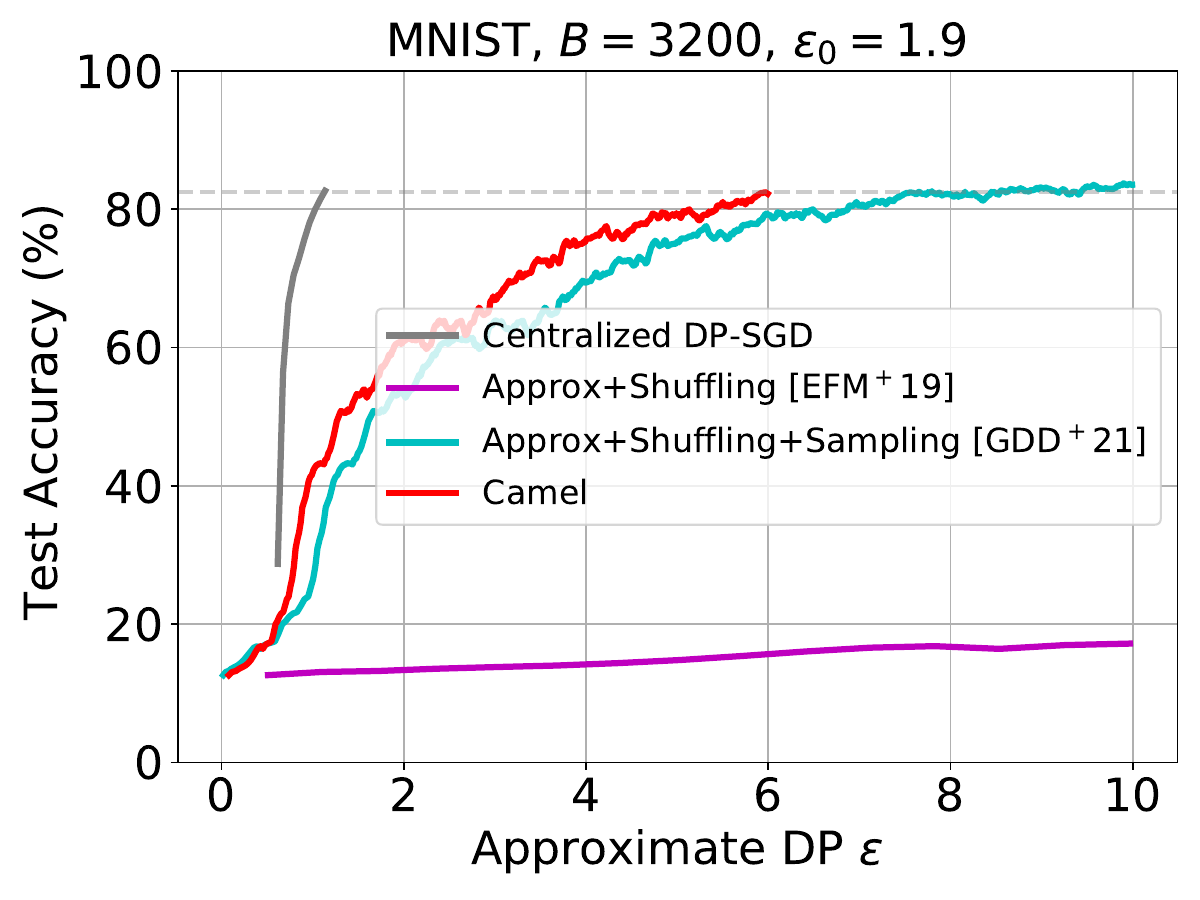}\\{(e)}
		\end{minipage}
		\hfill
		\begin{minipage}[t]{0.233\textwidth}
			\centering
			\includegraphics[width=\textwidth]{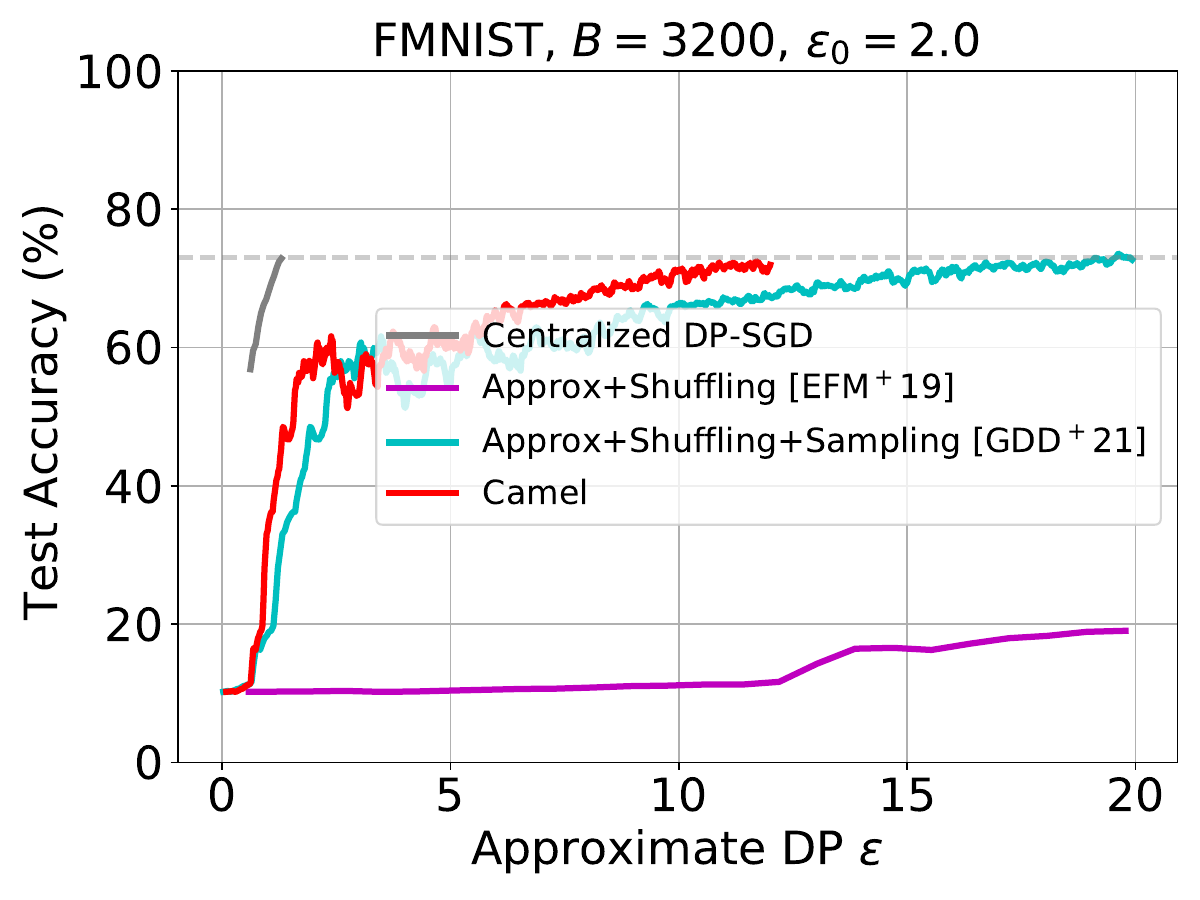}\\{(f)}
		\end{minipage}
		\caption{\revise{Privacy/utility trade-offs in comparison with existing methods on MNIST and FMNIST datasets. 
				For shuffle-model-based methods, $\varepsilon_0$ is set as 1.9 and 2.0 for the MNIST dataset and the FMNIST dataset, respectively.
				In Fig. \ref{fig:utility-vary-eps} (a) and (b), we fix the number of shuffled and sampled gradients $B$ as 12800 and vary the (overall) approximate DP $\varepsilon$ to test the utility of {\main} and the methods from \cite{AISTATS21,Erlingsson20}.  
				In Fig. \ref{fig:utility-vary-eps} (c) and (d), $B$ is fixed as 6400. 
				In Fig. \ref{fig:utility-vary-eps} (e) and (f), $B$ is set as 3200 for evaluation. 
				For Fig. \ref{fig:utility-vary-eps} (a)-(f), we also run the centralized DP-SGD method with the default parameters. 
		}}
		\label{fig:utility-vary-eps}
		\vspace{-10pt}
	\end{figure}
	
	\noindent \textbf{Implementation.} We implement a prototype system of {\main} in Python and Go. 
        Our code is open-source\footnote{https://github.com/Shuangqing-Xu/Camel}. 
	Specifically, we use Opacus\footnote{https://github.com/pytorch/opacus} with Pytorch to implement the differentially private model training protocol, while the server-side maliciously secure secret-shared shuffle protocol is implemented in Go. 
	For the finite field arithmetic associated with secret-shared shuffle, we follow \cite{NDSS22} to use the Goff library in \cite{Goff}. 
	We employ SHA256 for instantiating our hash function and AES in CTR mode for PRG. 
	All experiments are conducted on a workstation equipped with an Intel Xeon CPU boasting 64 cores running at 2.20GHz, 3 NVIDIA RTX A6000 GPUs, 256G RAM, and running Ubuntu 20.04.6 LTS.
	The client-server and inter-server communication is emulated by the loopback filesystem, where the delay of both communication is set to 40 ms, and the bandwidth is set to 100 Mbps to simulate a practical WAN setting. 
	
	\noindent \textbf{Datasets and Model Architectures.} We introduce two widely-used datasets and their corresponding model architectures involved in our experiments as follows: 
	\begin{itemize}
		\item \emph{MNIST\cite{LeCunBBH98}: MNIST contains images of 0-9 handwritten digits, which comprises a training set of $M=60,000$ examples and a test set of 10,000 examples. 
			We employ the same network as described in prior works \cite{Erlingsson20,AISTATS21}, outlined in Table \ref{tab:mnist_model}, with a total parameter count of $d=9594$.
		}
		
		\item \emph{FMNIST\cite{FMNIST}: FMNIST contains 60,000 training and 10,000 testing examples of Zalando's article images. Each is a 28x28 grayscale image labeled among 10 classes. 
			We utilize the widely-used ``2NN'' network architecture, as adopted in prior work \cite{fedavg}. 
			This architecture consists of two hidden layers, with each layer comprising 200 units and employing ReLU activations, leading to a total of $d=199,210$ parameters.}
		
	\end{itemize}
	\noindent The training data points are randomly shuffled and evenly distributed across $n$ clients.

	\noindent \revise{\textbf{Baselines.}} \revise{We compare the privacy amplification effect of our proposed {\main} with the work of \cite{Erlingsson20} (denoted as \textit{Approx+Shuffling}) that only considers privacy amplification by shuffling and the \textit{state-of-the-art} work \cite{AISTATS21} (denoted as \textit{Approx+Shuffling+Sampling}) on the shuffle model of DP in FL that considers composing amplification by shuffling with subsampling. 
		We compare {\main} with the aforementioned two baselines to demonstrate that {\main} can achieve a significantly better privacy-utility trade-off. 
		We also include a centralized DP-SGD method \cite{DPSGD} as a baseline. 
		It is noted that the non-private accuracy baselines using our predefined model architectures are 99\% and 89\% on MNIST and FMNIST datasets, respectively \cite{Erlingsson20}. 
		Furthermore, to illustrate the substantial communication-efficiency optimization achieved by {\main}, we also construct a baseline, denoted as {\baseline}, which directly secret-shares the noisy gradients and securely shuffles them without compression.
	}
	
	\noindent \textbf{Parameters.} For both the MNIST dataset and the FMNIST dataset, we fix the learning rate as $0.1$, $\ell_2$-norm bound $L=0.5$, and momentum as $0.5$. 
	In all our experiments, unless otherwise specified, we fix the number of clients $n=100$ and $\delta=10^{-5}$, smaller than $1/M$. 
	We vary $N$ to change the number of shuffled gradients and vary $B$ to change the number of shuffled and sampled gradients. 
	The optimal privacy parameter $\lambda$ in Theorem \ref{thm:analysis_overview} is obtained following the autodp library\footnote{\url{https://github.com/yuxiangw/autodp}}. 
	We consider individual parameters in a gradient as 32 bits and a 128-bit prime $p$.

	\begin{table}[!t]
		\centering
		\caption{Privacy/utility trade-offs of {\main} on the MNIST dataset and the FMNIST dataset. Here we vary the number of iterations $T$ to investigate the test accuracy (Acc) in \% and the overall approximate DP $\varepsilon$ (converted from RDP). }
		\label{tab:utility}
		\setlength\tabcolsep{3pt}
		
		\scalebox{0.9}{
			\begin{tabular}{cccccccccc}
				\hline \multirow{2}{*}{ Dataset }  & \multirow{2}{*}{$\varepsilon_0$}  & \multirow{2}{*}{$B$} &  \multicolumn{2}{c}{$T=500$} & \multicolumn{2}{c}{$T=1000$} & \multicolumn{2}{c}{$T=2000$} \\
				&   &   & Acc & $\varepsilon$ & Acc & $\varepsilon$ & Acc & $\varepsilon$  \\
				
				\hline \multirow{3}{*}{ MNIST } & \multirow{3}{*}{1.9} 
				& 3200 & 81.16\% & 5.84 & 84.42\% & 9.56 & 85.19\% & 15.92\\
				\cline{3-9} &  & 6400 & 83.31\% & 7.03 & 87.62\% & 11.40 & 89.51\% & 18.83 \\
				\cline{3-9} & & 12800 & \textbf{83.66\%} & 9.25 & \textbf{88.74\%} & 15.05 & \textbf{91.25\%} & 24.97 \\
				
				\hline \multirow{3}{*}{ FMNIST } & \multirow{3}{*}{2.0} 
				& 3200 & \textbf{64.01\%} & 7.23 & \textbf{72.76\%} & 12.02 & 76.44\% & 20.41 \\
				\cline{3-9} & & 6400 & 63.72\% & 8.36 & 72.57\% & 13.69 & \textbf{76.45\%} & 22.88 \\
				\cline{3-9} & & 12800 & 63.91\% & 10.80 & 70.12\% & 17.71 & 75.9\% & 29.68 \\
				\hline
			\end{tabular}
		}
		\vspace{-10pt}
	\end{table}
	
	\begin{table*}[!t]
		\centering
		\caption{Training cost of {\main} on the MNIST dataset and the FMNIST dataset per iteration.}
		\label{tab:effi_real_datasets}
		\setlength\tabcolsep{3pt}
		
		\scalebox{0.8}{
			\begin{tabular}{cccccccc}
				\hline \multirow{3}{*}{ Method }  & \multirow{3}{*}{ $N$ } & \multirow{3}{*}{ Dataset } & \multirow{3}{*}{\begin{tabular}{c} 
						Offline Comm.  \\
						Cost (MB)
				\end{tabular}}  & \multicolumn{4}{c}{Online Training Cost} \\
				
				&  &  &  & \multirow{2}{*}{\begin{tabular}{c} 
						Per-Client  \\
						Comp. Cost (s)
				\end{tabular}}  
				& \multirow{2}{*}{\begin{tabular}{c} 
						Server-Side  \\
						Comp. Cost (s)
				\end{tabular}}  & \multirow{2}{*}{\begin{tabular}{c} 
						Online Comm.  \\
						Cost (MB)
				\end{tabular}} & \multirow{2}{*}{\begin{tabular}{c} 
						Server-Side  \\
						Overall Runtime (s)
				\end{tabular}} \\
				\\
				\hline \multirow{2}{*}{ {\main} } & \multirow{2}{*}{ 400 } & { MNIST } & 0.031 & 0.040 & 0.019 & 0.213 & \textbf{0.796}\\
				& & { FMNIST } & 0.031 & 0.049 & 0.112 & 0.213 & \textbf{0.889} \\
				
				\multirow{2}{*}{ {\baseline} } & \multirow{2}{*}{ 400 } & { MNIST } & 29.279 & 0.066 & 2.501 & 205.096 & 19.669\\
				& & { FMNIST } & 607.965  & 0.489 & 34.963 & 4258.795 & 376.427 \\

				\hline \multirow{2}{*}{ {\main} } & \multirow{2}{*}{ 800 } & { MNIST } & 0.061 & 0.043 & 0.037 & 0.427 & \textbf{0.831}\\
				& & { FMNIST } & 0.061 & 0.073 & 0.184 & 0.427 & \textbf{0.978} \\
				
				\multirow{2}{*}{ {\baseline} } & \multirow{2}{*}{ 800 } & { MNIST } & 58.557 & 0.115 & 4.319 & 410.046 & 37.882\\
				& & { FMNIST } & 1215.930 & 0.778 & 64.953 & 8514.551 & 746.877 \\

				\hline \multirow{2}{*}{ {\main} } & \multirow{2}{*}{ 1600 } & { MNIST } & 0.122 & 0.041 & 0.065 & 0.854 & \textbf{0.893}\\
				& & { FMNIST } & 0.122 & 0.127 & 0.466 & 0.854 & \textbf{1.294} \\
				
				\multirow{2}{*}{ {\baseline} } & \multirow{2}{*}{ 1600 } & { MNIST } & 117.114 & 0.143 & 8.032 & 819.946 & 74.388 \\
				& & { FMNIST } & 2431.860  & 1.300 & 120.713 & 17026.062 & 1483.558 \\

				\hline \multirow{2}{*}{ {\main} } & \multirow{2}{*}{ 3200 } & { MNIST } & 0.244 & 0.073 & 0.171 & 1.709 & \textbf{1.068}\\
				& & { FMNIST } & 0.244 & 0.234 & 0.947 & 1.709 & \textbf{1.844} \\
				
				\multirow{2}{*}{ {\baseline} } & \multirow{2}{*}{ 3200 } & { MNIST } & 234.229 & 0.276 & 15.383 & 1639.746 & 147.323 \\
				& & { FMNIST } & 4863.721 & 2.808 & 239.109 & 34049.084 & 2963.796\\
				\hline
			\end{tabular}
		}
		\vspace{-5pt}
	\end{table*}

	\subsection{Numerical Results of Privacy Amplification}
	\label{sec:exp:numerical}
	
	In this section we provide the numerical results of privacy amplification by comparing the bounds of {\main} with the approximate DP bounds from closely related works \cite{Erlingsson20,AISTATS21}. 
	Specifically, the work of \cite{Erlingsson20} only uses privacy amplification by shuffling, without considering privacy amplification by subsampling. 
	The work of \cite{AISTATS21} composes amplification by shuffling with amplification by subsampling to achieve tighter bounds. 
	In Fig. \ref{fig:numerical_bounds}, we vary the LDP level $\varepsilon_0$, number of (shuffled) messages $N$, subsampling rate $\gamma$, and number of iterations $T$, and plot the bounds on approximate ($\varepsilon,\delta$)-DP of different methods for fixed $\delta=10^{-5}$. 
	We can observe that the work of \cite{AISTATS21} has already improved the work of \cite{Erlingsson20} by deriving tighter bounds, and the improvement is significantly impacted by the sampling rate $\gamma$. 
	Notably, when $\gamma=1$, the bound of \cite{AISTATS21} is identical to the bound of \cite{Erlingsson20}. 
	We can also find that the bounds derived in \cite{AISTATS21} are always looser than the bounds derived in our work by analyzing RDP. 
	In other words, our work provides the tightest bound in all cases. 
	This demonstrates the advantages of analyzing the RDP of multiple iterations.

	\begin{figure}[t!]
		
		\centering
		
		\begin{minipage}[t]{0.4\linewidth}
			\centering
			\includegraphics[width=\linewidth]{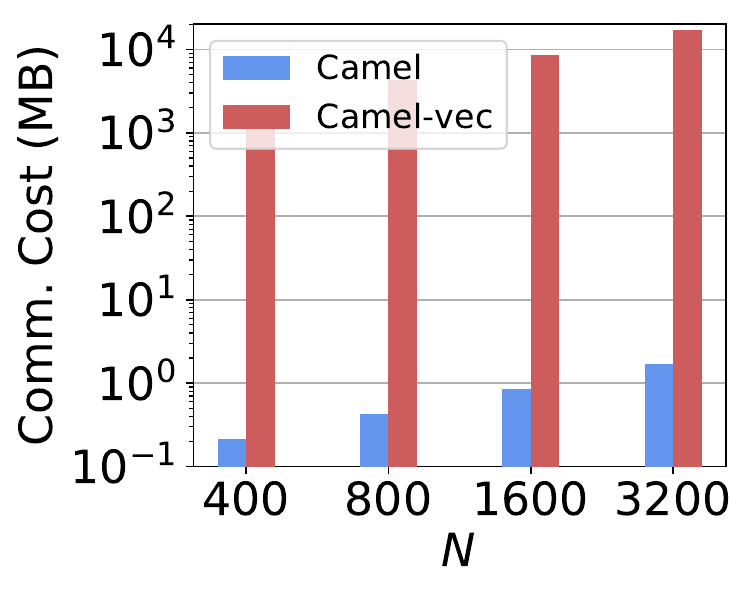}\\(a)
		\end{minipage}
		\begin{minipage}[t]{0.4\linewidth}
			\centering
			\includegraphics[width=\linewidth]{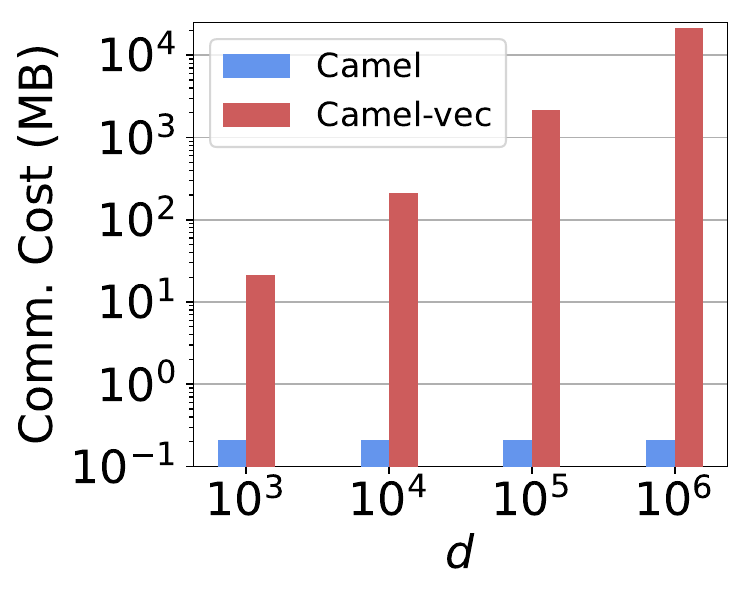}\\(b)
		\end{minipage}
		
		\caption{Comparison of communication cost per shuffle: (a) communication cost by varying the number of noisy gradients (to shuffle) $N\in\{400,800,1600,3200\}$ and fixing $d=10^5$, and (b) communication cost by varying $d\in\{10^3,10^4,10^5,10^6\}$ and fixing $N=400$. }
		\label{fig:effi_comm_cost}
		\vspace{-15pt}
		
	\end{figure}
	
	\begin{figure}[t!]
		
		\centering
		
		\begin{minipage}[t]{0.4\linewidth}
			\centering
			\includegraphics[width=\linewidth]{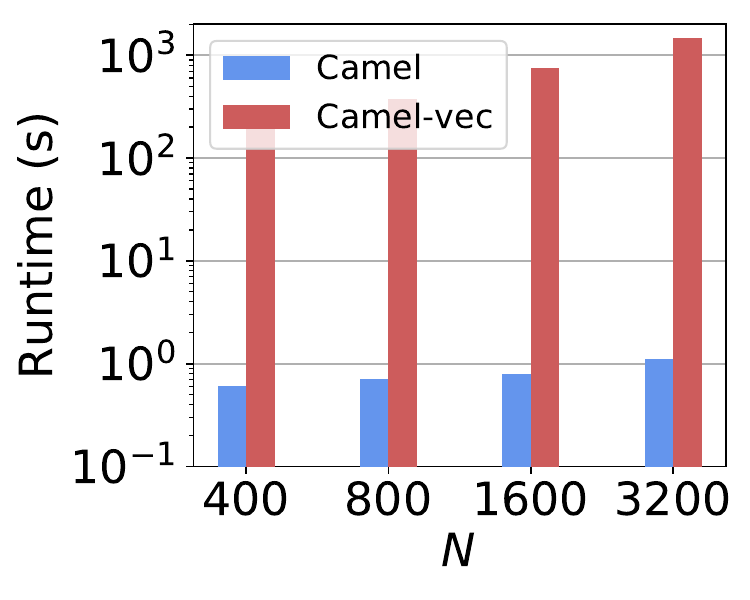}\\(a)
		\end{minipage}
		\begin{minipage}[t]{0.4\linewidth}
			\centering
			\includegraphics[width=\linewidth]{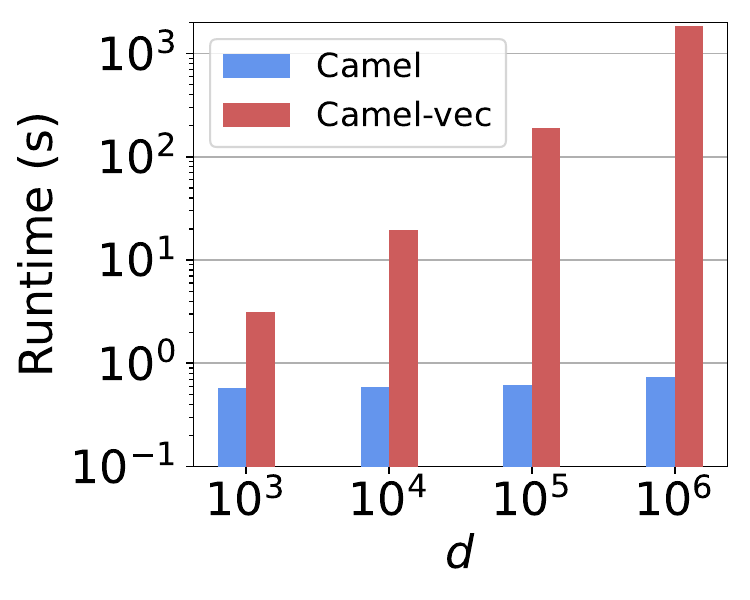}\\(b)
		\end{minipage}
		\caption{Comparison of server-side overall runtime per shuffle: (a) runtime by varying the number of noisy gradients (to shuffle) $N\in\{400,800,1600,3200\}$ and fixing $d=10^5$, and (b) runtime by varying $d\in\{10^3,10^4,10^5,10^6\}$ and fixing $N=400$. }
		\label{fig:effi_run_time}
		\vspace{-15pt}
		
	\end{figure}
	
	\subsection{Utility}
	\label{sec:exp:utility}
	
	In this section, we first evaluate the utility performance of our proposed {\main}, and further investigate its privacy-utility trade-offs by varying different parameters and comparing it with baselines. 
	For all experiments in this section, we set $\varepsilon_0=1.9$ for the MNIST dataset and $\varepsilon_0=2.0$ for the FMNIST dataset. 
	Firstly, we vary $T\in[2000]$ and $B\in\{3200,6400,12800\}$ to plot the test accuracy in \%. 
	From Fig. \ref{fig:utility_vary_T} we can observe that (1) despite the locally injected noise, {\main} finally converges on all datasets, and that (2) the number of shuffled and sampled gradients $B$ impacts the model utility. 
	Besides, we also care about the privacy-utility trade-offs of {\main}. 
	In Table \ref{tab:utility}, we present the results of model accuracy and the corresponding approximate DP $\varepsilon$ on MNIST and FMNIST datasets by varying $T\in\{500,1000,2000\}$ and $B\in\{3200,6400,12800\}$. 
	Note that changing $B$ results in different values of $\gamma$ and privacy-utility trade-offs. 
	It is also evident that a larger $B$ gives better model utility on the MNIST dataset but is not necessary for the FMNIST dataset. 
	Therefore, it is important to select $B$ to give the optimal privacy-utility trade-offs. 
	We can further observe that with a relatively small budget $\varepsilon=15.05$, {\main} achieves an accuracy of 88.74\% on the MNIST dataset. 
	Generally, {\main} gives satisfactory model utility under reasonable privacy budgets. 
	
	\revise{
		Furthermore, we compare {\main} with two baselines to show that {\main} achieves a significantly better privacy-utility trade-off than prior related work on FL in the shuffle model of DP. 
		We also include the privacy-utility trade-off of the centralized DP-SGD baseline by fixing the learning rate as $0.1$, clipping bound as $0.5$, batch size as $500$, noise multiplier as $1.3$ for both datasets, and varying the approximate DP $\varepsilon$. 
		Fig. \ref{fig:utility-vary-eps} shows our evaluation results on the MNIST dataset and FMNIST dataset, where the results of shuffle-model-based methods are obtained by fixing the LDP level $\varepsilon_0$ for each FL iteration and varying the overall approximate DP $\varepsilon$. 
		It can be observed that given the same $\varepsilon$, {\main} yields the best model utility in all cases compared to previous methods in the shuffle model of DP \cite{AISTATS21,Erlingsson20}. 
		This is because our work derives tighter bounds on $\varepsilon$ over \cite{AISTATS21,Erlingsson20}, as evidenced in Section \ref{sec:exp:numerical}, thereby yielding the improved privacy-utility trade-off results. 
		It is important to note that although the centralized DP-SGD baseline \cite{DPSGD} can yield the same accuracy with smaller $\varepsilon$, it needs access to raw training datasets for centralized processing and does not protect client privacy during the training process.  
	}

	\subsection{Efficiency}
	\label{sec:exp:efficiency}

	\begin{table*}[!t]
		\centering
		\caption{\revise{Training cost comparison of {\main} and KLS21 \cite{kairouz21} on the MNIST dataset and the FMNIST dataset per iteration.}}
		\label{tab:effi_cmp_secagg}
		\setlength\tabcolsep{3pt}
		\scalebox{0.8}{\revise{
				\begin{tabular}{cccccc}
					\hline  \multirow{2}{*}{ Method }  & \multirow{2}{*}{ Dataset } &  \multirow{2}{*}{\begin{tabular}{c} 
							Per-Client   \\
							Comp. Cost (s)
					\end{tabular}} & \multirow{2}{*}{\begin{tabular}{c} 
							Server-Side \\
							Comp. Cost (s)
					\end{tabular}}  &  \multirow{2}{*}{\begin{tabular}{c} 
							Per-Client   \\
							Comm. Cost (KB)
					\end{tabular}} &  \multirow{2}{*}{\begin{tabular}{c} 
							Server-Side \\
							Overall Runtime (s)
					\end{tabular}} \\
					\\
					\hline \multirow{2}{*}{ {\main} } & { MNIST } & 0.040 & 0.019 & 0.250 & 0.796\\
					& { FMNIST }  & 0.049 & 0.112 & 0.250 & 0.889 \\ 
					\multirow{2}{*}{ {KLS21 \cite{kairouz21}} } & { MNIST } & 0.185 & 0.060 & 55.30 & 0.220 \\ 
					& { FMNIST }  & 0.572 & 0.554 &  795.6 & 0.776 \\ 
					\hline
		\end{tabular}}}
		\vspace{-5pt}
	\end{table*}

	We start with comparing the training cost of {\main} with the baseline {\baseline} which does not consider gradient compression on the MNIST dataset and the FMNIST dataset. 
	Table \ref{tab:effi_real_datasets} presents the evaluation results, where the offline communication costs include the data transmission of necessary materials required for secure shuffling.  
	The online communication costs is system-wide, which include both client-server and inter-server data transmission throughout a training iteration. 
	The computation cost per client includes the time for computing, perturbing, compressing, and MACing gradients. 
	The server-side computation cost includes the overall computation time for performing maliciously secure secret-shared shuffle and gradient subsampling, decompression, and aggregation. 
	The server-side overall runtime per iteration includes server-side computation time, network latency, and data transmission time. 
	From Table \ref{tab:effi_real_datasets}  we can get the following observations. 
	Firstly, compared to {\baseline}, {\main} significantly reduces both offline and online communication costs. Specifically, {\main} reduces online communication costs by 965$\times$ and 20,029$\times$ over {\baseline} on the MNIST dataset and the FMNIST datasets for $N=3200$, respectively.
	This demonstrates the effectiveness of our noisy gradient compression mechanism in cutting both client-server and inter-server communication costs. 
	Secondly, the training cost of {\main} per iteration is also thousands of times lower than {\baseline}, for example, 1,607$\times$ lower on the FMNIST dataset when $N=3200$. 
	Notably, the communication cost dominates the training cost because our secure shuffling protocol is based on secret sharing. 
	Furthermore, it is noticeable that the server-side computation cost in {\main} are lower than those in {\baseline}.
	This is attributed to the fact that in {\main}, the noisy gradients intended for shuffling are compressed prior to secure shuffling. 
	The compression not only saves communication costs, but also results in reduced computation costs compared to the direct manipulation of high-dimensional vectors in {\baseline}.
	
	To further investigate the impact of the number of noisy gradients (to shuffle) $N$ and the gradient vector dimension $d$ on training efficiency, we vary $N$ and $d$ to evaluate communication costs (including both client-server and inter-server data transmission) and server-side runtime of performing a secret-shared shuffle (including computation time, network latency, and data transmission time). 
	%
	%
	Fig. \ref{fig:effi_comm_cost} (a) and Fig. \ref{fig:effi_run_time} (a) present the evaluation results by varying $N\in\{400,800,1600,3200\}$ and fixing $d=10^5$, from which we can observe that both {\main} and {\baseline} experience increases in communication costs and runtime with $N$, but the former exhibits significantly lower costs compared to the latter. 
	Moreover, from Fig. \ref{fig:effi_comm_cost} (b) and Fig. \ref{fig:effi_run_time} (b), where we vary $d\in\{10^3,10^4,10^5,10^6\}$ and fix $N=400$, it becomes evident that $d$ has a substantial impact on {\baseline}'s communication costs and runtime while exerting minimal influence on {\main}. 
	This is because in {\main}, noisy gradients are compressed into fixed-length messages, each consisting of a 128-bit seed and a 1-bit sign in our experiment.

	\begin{figure}[t!]
		\centering
		\begin{minipage}[t]{0.32\linewidth}
			\centering
			\includegraphics[width=\linewidth]{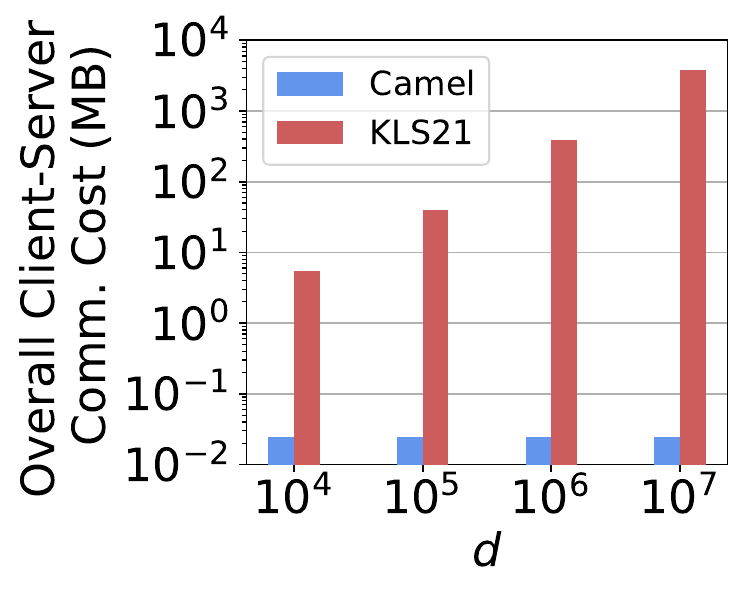}\\ \:\:\:\:\:\:\:\:\:(a)
		\end{minipage}
		\begin{minipage}[t]{0.32\linewidth}
			\centering
			\includegraphics[width=\linewidth]{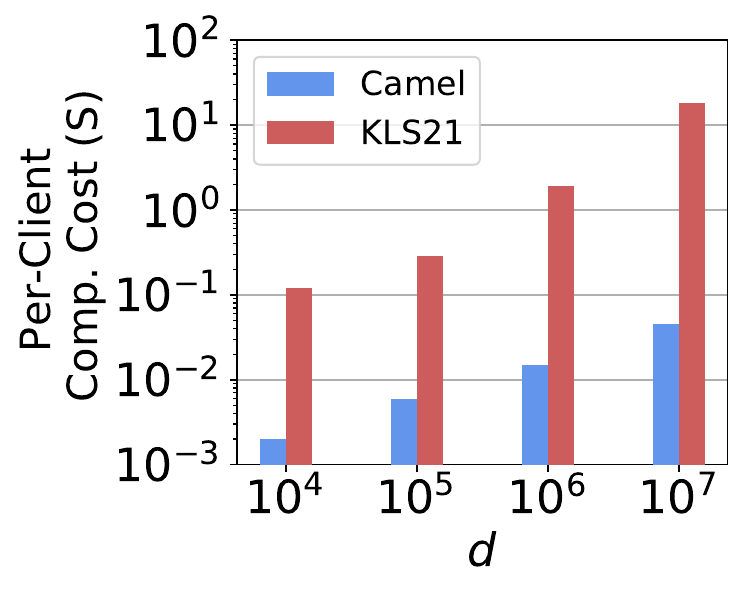}\\ \:\:\:\:\:\:\:\:\:(b)
		\end{minipage}
		\begin{minipage}[t]{0.32\linewidth}
			\centering
			\includegraphics[width=\linewidth]{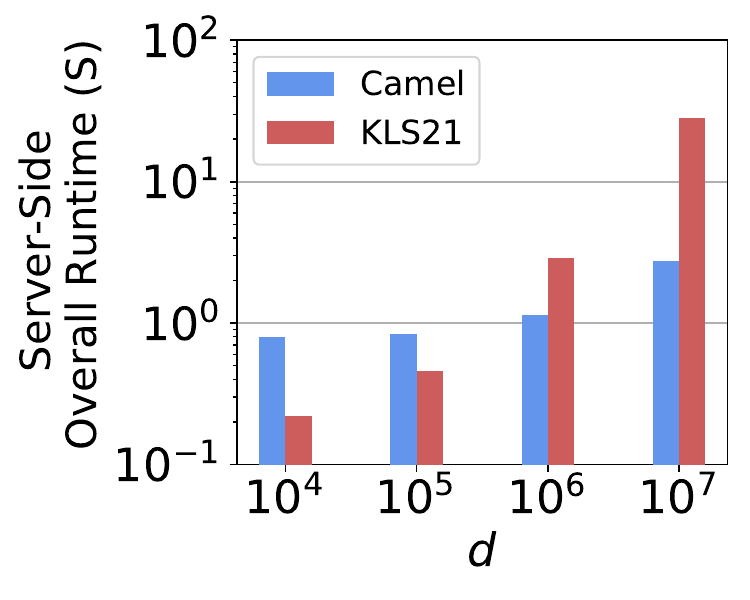}\\ \:\:\:\:\:\:\:\:\:(c)
		\end{minipage}
		
		\caption{\revise{Training cost comparison of {\main} and KLS21 \cite{kairouz21} per iteration by varying $d\in\{10^4,10^5,10^6,10^7\}$ and fixing $K=N=400, n=100$: (a) overall client-server communication cost comparison, and (b) per-client computation cost comparison, and (c) server-side overall runtime (including server-side computation time, network latency, and data transmission time) comparison. }}
		\label{fig:effi_secagg_vary_d}
		\vspace{-15pt}
	\end{figure}

	\subsection{Efficiency Comparison with Approach Combining Secure Aggregation and DP}
	\label{sec:exp:comparison-secagg}
	
	\revise{
		We now compare {\main} with the orthogonal secure-aggregation-based approach, which could also provide the output model with DP guarantee. 
		Recall that secure aggregation would restrict the use of gradient compression (as employed in {\main}) for communication efficiency optimization because decompression is required before aggregation. 
		Therefore, we focus on demonstrating the advantage of {\main} in client-side efficiency. 
		Specifically, we compare {\main} (with malicious security) with the work of \cite{kairouz21} (denoted as KLS21)\footnote{Code from \url{https://github.com/google-research/federated/tree/master/distributed_dp}} that considers a semi-honest server. 
		For a fair comparison on the MNIST and FMNIST datasets, we use the same experimental setup as in {\main}. 
		Specifically, we adopt the same model architecture, set each gradient parameter to $32$ bits, fix the number of clients at $n=100$, and use a batch size of $K=400$ gradients per iteration (for {\main}, we set the number of gradients to shuffle as $N=K$).
		Since our {\main} employs the mini-batch SGD algorithm as the underlying learning algorithm, which is also commonly seen in other FL works \cite{AISTATS21, NeurIPS21, liu2020fedsel}, we implement KLS21 with the same learning algorithm, instead of the inherently different FedAvg algorithm \cite{fedavg} originally considered in KLS21. 
		For this reason, KLS21 cannot be directly included for a privacy-utility comparison with {\main}. 
		However, it is still meaningful to investigate and compare the communication and computation efficiency of KLS21 and {\main}. 
		For efficiency comparison, we use the widely popular protocol\footnote{Code from \url{https://github.com/55199789/PracSecure}} by Bonawitz et al. \cite{BonawitzIKMMPRS17} to connect the clients and the server in KLS21. 
		Under such setting, each client in KLS21 would compute and clip $K/n$ gradients, locally aggregate these gradients, process the aggregated gradient (discretize it and perturb it with discrete Gaussian noise), and then send it to the server for secure aggregation. 
	}
	
	\begin{figure}[t!]
		\centering
		\begin{minipage}[t]{0.32\linewidth}
			\centering
			\includegraphics[width=\linewidth]{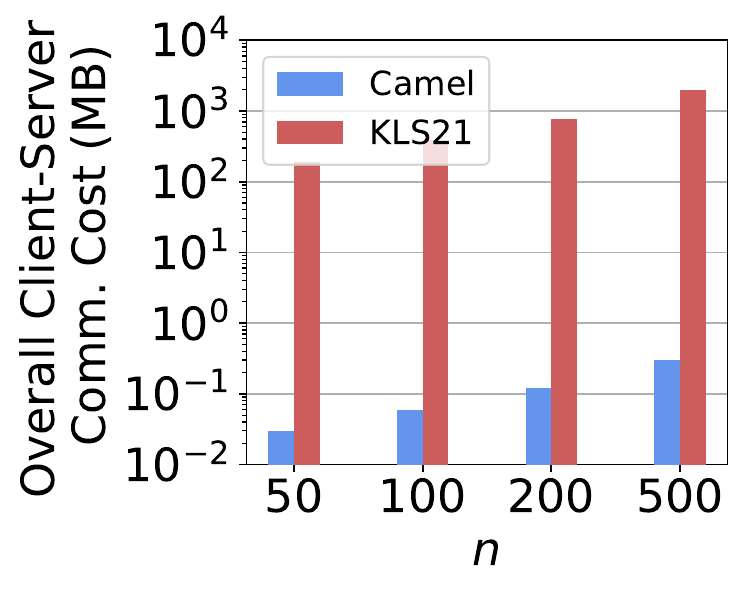}\\ \:\:\:\:\:\:\:\:\:(a)
		\end{minipage}
		\begin{minipage}[t]{0.32\linewidth}
			\centering
			\includegraphics[width=\linewidth]{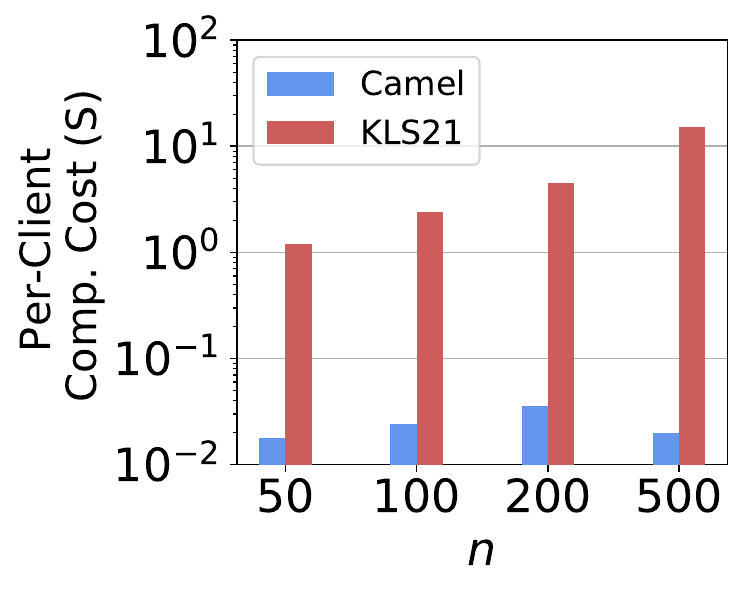}\\ \:\:\:\:\:\:\:\:\:(b)
		\end{minipage}
		\begin{minipage}[t]{0.32\linewidth}
			\centering
			\includegraphics[width=\linewidth]{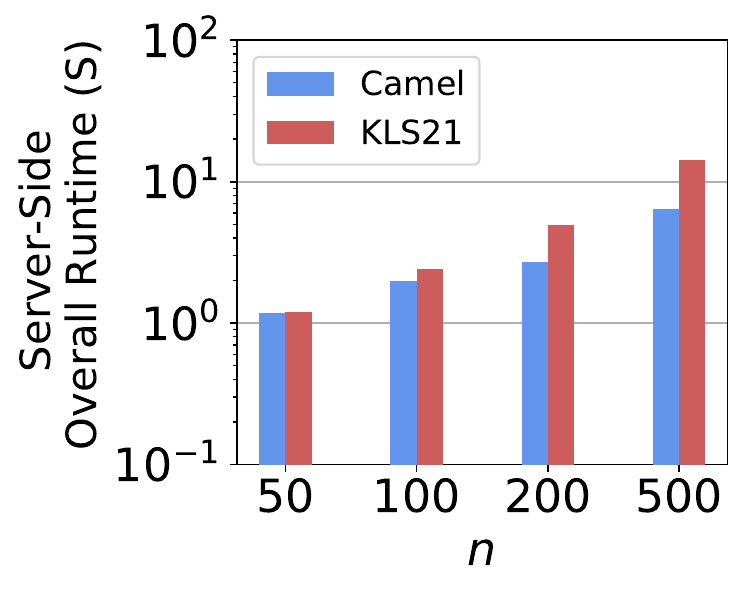}\\ \:\:\:\:\:\:\:\:\:(c)
		\end{minipage}
		
		\caption{\revise{Training cost comparison of {\main} and KLS21 \cite{kairouz21} per iteration by varying $n\in\{50,100,200,500\}$ and fixing $N/n=K/n=10,d=10^6$: (a) overall client-server communication cost comparison, and (b) per-client computation cost comparison, and (c) server-side overall runtime (including server-side computation time, network latency, and data transmission time) comparison. }}
		\label{fig:effi_secagg_vary_n}
		\vspace{-15pt}
	\end{figure}
	\revise{
		Table \ref{tab:effi_cmp_secagg} presents the training costs of {\main} and KLS21 on MNIST and FMNIST datasets at each iteration. 
		The per-client computation time includes the time required for computing and processing gradients.
		The server-side overall runtime consists of server-side computation time, network latency, and data transmission time. 
		For a fair comparison, we assume that in both {\main} and KLS21 all clients are synchronized, and the data transmission time involves only the time required to transmit a single client's data to the server. 
		We observe that the per-client computation cost, the per-client communication cost, and the server-side computation cost are significantly lower than those of the baseline KLS21. 
		%
		%
		This is attributed to the complex pairwise masking strategy in secure aggregation \cite{BonawitzIKMMPRS17}. 
		%
	}
	
	\revise{
		In contrast, due to the application of gradient compression, each client in {\main} only needs to send compressed gradients (each with a fixed size regardless of $d$) to the server, resulting in a smaller and unchanged per-client communication cost when transitioning from the MNIST dataset to the FMNIST dataset. 
		It is also noted that the server-side overall runtime of {\main} is comparable to that of KLS21. 
		Notably, the gap between KLS21 and {\main} decreases when shifting from the MNIST dataset to the FMNIST dataset. 
		This is because {\main}'s server-side overall runtime is largely impacted by the network latency (0.76 seconds of network latancy out of 0.796 seconds of server-side overall runtime on MNIST dataset). 
		Although network latency has little impact on the server-side overall runtime of KLS21, it is greatly influenced by the gradient vector dimension $d$.  
		This indicates a turning point where, as $d$ increases, the server-side overall runtime of {\main} becomes lower than that of KLS21. 
	}
	
	\revise{
		To identify this turning point and further investigate system efficiency concerning $d$, we conduct experiments on a synthetic dataset by varying $d$ and fixing $K=N=400, n=100$. 
		We evaluate the overall client-server communication cost (Fig. \ref{fig:effi_secagg_vary_d} (a)), per-client computation cost (Fig. \ref{fig:effi_secagg_vary_d} (b)), and server-side overall runtime (Fig. \ref{fig:effi_secagg_vary_d} (c)) of {\main} and KLS21. 
		Here, the per-client computation time consists of the time required for processing the gradients. 
		From Fig. \ref{fig:effi_secagg_vary_d}, we observe that the overall client-server communication cost and per-client computation cost of {\main} are always lower than those of KLS21. 
		Notably, the client-server communication cost for KLS21 increases linearly with $d$, while it remains stable for {\main}. 
		When $d\geq10^6$, the server-side computation cost for KLS21 surpasses that of {\main}.
		This suggests a turning point for $d$ between $10^5$ and $10^6$, where {\main} starts to demonstrate better server-side overall runtime performance than KLS21. 
		This indicates that {\main} is well-suited for handling large-scale models.
	}
	
	\revise{
		We also note that the communication and computation efficiency of secure-aggregation-based methods are significantly impacted by the number of clients $n$ \cite{BonawitzIKMMPRS17}. 
		Therefore, we conduct experiments on a synthetic dataset by varying $n\in\{50,100,200,500\}$ and fix $N/n=K/n=10,d=10^6$ (i.e., we consider each client locally process $10$ gradients regardless of $n$). 
		We evaluate the overall client-server communication cost (Fig. \ref{fig:effi_secagg_vary_n} (a)), per-client computation cost (Fig. \ref{fig:effi_secagg_vary_n} (b)), and server-side overall runtime (Fig. \ref{fig:effi_secagg_vary_n} (c)) of {\main} and KLS21. 
		From Fig. \ref{fig:effi_secagg_vary_n} we can observe a similar trend as in Fig. \ref{fig:effi_secagg_vary_d}, i.e., the overall client-server communication cost and per-client computation cost of {\main} are always lower than those of KLS21. 
		Notably, the per-client computation cost for processing gradients of {\main} is very small, at around 0.03 seconds.
		It is also observed that when $n$ is small, the server-side overall runtime of {\main} and KLS21 is comparable, while {\main} is better than KLS21 with the increase of $n$. 
		This indicates that {\main} is better suited for the more practical scenario where a large number of clients are involved. 
	}
	
	\vspace{-12pt}
	
	\section{Discussion}
	\label{sec:discussion}
	
	\noindent \revise{\textbf{Comparison with existing maliciously secure FL frameworks.} 
		While other maliciously secure FL frameworks exist, most of them (like ELSA \cite{ELSA}) do not provide DP guarantees with good utility. 
		Prior works \cite{kairouz21,AgarwalKL21} that rely on secure aggregation and provide DP guarantees could be extended to defend against a malicious server by using an extended version of the underlying secure aggregation technique that supports verifiability (e.g., \cite{HahnKKH23}).
		However, recall that the use of secure aggregation hinders system-wide communication efficiency optimization for FL (e.g., through gradient compression as in Camel, which requires decompression before aggregation).
		Also, as indicated by the experimental results in Section \ref{sec:exp:comparison-secagg}, our system {\main}, even with malicious security, has already achieved a significant advantage in client-side efficiency over the secure aggregation approach-based work \cite{kairouz21} that considers an honest-but-curious adversary.
		Such efficiency advantage will be further amplified when compared to the aforementioned extended approaches that use a verifiable version of the secure aggregation with increased costs. 
	}
	
	\noindent \revise{\textbf{Multi-message shuffle model of DP for FL.} 
		In this paper, we have followed most existing works \cite{AISTATS21,NeurIPS21,liu2021flame} to apply the single-message shuffle model of DP in FL. 
		We also notice that there is a recent trend of applying the multi-message shuffle model of DP \cite{BalleBGN20,GhaziG0PV21} to get better privacy-utility trade-offs over the single-message shuffle model of DP. 
		However, multi-message shuffling has rarely been applied in FL so far. 
		Existing multi-message shuffling works like \cite{BalleBGN20,GhaziG0PV21} focus on the general problem of private summation and do not specifically target FL where communication efficiency is essential (besides privacy and utility) and cannot be overlooked. 
		To our best knowledge, there is no known practical application of multi-message shuffling in FL without sacrificing communication efficiency compared to single-message shuffling. 
	}
	
	\noindent \revise{\textbf{Extending design beyond the three-server model.} 
		Recall that {\main} is designed and built in the three-server distributed trust setting. 
		We note that it is possible to extend {\main} to a $k$-server setting ($k>3$). 
		The current three-server secret-shared shuffle protocol essentially applies in a secure manner a composition of permutations separately held by two servers. 
		Thus, to extend to more servers, a direction is to use a pairwise processing strategy, where each server holding a permutation interacts with all other servers to get the permutation securely applied and this process repeats until all permutations have been applied. 
		With such strategy as a basis, how to properly adapt the integrity checks should be further explored. 
		%
	}
	
	\section{Conclusion}
	\label{sec:conclusion}
	This paper presents {\main}, a new communication-efficient and maliciously secure FL framework in the shuffle model of DP. 
	{\main} first departs from prior works by ambitiously supporting integrity check for the shuffle computation, achieving security against malicious adversary. 
	In particular, {\main}'s design leverages an emerging distributed trust setting and a trending cryptographic primitive of secret-shared shuffle, with custom techniques developed to improve system-wide communication efficiency and harden the security of server-side computation. 
	Furthermore, through analyzing the RDP of the overall FL process, we also derive a tighter bound compared to existing approximate DP bounds. 
	We conduct extensive experiments over two real-world datasets to demonstrate {\main}'s better privacy-utility trade-off and promising performance.
	
	\section*{Acknowledgments}
	
	We sincerely thank the shepherd and the anonymous reviewers for their constructive and invaluable feedback. 
	This work was supported in part by the Guangdong Basic and Applied Basic Research Foundation under Grant 2023A1515010714 and Grant 2024A1515012299, by the National Natural Science Foundation of China under Grant 62071142, and by the Shenzhen Science and Technology Program under Grant JCYJ20220531095416037 and Grant JCYJ20230807094411024.
	
	\bibliographystyle{ACM-Reference-Format}
	\bibliography{reference}
	
	\appendix
	
	\section{Proof of Lemma \ref{lem:NoisyGradCmpr}}
	\label{appendix:lem_NoisyGradCmpr_proof}
	\textit{Proof}. 
	Firstly, we prove that our proposed mechanism \textsf{NoisyGradCmpr} compresses gradients without incurring additional utility loss compared to DJW18 (denoted as $\mathsf{M}(\cdot)$). 
	Recall in \textsf{NoisyGradCmpr} (Algorithm \ref{algo:NoisyGradCmpr}), a $d$-dimensional random vector $\boldsymbol{v}$ is generated at the client-side using a random seed $s$ and a PRG. 
	Given that the server shares the same PRG and receives seed from the client, it can reproduce the same vector $\boldsymbol{v}$. 
	Therefore, employing a PRG for gradient compression is lossless, i.e., we have $\mathbb{E}[\textsf{R}(\boldsymbol{x})]=\mathbb{E}[\textsf{M}(\boldsymbol{x})]$  and $\mathbb{E}\|\textsf{R}(\boldsymbol{x}) - \textsf{M}(\boldsymbol{x})\|^2_2 = 0$. 
	Due to the lossless compression property, we only need to prove that the LDP perturbation mechanism applied in Algorithm \ref{algo:NoisyGradCmpr} and \ref{algo:NoisyGradDcmp} is unbiased, guarantees $\varepsilon_0$-LDP, and has bounded variance. 
	Note that it has been shown that $\textsf{M}(\cdot)$ is an $\varepsilon_0$-LDP mechanism \cite{duchi2018minimax}. 
	Since our proposed \textsf{NoisyGradCmpr} and \textsf{NoisyGradDcmp} only use seed and a PRG to replace the randomly generated vector $\boldsymbol{v}$ in $\textsf{M}(\cdot)$, our mechanism maintains the same unbiased property with $\textsf{M}(\cdot)$, where $\mathbb{E}[\textsf{M}(\boldsymbol{x})] = \boldsymbol{x}$. 
	Therefore, the variance of our mechanism is bounded as follows: 
	\begin{equation}
		\notag
		\begin{aligned}
			& \:\:\mathbb{E}\|\textsf{R}(\boldsymbol{x}) - \boldsymbol{x}\|^2_2 \\
			& = \mathbb{E}\|\textsf{R}(\boldsymbol{x}) - \textsf{M}(\boldsymbol{x}) + \textsf{M}(\boldsymbol{x}) - \boldsymbol{x}\|^2_2 \\
			& \stackrel{\makebox[0pt]{\mbox{\normalfont\footnotesize (a)}}}{=} \mathbb{E}\|\textsf{R}(\boldsymbol{x}) - \textsf{M}(\boldsymbol{x})\|^2_2  + \mathbb{E}\|\textsf{M}(\boldsymbol{x}) - \boldsymbol{x}\|^2_2 \\
			& \stackrel{\makebox[0pt]{\mbox{\normalfont\footnotesize (b)}}}{=} \mathbb{E}\|\textsf{M}(\boldsymbol{x}) - \boldsymbol{x}\|^2_2 \\
			& \stackrel{\makebox[0pt]{\mbox{\normalfont\footnotesize (c)}}}{\leq}  \left(L\frac{3 \sqrt{\pi}\sqrt{d}}{4} \frac{e^{\varepsilon_0}+1}{e^{\varepsilon_0}-1}\right)^2,  
		\end{aligned}
	\end{equation}
	
	\noindent where step (a) utilizes the property that $\mathbb{E}[\textsf{R}(\boldsymbol{x})]=\mathbb{E}[\textsf{M}(\boldsymbol{x})]$ and \textsf{M} is unbiased, implying that the cross-multiplication term is zero \cite{AISTATS21}. 
	Step (b) follows because $\mathbb{E}\|\textsf{R}(\boldsymbol{x}) - \textsf{M}(\boldsymbol{x})\|^2_2 = 0$. 
	Step (c) utilizes the property that $\mathbb{E}\|\textsf{M}(\boldsymbol{x}) - \boldsymbol{x}\|^2_2 \leq \mathbb{E}\|\textsf{M}(\boldsymbol{x})\|^2_2 \leq \left(L\frac{3 \sqrt{\pi}\sqrt{d}}{4} \frac{e^{\varepsilon_0}+1}{e^{\varepsilon_0}-1}\right)^2$ \cite{duchi2018minimax}.

	\begin{algorithm*}[!t]
		\caption{The Complete Protocol of Our Proposed Maliciously Secure Secret-Shared Shuffle (\textsf{VeriShuffle})} 
		\label{algo:malicious-shuffle}
		\begin{algorithmic}[1]
			\Require A length-$N$ vector $\llbracket \boldsymbol{x} \rrbracket$ secret-shared among $\mathcal{S}_1$ and $\mathcal{S}_2$.   
			\Ensure A randomly permuted vector $\llbracket \pi_2(\pi_1(\pi_{12}(\boldsymbol{x}))) \rrbracket$ secret-shared among $\mathcal{S}_1$ and $\mathcal{S}_2$. 
			
			\State // \emph{\underline{Offine phase of the secret-shared shuffle:}}
			\State $\mathcal{S}_{\{1,2,3\}}$ exchange and expand seeds such that $\mathcal{S}_1$ holds $\pi_1, \pi_{12},\boldsymbol{a}'_2,\boldsymbol{b}_2$; $\mathcal{S}_2$ holds $\pi_2, \pi_{12},\boldsymbol{a}_1$;  $\mathcal{S}_3$ holds $\pi_1,\pi_2,\boldsymbol{a}'_2,\boldsymbol{b}_2,\boldsymbol{a}_1$. 
			\State $\mathcal{S}_3$ calculates $\boldsymbol{\Delta} = \pi_2(\pi_1(\boldsymbol{a}_1) + \boldsymbol{a}'_2) - \boldsymbol{b}_2$ and sends $\boldsymbol{\Delta}$ to $\mathcal{S}_2$. 
			\State $\mathcal{S}_{\{1,2\}}$ locally apply $\pi_{12}$ on $\llbracket \boldsymbol{x} \rrbracket$ to obtain $\llbracket \hat{\boldsymbol{x}} \rrbracket$, where $\hat{\boldsymbol{x}}=\pi_{12}(\boldsymbol{x})$. 
			
			\State // \emph{\underline{Online phase of the secret-shared shuffle:}}
			
			\State // \emph{\underline{Step (1) for secret-shared shuffle:}}
			\State $\mathcal{S}_2$ masks its input share $\langle \hat{\boldsymbol{x}}\rangle_2$ using $\boldsymbol{a}_1$ and sends $\boldsymbol{z}_2 \leftarrow \langle \hat{\boldsymbol{x}}\rangle_2 - \boldsymbol{a}_1$ to $\mathcal{S}_1$.
			
			\State // \emph{\underline{Integrity check for $\boldsymbol{z}_2$ sent from $\mathcal{S}_2$:}}
			
			\State $\mathcal{S}_1$ locally computes $\hat{\boldsymbol{x}} - \boldsymbol{a}_1$, where $\hat{\boldsymbol{x}} - \boldsymbol{a}_1 =$ $\boldsymbol{z}_2+\langle\hat{\boldsymbol{x}}\rangle_1$. 
			\State $\mathcal{S}_1$ splits $\hat{\boldsymbol{x}} - \boldsymbol{a}_1$ into two shares and discloses one share to $\mathcal{S}_3$. 
			\State $\mathcal{S}_3$ splits $\boldsymbol{a}_1$ into two shares and discloses one share to $\mathcal{S}_{1}$. 
			\State $\mathcal{S}_{\{1,3\}}$ calculate $\llbracket \hat{\boldsymbol{x}}\rrbracket$ by summing $\llbracket \hat{\boldsymbol{x}}- \boldsymbol{a}_1\rrbracket$ and $\llbracket\boldsymbol{a}_1\rrbracket$. 
			\State $\mathcal{S}_{\{1,3\}}$ calculate the verification tag $f$ following Eq. \ref{eq:verification} and outputs abort if $f\neq0$.  
			
			\State // \emph{\underline{Step (2) for secret-shared shuffle:}}
			\State $\mathcal{S}_1$ sets its output to be $\langle\boldsymbol{y}\rangle_1\leftarrow\boldsymbol{b}_2$ and sends $\boldsymbol{z}_1 \leftarrow \pi_1(\boldsymbol{z}_2+\langle\hat{\boldsymbol{x}}\rangle_1) - \boldsymbol{a}'_2$ to $\mathcal{S}_2$. 
			
			\State // \emph{\underline{Integrity check for $\boldsymbol{z}_1$ sent from $\mathcal{S}_1$:}}
			\State $\mathcal{S}_2$ locally computes $\pi_2(\boldsymbol{z}_1)$. 
			\State $\mathcal{S}_2$ splits $\pi_2(\boldsymbol{z}_1)$ into two shares and discloses one share to $\mathcal{S}_3$.
			\State $\mathcal{S}_3$ splits $\pi_2(\pi_1(\boldsymbol{a}_1)+\boldsymbol{a}'_2)$ into two shares and discloses one share to $\mathcal{S}_{2}$. 
			\State $\mathcal{S}_{\{2,3\}}$ calculate $\llbracket \pi_2(\pi_1(\hat{\boldsymbol{x}}))\rrbracket$ by computing  $\llbracket\pi_2(\boldsymbol{z}_1)\rrbracket + \llbracket\pi_2(\pi_1(\boldsymbol{a}_1)+\boldsymbol{a}'_2)\rrbracket$. 
			\State $\mathcal{S}_{\{2,3\}}$ calculate the verification tag $f$ following Eq. \ref{eq:verification} and outputs abort if $f\neq0$.
			
			\State // \emph{\underline{Step (3) for secret-shared shuffle:}}
			\State $\mathcal{S}_2$ sets its output to be $\langle\boldsymbol{y}\rangle_2\leftarrow\pi_2(\boldsymbol{z}_1) + \boldsymbol{\Delta}$. 
			
			\State // \emph{\underline{Post-shuffle integrity check:}} 
			\State $\mathcal{S}_{\{1,2\}}$ calculate the verification tag $f$ using $\llbracket \pi_2(\pi_1(\pi_{12}(\boldsymbol{x}))) \rrbracket$ following Eq. \ref{eq:verification} and outputs abort if $f\neq0$.
		\end{algorithmic}
	\end{algorithm*}

	\section{The Complete Protocols of Our Maliciously Secure Secret-shared Shuffle and {\main}}
	\label{appdix:complete_protocols}
	The complete protocol of our proposed maliciously secure secret-shared shuffle (denoted as \textsf{VeriShuffle}), integrated with our proposed defense mechanism, is presented in Algorithm \ref{algo:malicious-shuffle}. 
	The complete protocol of {\main} is then given in Algorithm \ref{algo:malicious-SGD-workflow}. 

	\begin{algorithm*}[!t]
		\caption{The Complete Protocol for Our Proposed Communication-Efficient and Maliciously Secure FL in the Shuffle Model of DP (\main)} 
		\label{algo:malicious-SGD-workflow}
		\begin{algorithmic}[1]
			\Require Each client $\mathcal{C}_i$ holds local dataset $\mathcal{D}_i$ for $i\in[n]$. 
			\Ensure$ \mathcal{S}_1,\mathcal{S}_2$ and each client $\mathcal{C}_{i}$ obtain a global model $\theta$ that satisfies $(\varepsilon,\delta)$-DP for $i\in[n]$. 
			
			\State $\mathcal{S}_1$, $\mathcal{S}_2$ initialize: ${\theta}_0 \in \mathcal{G}$. 
			\For {$t\in [T]$}\
			\State // \emph{\underline{Clients:}}
			\For{each client $\mathcal{C}_i$}
			\State \multiline{$\mathcal{C}_i$ chooses uniformly at random a set $\mathcal{U}_{i}$ of $s$ data points.}
			\For{data point $j\in\mathcal{U}_{i}$}
			\State $\mathbf{g}^t_{ij}\leftarrow\nabla_{\theta_t}F(\theta_t, \boldsymbol{x}_{ij})$.
			\State $\tilde{\mathbf{g}}^t_{ij} \leftarrow \mathbf{g}^t_{ij} / \max \left\{1, \frac{\left\|\mathbf{g}^t_{ij}\right\|_2}{L}\right\}$.
			\State \multiline{$\boldsymbol{r}^t_{ij} \leftarrow$ \textsf{NoisyGradCmpr}($\tilde{\mathbf{g}}^t_{ij}$). \Comment \emph{LDP with compression.}}
			\EndFor
			\State \multiline{$\mathcal{C}_i$ splits compressed gradients $\{\boldsymbol{r}^t_{ij}\}_{j\in\mathcal{U}_{i}}$ into two shares among $\mathcal{S}_1$ and $\mathcal{S}_2$. }
			\EndFor
			
			\State // \emph{\underline{Server-side computation:}}

			\State \multiline{$\llbracket\pi(\{\boldsymbol{r}_j\}_{j\in[N]})\rrbracket$ $\leftarrow$ \textsf{VeriShuffle}($\llbracket \{\boldsymbol{r}_j\}_{j\in[N]}\rrbracket$).\Comment \emph{Maliciously secure secret-shared shuffle of $N=ns$ compressed gradients.}}
			\State // \emph{\underline{Integrity check of secret shares:}} \label{code:mali-FL:15}
			\State $\mathcal{S}_{1}$ hashes $\langle \{\boldsymbol{r}_i\}_{i\in [B]} \rangle_1$ to get $h_1$ and $\mathcal{S}_{2}$ hashes $\langle \{\boldsymbol{r}_i\}_{i\in [B]} \rangle_2$ to get $h_2$. 
			\State $\mathcal{S}_{\{1,2\}}$ exchange hashes $h_1,h_2$. 
			\State $\mathcal{S}_{1}$ discloses $\langle \{\boldsymbol{r}_i\}_{i\in [B]} \rangle_1$ to $\mathcal{S}_{2}$, and $\mathcal{S}_{2}$ discloses $\langle \{\boldsymbol{r}_i\}_{i\in [B]} \rangle_2$ to $\mathcal{S}_{1}$. 
			\State $\mathcal{S}_{1}$ hashes $\langle \{\boldsymbol{r}_i\}_{i\in [B]} \rangle_2$ to get $h'_2$ and $\mathcal{S}_{2}$ hashes $\langle \{\boldsymbol{r}_i\}_{i\in [B]} \rangle_1$ to get $h'_1$. 
			\State $\mathcal{S}_{\{1,2\}}$ check if $h'_1=h_1$ and $h'_2=h_2$, and output abort if the hashes do not meet. 
			
			\State // \emph{\underline{MAC check of reconstructed compressed noisy gradients:}}
			\State $\mathcal{S}_{\{1,2\}}$ compute and check if $ t=\boldsymbol{r}\cdot \boldsymbol{k}$ for each compressed noisy gradient, and output abort if any verification fails. 
			
			\State // \emph{\underline{Integrity check of decompression and aggregation results:}}
			\State $\mathcal{S}_{\{1,2\}}$ decompress the compressed noisy gradients and get $\{\mathbf{g}_i\}_{i\in [B]}$.
			
			\State $\mathcal{S}_{\{1,2\}}$ aggregate the decompressed noisy gradients and get $\overline{\mathbf{g}}^t\leftarrow \frac{\sum_{i=1}^B\mathbf{g}_i}{B}$. \label{code:malicious-fl:sample}
			
			\State $\mathcal{S}_{\{1,2\}}$ update the global model $\theta_{t+1} \leftarrow \prod_{\mathcal{G}}\left(\theta_t-\eta_t \overline{\mathbf{g}}^t\right)$. \label{code:malicious-fl:aggregation}
			
			\State $\mathcal{S}_{1}$ hashes $\theta_{t+1}$ to get $h_1$ and $\mathcal{S}_{2}$ hashes $\theta_{t+1}$ to get $h_2$. 
			
			\State $\mathcal{S}_{\{1,2\}}$ exchange the hashes $h_1,h_2$ and output abort if the hashes do not meet, else $\mathcal{S}_{\{1,2\}}$ broadcast $\theta_{t+1}$ to the clients. \label{code:mali-FL:28}
			
			\EndFor
		\end{algorithmic}
	\end{algorithm*}

	\section{Security Proof}
	\label{appdix:security_proof}
	In this section, we provide analysis for {\main} under the threat model defined in Section \ref{sec:threat_model}. 
	Recall that {\main} consists of two stages: local process and server process. 
	The local process is finished by each client without interaction with other clients and the servers. 
	Therefore, we only need to provide security analysis for the server process, which comprises the secret-shared shuffle process and post-shuffle server-side training process, including gradient subsampling, decompression, aggregation (line \ref{code:mali-FL:15} - \ref{code:mali-FL:28}). 
	Specifically, we consider a scenario where a probabilistic polynomial time (PPT) adversary $\mathcal{A}$ statically corrupts at most one of the three servers, allowing the corrupted server, under the control of $\mathcal{A}$, to behave maliciously.
	Below we follow the simulation-based paradigm \cite{Lindell17} to prove the security of our (maliciously secure) secret-shared shuffle protocol, and then prove the security of the server-side post-shuffle training process. 
	We first follow \cite{NDSS22} to define the ideal functionality of a maliciously secret-shared shuffle. 
	
	\begin{defn}
		\label{def:shuffle_functionality}
		\textit{\textbf{(Secret-shared shuffle ideal functionality $\mathcal{F}_{\textsf{shuffle}}$).}} 
		\textit{
			%
			The functionality of $\mathcal{F}_{\textsf{shuffle}}$ also interacts with three servers; at most, one server is controlled by the malicious adversary, and the rest servers are honest. 
			In each round: \\
			\textit{\textbf{Input.}}  $\mathcal{F}_{\textsf{shuffle}}$ receives $N$ gradients from clients. \\
			\textit{\textbf{Computation.}} $\mathcal{F}_{\textsf{shuffle}}$ initiates an empty table $\mathcal{T}$ and fills $\mathcal{T}$ with the received $N$ gradients. 
			Then $\mathcal{F}_{\textsf{shuffle}}$ shuffles $\mathcal{T}$ by sampling a uniformly random permutation $\pi:\mathbb{Z}_N\rightarrow\mathbb{Z}_N$ and applies it to $\mathcal{T}$, resulting $\mathcal{T}'\leftarrow \pi(\mathcal{T})$. 
			Next, $\mathcal{F}_{\textsf{shuffle}}$ sends $\mathcal{T}'$ to the adversary, who can respond with \textsf{continue} or \textsf{abort}. 
			Before responding with \textsf{continue}, the adversary may choose to modify any message $\mathcal{T}'_{\pi(i)}$. 
			\\
			\textit{\textbf{Output.}} 
			Let $\mathcal{T}''$ denote the resulting table with any modifications made by the adversary.
			$\mathcal{F}_{\textsf{shuffle}}$ outputs $\mathcal{T}''$ if the adversary sends \textsf{continue}. \\
		}
	\end{defn}

	We now prove the security of our secret-shared shuffle protocol against a malicious adversary. 
	Prior to simulating servers $\mathcal{S}_{\{1,2,3\}}$, we first presume that the secret-shared multiplication achieves an ideal functionality denoted as $\mathcal{F}_{\textsf{mult}}$, whose messages can be simulated by the simulator $\mathcal{S}_{\textsf{mult}}$. 
	This functionality accepts as inputs shares of two values and yields shares of their product. 
	Moreover, in our proof, we will model hash functions as random oracles when servers exchange hashes of their messages before disclosing them to each other.

	\begin{thm}
		\label{thm:mali-shuffle}
		\textit{\textbf{(Maliciously secure secret-shared shuffle)}}. 
		\textit{
			Assuming that PRG is a random oracle and that the secret-shared multiplication achieves $\mathcal{F}_{\textsf{mult}}$, the secret-shared shuffle protocol (Algorithm \ref{algo:malicious-shuffle}) achieves the functionality (Definition. \ref{def:shuffle_functionality}) in the presence of a malicious adversary. 
		}
	\end{thm}
	
	\begin{proof}
		We first describe the simulation for the servers $\mathcal{S}_{\{1,2\}}$ and then move on to the server $\mathcal{S}_3$. 
		
		\noindent\textbf{Simulator for $\mathcal{S}_1$}:
		We divide the adversary's view into different parts and simulate each of them accordingly. 
		Specifically, we first provide simulators for shuffling and then for the checks associated with shuffling, including our proposed integrity checks during shuffling and post-shuffle blind MAC verification. 
		
		\textit{Shuffling.} 
		In the secure shuffling process, the only information $\mathcal{S}_1$ receives is $\boldsymbol{z}_2 \leftarrow \langle \hat{\boldsymbol{x}}\rangle_2 - \boldsymbol{a}_1$. 
		Since $\boldsymbol{z}_2$ is generated at $\mathcal{S}_2$ by masking $\langle \hat{\boldsymbol{x}}\rangle_2$ with the uniformly random $\boldsymbol{a}_1$, $\boldsymbol{z}_2$ is uniformly random in $\mathcal{S}_1$'s view. 
		The rest of the messages are identical to the information sent by $\mathcal{S}_1$ in the real protocol. 
		%
		
		\textit{Integrity check for $\boldsymbol{z}_2$ sent from $\mathcal{S}_2$.} 
		During the shuffling phase, we provide an integrity check for $\boldsymbol{z}_2$ immediately after $\mathcal{S}_2$ sends it to $\mathcal{S}_1$, as illustrated in Fig. \ref{fig:mac_check_z2} and Algorithm \ref{algo:malicious-shuffle}. 
		Step (1) is a local computation. 
		Step (2) is trivial to simulate since $\mathcal{S}_1$ receives nothing in this step. 
		Step (3) is also trivial to simulate because the secret shares are randomly generated at $\mathcal{S}_3$ and thus are uniformly random in $\mathcal{S}_1$'s view. 
		Step (4) is trivial to simulate because it only requires the basic addition operation in the secret sharing domain. 
		In Step (5), the simulator provides the adversary with random strings to simulate Beaver triples and utilizes simulator $\mathcal{S}_{\textsf{mult}}$ for simulating Beaver multiplications on shares. 
		Finally, when hashing and reconstructing the verification tag $f$, the simulator uses the negation of the adversary's share if the adversary executes the MAC operation honestly. 
		Otherwise, the simulation samples a random value from $\mathbb{Z}_p$. 
		If the adversary deviates from the protocol in any case, the simulation sends an \textsf{abort} signal to the ideal functionality, resulting in the output of $\bot$. 
		Moreover, if the adversary is detected for sending a message whose hash does not match the hash it previously sent, the simulation sends an \textsf{abort} signal to the ideal functionality, also leading to the output of $\bot$.
		
		\textit{Integrity check for $\boldsymbol{z}_1$ sent from $\mathcal{S}_1$.} 
		During the shuffling phase, an integrity check for $\boldsymbol{z}_1$ is also provided after $\mathcal{S}_1$ sends it to $\mathcal{S}_2$, as illustrated in Fig. \ref{fig:mac_check_z1} and Algorithm \ref{algo:malicious-shuffle}. 
		Here, Steps (1) - (5) are trivial to simulate because $\mathcal{S}_1$ is excluded from the integrity check and receives nothing in these steps.

		\textit{Post-shuffle blind MAC verification.} 
		After the shuffling phase, the simulator checks the table $\mathcal{T}'$ received from the ideal functionality, identifying the locations where the messages generated by the adversary have been placed. 
		%
		%
		The simulator provides the adversary with random strings to simulate Beaver triples, employs simulator $\mathcal{S}_{\textsf{mult}}$ to simulate Beaver multiplications for honest shares, and exactly adheres to the protocol for shares of adversary-controlled messages.
		When hashing and reconstructing the verification tag $f$, the simulator uses the negation of the adversary's share if the adversary executes the MAC operation honestly. 
		Otherwise, the simulation samples a random value from $\mathbb{Z}_p$. 
		In any case, where the adversary deviates from the protocol, the simulation sends \textsf{abort} to the ideal functionality, leading to the output of $\bot$. 
		Furthermore, if the adversary is detected for sending a message whose hash does not match the hash it previously sent, the simulation sends \textsf{abort} to the ideal functionality, leading to the output of $\bot$.

		\noindent\textbf{Simulator for $\mathcal{S}_2$}:
		\textit{Shuffling.} 
		To simulate $\mathcal{S}_2$, the simulation must simulate the message $\boldsymbol{z}_1$ sent by $\mathcal{S}_1$. 
		The simulation can solve for the value $\boldsymbol{z}_1=\pi_2^{-1}\left(\langle \hat{\boldsymbol{y}}\rangle_2-\boldsymbol{\Delta}\right)$ given $\mathcal{S}_2$'s input $\pi_2, \boldsymbol{\Delta}$ and output $\langle \hat{\boldsymbol{y}}\rangle_2 \leftarrow \pi_2\left(\boldsymbol{z}_1\right)+\boldsymbol{\Delta}$. 
		The rest of the simulation follows the protocol honestly using $\mathcal{S}_2$'s inputs. 
		This simulation is distributed identically to the view of $\mathcal{S}_2$ in the real protocol because all the inputs,
		outputs and messages are exactly equal to the values that would be sent and received in the real protocol. 
		
		\textit{Integrity check for $\boldsymbol{z}_2$ sent from $\mathcal{S}_2$.} 
		The simulation for $\mathcal{S}_2$ regarding this check is similar with the aforementioned simulation for $\mathcal{S}_1$, because $\mathcal{S}_2$ is excluded from the integrity check and receives nothing in these steps.

		\textit{Integrity check for $\boldsymbol{z}_1$ sent from $\mathcal{S}_1$.} 
		The simulation for $\mathcal{S}_2$ regarding this check also follows the aforementioned simulation for $\mathcal{S}_1$.
		In Steps (1) and (2), $\mathcal{S}_2$ receives nothing. 
		In Step (3), the secret shares are randomly generated at $\mathcal{S}_3$ and thus are uniformly random in $\mathcal{S}_2$'s view. 
		In Step (4), $\mathcal{S}_2$ receives nothing because it only requires the basic addition operation in the secret sharing domain. 
		The simulation of Step (5) is trivial because the simulator provides the adversary with random strings and utilizes simulator $\mathcal{S}_{\textsf{mult}}$ for simulating Beaver multiplications on shares. 
		The simulation for hashing and reconstructing the verification tag is similar to the simulation of $\mathcal{S}_1$ in integrity check for $\boldsymbol{z}_2$ sent from $\mathcal{S}_2$. 
		Therefore, the simulator for $\mathcal{S}_2$ is trivial to construct.

		\textit{Post-shuffle blind MAC verification.} 
		The simulation for $\mathcal{S}_2$ trivially follows the aforementioned simulation for $\mathcal{S}_1$ because $\mathcal{S}_{\{1,2\}}$ are interchangeable in this process.

		\noindent\textbf{Simulator for $\mathcal{S}_3$}: Regarding the simulation of shuffling and post-shuffle blind MAC verification process,the view of $\mathcal{S}_3$ only consists of the random seeds it receives from $\mathcal{S}_{\{1,2\}}$, the messages it sends to $\mathcal{S}_{\{1,2\}}$, and the output of the protocol. 
		The random seeds are simulated by random strings. 
		For an adversary adhering strictly to the protocol, the simulator refrains from sending \textsf{abort} to the ideal functionality, maintains the table $\mathcal{T}'$ unchanged when provided with the opportunity, and obtains the output $\mathcal{T}''$ from the ideal functionality. 
		For an adversary incorrectly sending a malformed shuffle correlation or malformed Beaver triple for the blind MAC verification, the simulation responds with \textsf{abort} to the ideal functionality, resulting in the output $\bot$. 
		Regarding the simulation of integrity check for $\boldsymbol{z}_2$ sent from $\mathcal{S}_2$ and the integrity check for $\boldsymbol{z}_1$ sent from $\mathcal{S}_1$, the secret shares are randomly generated at $\mathcal{S}_{\{1,2\}}$ and thus are uniformly random in $\mathcal{S}_3$'s view at Step (3) of both checks. 
		At Step (4), $\mathcal{S}_3$ receives nothing. 
		At Step (5), the simulation for calculating the verification tag and outputting it is similar to the simulation $\mathcal{S}_{\{1,2\}}$ because $\mathcal{S}_3$ is interchangeable with the other server at this step. 
		%
	\end{proof}

	Then we demonstrate that the servers $\mathcal{S}_{\{1,2\}}$ can detect occurrences of deviation behavior following a maliciously secure secret-shared shuffle, particularly during the subsampling, decompression, and aggregation processes. 
	Recall that in our maliciously secure {\main}, $\mathcal{S}_{3}$ stays offline after the secure shuffling process at each round, therefore we do not need to consider the behavior of $\mathcal{S}_{3}$ in the following analysis. 
	
	\begin{thm}
		After the secret-shared shuffle, during the subsequent subsampling, decompression, and aggregation operations (collaboratively conducted by $\mathcal{S}_{\{1,2\}}$), a malicious server can be detected by the other honest server during the verification process. 
	\end{thm}
	
	\begin{proof}
		Recalling from Section \ref{sec:threat_model}, we operate within a three-server honest-majority setting, indicating that, at most, one server may be malicious in our scenario. 
		Before reconstructing the first $B$ shuffled gradients, we let $\mathcal{S}_{\{1,2\}}$ exchange hashes of their shares of the first $B$ shuffled gradients. 
		Therefore, if a server detects that the share it receives does not match the hash, the misbehavior of another server will be detected, and the protocol will abort. 
		Next, to detect if a malicious server tamper with secret shares, $\mathcal{S}_{\{1,2\}}$ both check each (compressed) gradient $\boldsymbol{r}_i$'s MAC by verifying if $\sum_{j=1}^{\ell} \boldsymbol{k}_i[j] \cdot \boldsymbol{r}_i[j] - t_i$ equals zero. 
		If the verification succeeds, it ensures the security of the subsampling process. 
		
		The integrity check for gradient decompression and aggregation is more facilitated because we only need to check the consistency of the output model at $\mathcal{S}_{\{1,2\}}$. 
		The reason we can check in this way lies in that $\mathcal{S}_{\{1,2\}}$, if both honest, will get the same decompressed gradients $\{\mathbf{g}_{i}\}_{i\in[B]}$ due to the security of PRG, and resulting in the same output model $\theta_{t+1}$. 
		Therefore, if the hashes do not match, the protocol outputs abort. 
		The hashes possess high entropy, ensuring that a malicious server incorrectly performing gradient decompression and aggregation gains no information from the hash of the other server.

	\end{proof}

        \section{Proof of Theorem \ref{thm:analysis_overview}}
	\label{appendix:analysis_overview}
	\subsubsection{Proof of Privacy}
	Since {\main} consists of a sequence of $T$ adaptive mechanisms $\mathcal{M}_1,\cdots,\mathcal{M}_T$, to analyze the total privacy guarantee, we first focus on the RDP of each $\mathcal{M}_t$, where $t\in[T]$. 
	
	\begin{lem}
		\label{lem:privacy_amplification}
		For subsampling rate $\gamma = \frac{B}{M}$ and \textsf{NoisyGradCmpr} that satisfies $\varepsilon_0$-LDP, where $\varepsilon_0 \leq \frac{1}{2}\log (\gamma M / \log (1 / \tilde{\delta}))$ and $\tilde{\delta}\in(0,1)$. 
		For a single iteration $t\in[T]$, the mechanism that composes shuffling and subsampling is ($\lambda,\frac{\lambda \log^2(1+\gamma(e^{\tilde{\varepsilon}} - 1)) }{2}$)-RDP, where $\tilde{\varepsilon}=\mathcal{O}\left(\min \left\{\varepsilon_0, 1\right\} e^{\varepsilon_0} \sqrt{{\log (1 / \tilde{\delta})}/{\gamma M}}\right)$.
	\end{lem}

	\begin{proof}
		To analyze the RDP of a single iteration, we first analyze the DP of an iteration and then follow \cite{ErlingssonFMRTT19} to get RDP using the fact from \cite{BunS16}. 
		From the privacy amplification by shuffling result \cite{BalleBGN19}, the shuffling mechanism enables ($\tilde{\varepsilon},\tilde{\delta}$)-DP for $\tilde{\delta}\in(0,1)$ and $\varepsilon_0 \leq \frac{1}{2}\log (\gamma M / \log (1 / \tilde{\delta}))$, we have 
		\begin{equation}
			\label{eq:tilde_eps}
			\tilde{\varepsilon}=\mathcal{O}\left(\min \left\{\varepsilon_0, 1\right\} e^{\varepsilon_0} \sqrt{{\log (1 / \tilde{\delta})}/{\gamma M}}\right).
		\end{equation}

		From the privacy amplification by subsampling result \cite{BalleBG18}, for a mechanism $\mathcal{M}$ satisfying (${\varepsilon},{\delta}$)-DP, composing the subsampled mechanism with $\mathcal{M}$ enables (${\varepsilon}',{\delta}'$)-DP with ${\varepsilon}' = \log(1+\gamma(e^{\varepsilon}-1))$ and ${\delta}'=\gamma\delta$. 
		Following the results from \cite{AISTATS21}, we utilize privacy amplification by shuffling, and then compose the amplified results with privacy amplification by subsampling. 
		In this way, the mechanism that composes shuffling and subsampling is $(\log(1+\gamma(e^{\tilde{\varepsilon}} - 1)),\gamma\tilde{\delta})$-DP, where $\tilde{\varepsilon}$ follows Eq. \ref{eq:tilde_eps}, subsampling rate $\gamma = \frac{ks}{nr}$, $\varepsilon_0 \leq \frac{1}{2}\log (\gamma M / \log (1 / \tilde{\delta}))$ and $\tilde{\delta}\in(0,1)$.
		
		Next, we follow \cite{ErlingssonFMRTT19} to get RDP using the fact from \cite{BunS16} and get the RDP of each iteration, i.e., ($\lambda,\frac{\lambda \log^2(1+\gamma(e^{\tilde{\varepsilon}} - 1)) }{2}$)-RDP. 
	\end{proof}
	
	Based on the result of Lemma \ref{lem:privacy_amplification}, we can obtain the RDP of $T$ iterations by using Lemma \ref{lem:sequantial}. 
	Next, using Lemma \ref{lem:rdp_to_dp}, we can convert RDP to DP and finally arrive at Theorem \ref{thm:analysis_overview}.

	\subsubsection{Proof of Communication}
	The naive LDP mechanism DJW18 \cite{duchi2018minimax} requires $de$ bits to represent a $d$-dimensional vector $\boldsymbol{v}$, where $e$ denotes the number of bits to represent a dimension. 
	This mechanism, if directly adopted for the secret-shared shuffle of $N$ noisy gradients, incurs a communication complexity of $\mathcal{O}(Nd)$ of both client-server communication and inter-server communication. 
	The client-server communication cost is $\mathcal{O}(Nd)$ because $N$ gradients are secret-shared among the servers. 
	The inter-server communication cost is also $\mathcal{O}(Nd)$. 
	This is because the communication cost of the secret-shared shuffle protocol we adopt from \cite{NDSS22} is $\mathcal{O}(Nl)$,where the message length $l$ is associated with $de$ in the naive DJW18 mechanism. 
	
	In our proposed \textsf{NoisyGradCmpr} mechanism, we reduce this communication cost by using PRGs to losslessly compress the perturbed gradient vectors. 
	The idea is as follows. 
	Since the output perturbed vector in the naive DJW18 is generated using a sign bit $sgn$ and a $d$-dimensional random vector $\boldsymbol{v}$, we can use PRG to compress $\boldsymbol{v}$ by inputting a random seed $s$ to PRG. 
	Therefore, we only need to transmit \textit{fixed-length} compressed vector, comprising a sign bit $sgn$ and a random seed $s$. 
	With the same PRG to recover $\boldsymbol{v}$, servers could obtain the same output perturbed vector as DJW18. 
	In this way, our compression method significantly reduces the client-server and inter-server communication complexity from $\mathcal{O}(Nd)$ to $\mathcal{O}(N)$. 
	
	\subsubsection{Proof of Convergence}
	At iteration $t\in[T]$ of {\main}, the $ks$ received compressed and perturbed gradients are averaged as $\overline{\mathbf{g}}^t\leftarrow \frac{\sum_{i=1}^B\mathbf{g}_i}{B}$, and the global model is updated as $\theta_{t+1} \leftarrow \prod_{\mathcal{G}}\left(\theta_t-\eta_t \overline{\mathbf{g}}^t\right)$. 
	Since our proposed noisy gradient compression mechanism \textsf{NoisyGradCmpr}, used in couple with the decompression mechanism \textsf{NoisyGradDcmp}, is unbiased, the average gradient $\overline{\mathbf{g}}^t$ is also unbiased, i.e., we have $\mathbb{E}\left[\overline{\mathbf{g}}^t\right]=\nabla_{\theta_t} F\left(\theta_t\right)$. 
	Now we demonstrate that the second moment of the $\overline{\mathbf{g}}^t$ is bounded: 
	\begin{equation}
		\label{eq:second_moment}
		\begin{aligned}
			\mathbb{E}\left\|\overline{\mathbf{g}}^t\right\|_2^2 & =\left\|\mathbb{E}\left[\overline{\mathbf{g}}^t\right]\right\|_2^2+\mathbb{E}\left\|\overline{\mathbf{g}}^t-\mathbb{E}\left[\overline{\mathbf{g}}^t\right]\right\|_2^2 \\
			& \stackrel{(a)}{\leq}  L^2 + \mathbb{E}\left\|\overline{\mathbf{g}}^t-\mathbb{E}\left[\overline{\mathbf{g}}^t\right]\right\|_2^2 \\
			& \stackrel{(b)}{\leq}  L^2+\frac{14 L^2  d}{k s}\left(\frac{e^{\varepsilon_0}+1}{e^{\varepsilon_0}-1}\right)^2 \\
			& \stackrel{(c)}{=}  L^2+\frac{14 L^2  d}{\gamma M}\left(\frac{e^{\varepsilon_0}+1}{e^{\varepsilon_0}-1}\right)^2,
		\end{aligned}
	\end{equation}

	\noindent Step (a) follows from the fact that $\ell(\theta,\boldsymbol{x})\leq L$ \cite{Shwartz12}, which implies $\left\|\nabla_{\theta_t} F\left(\theta_t\right)\right\| \leq L$. 
	Step (b) follows from \cite[Corollary 1]{AISTATS21}. 
	Step (c) uses $\gamma = \frac{ks}{M}$. 
	For simplicity, by letting $G=L\sqrt{1+\frac{14d}{\gamma M}(\frac{e^{\varepsilon_0}+1}{e^{\varepsilon_0}-1})^2}$, we have $\mathbb{E}\left\|\overline{\mathbf{g}}^t\right\|_2^2 \leq G^2$.
	
	Then the standard bound on the convergence of SGD for convex functions from \cite{Shamir013} could be used, which is given as follows: 
	\begin{lem}
		\label{lem:sgd}
		\textit{Let $F(\theta)$ be a convex function, and the set $\mathcal{G}$ has diameter $D$. Consider a SGD algorithm $\theta_{t+1} \leftarrow \prod_{\mathcal{G}}\left(\theta_t-\eta_t {\mathbf{g}}^t\right)$, where ${\mathbf{g}}^t$ satisfies $\mathbb{E}\left[\mathbf{g}^t\right]=\nabla_{\theta_t} F\left(\theta_t\right)$ and $\mathbb{E}\left\|\mathbf{g}^t\right\|_2^2 \leq G^2$. 
			By setting $\theta^*=\mathop{\arg\min}_{\theta \in \mathcal{G}} F(\theta)$ and $\eta_t = \frac{D}{G\sqrt{t}}$, we get: }
		$$
		\begin{aligned}
			\mathbb{E}\left[F\left(\theta_t\right)\right]-F(\theta^*) & \leq 2 D G \frac{2+\log (T)}{\sqrt{T}} \\
			& =\mathcal{O}\left(D G \frac{\log (T)}{\sqrt{T}}\right). 
		\end{aligned}
		$$
	\end{lem}
	
	Since {\main} satisfies the assumptions of Lemma \ref{lem:sgd}, the result of $\mathbb{E}\left\|\mathbf{g}^t\right\|_2^2 \leq G^2$ in Eq. \ref{eq:second_moment} could be used to guarantee that the output $\theta_T$ of {\main} satisfies: 
	$$
	\begin{aligned}
		&\mathbb{E}\left[F\left(\theta_t\right)\right]-F(\theta^*) \leq
		\\
		&\mathcal{O}\left(D L \left(1+\sqrt{\frac{14 d}{\gamma M}}\left(\frac{e^{\varepsilon_0}+1}{e^{\varepsilon_0}-1}\right)\right) \frac{ \log (T) }{\sqrt{T}}\right),
	\end{aligned}
	$$
	
	\noindent where we follow \cite{AISTATS21} to use the inequality $\sqrt{1+\frac{14 d}{\gamma M}(\frac{e^{\varepsilon_0}+1}{e^{\varepsilon_0}-1})^2} \leq(1+\sqrt{\frac{14 d}{\gamma M}}(\frac{e^{\varepsilon_0}+1}{e^{\varepsilon_0}-1}))$. 
	
	This completes the proof of the third part of Theorem \ref{thm:analysis_overview}.

\end{document}